\documentclass[reqno]{amsart}
\usepackage[a4paper]{geometry}
\usepackage[T1]{fontenc}
\usepackage{amssymb}
\usepackage[mathscr]{eucal}
\usepackage{mathtools}
\usepackage{microtype}
\usepackage{enumerate,paralist}
\usepackage{hyperref}
\usepackage[initials]{amsrefs}
\numberwithin{equation}{section}
\theoremstyle{plain}
\newtheorem{theorem}{Theorem}[section]
\newtheorem{proposition}[theorem]{Proposition}
\newtheorem{lemma}[theorem]{Lemma}
\newtheorem*{statement*}{Statement}
\theoremstyle{remark}

\newtheorem*{remark*}{Remark}
\newtheorem*{remarks*}{Remarks}
\newtheorem*{stepone*}{Step 1}
\newtheorem*{steptwo*}{Step 2}
\newtheorem*{stepthree*}{Step 3}
\newtheorem*{stepfour*}{Step 4}
\newtheorem*{stepfive*}{Step 5}
\newtheorem*{stepsix*}{Step 6}
\newtheorem*{summary*}{Summary}
\newtheorem*{warning*}{Warning}
\theoremstyle{definition}

\newtheorem*{assumption*}{Assumption}
\newtheorem*{assumption-0*}{Assumption (H0)}
\newtheorem*{assumption-1*}{Assumption (H1)}
\newtheorem*{assumption-2*}{Assumption (H2)}
\newtheorem*{definition*}{Definition}
\newtheorem*{notation*}{Notation}
\newtheorem*{notations*}{Notations}

\providecommand{\BS}[1]{\boldsymbol{#1}}
\providecommand{\C}[1]{\mathcal{#1}}
\providecommand{\D}[1]{\mathbb{#1}}
\providecommand{\E}[1]{\mathscr{#1}}
\providecommand{\F}[1]{\mathfrak{#1}}
\providecommand{\R}[1]{\mathrm{#1}}
\providecommand{\abs}[1]{\lvert#1\rvert}
\providecommand{\accol}[1]{\lbrace#1\rbrace}
\providecommand{\av}[1]{\langle#1\rangle}
\providecommand{\croch}[1]{\lbrack#1\rbrack}
\providecommand{\norm}[1]{\lVert#1\rVert}
\providecommand{\scal}[1]{\langle#1\rangle}
\renewcommand{\Im}{\operatorname{Im}}
\renewcommand{\Re}{\operatorname{Re}}
\newcommand{\cercle}{\mathbb{T}}
\newcommand{\dd}{\mathrm{d}}
\newcommand{\eul}{\mathrm{e}}
\newcommand{\ii}{\mathrm{i}}
\newcommand{\vece}{\mathrm{e}}

\newcommand{\spec}{\sigma} 
\newcommand{\uw}{\underline\omega} 
\newcommand\ut{\underline t} 
\newcommand{\ek}{\mathrm{e}^{\ii\xi}} 
\newcommand{\nk}{\mathrm{e}^{\ii\eta}}   
\DeclareMathOperator{\card}{card}
\DeclareMathOperator{\class}{C}
\DeclareMathOperator{\dist}{dist}

\DeclareMathOperator{\ord}{O}
\DeclareMathOperator{\osmall}{o}
\DeclareMathOperator{\supp}{supp}
\DeclareMathOperator{\tr}{tr}

\newcommand{\diag}{\mathrm{diag}}
\newcommand{\field}{\mathrm{field}}

\newcommand{\rabi}{\mathrm{Rabi}}
\begin{document} 
\title[Oscillatory behavior of large eigenvalues in quantum Rabi models]{Oscillatory behavior of large eigenvalues\\
in quantum Rabi models}
\author[A. Boutet de Monvel]{Anne Boutet de Monvel$^1$}
\author[L. Zielinski]{Lech Zielinski$^2$}
\address{$^1$Institut de Math\'ematiques de Jussieu-PRG, Universit\'e Paris Diderot\\
b\^atiment Sophie Germain, case 7012, 75205 Paris Cedex 13, France\\
E-mail: anne.boutet-de-monvel@imj-prg.fr}
\address{$^2$Universit\'e du Littoral C\^ote d'Opale\\
       Laboratoire de Math\'ematiques Pures et Appliqu\'ees\\
       62100 Calais, France\\
E-mail: Lech.Zielinski@lmpa.univ-littoral.fr}
\begin{abstract} 
We investigate the large $n$ asymptotics of the $n$-th eigenvalue for a class of unbounded self-adjoint operators defined by infinite Jacobi matrices with discrete spectrum. In the case of the quantum Rabi model we obtain the first three terms of the asymptotics which determine the parameters of the model. This paper is based on our previous paper \cite{BZ5} that it completes and improves.
\end{abstract}
\maketitle 
\section{Introduction}  \label{sec:intro}
\subsection{The quantum Rabi model}  \label{sec:JC}

This paper is motivated by the quantum Rabi model \cite{Bra16} describing the simplest interaction between radiation and matter (see \cite{Scu}). It is also called the Jaynes--Cummings model without the rotating-wave approximation. The Hamiltonian of this model is a self-adjoint operator $H_{\rabi}$ with discrete spectrum defined in Appendix. We refer to \cite{X} for a list of recent works on this model. 

A natural question is whether the spectrum of $H_{\rabi}$ determines the values of all the parameters involved in its definition. These parameters are listed in the Appendix and denoted by $\omega$, $E$, and $g$. The answer to this question is positive. In \cite{BZ6} we explain how to recover the values of all the parameters from the spectrum of $H_{\rabi}$. The method described in \cite{BZ6} is based on the three term asymptotics for large eigenvalues of $H_{\rabi}$. This asymptotic formula can be easily deduced from the formula \eqref{E0} of this paper using the well known fact (see \cites{Tur1,Tur2,Y}) that $H_{\rabi}$ can be written as the direct sum of two Jacobi matrices (see Appendix) to which Theorem~\ref{thm:11} applies. 

Further on $J$ denotes an infinite real Jacobi matrix  
\begin{equation}\label{J} 
J=\begin{pmatrix}  
d(1) & a(1) &0& 0 &\dots\\ 
a(1) & d(2) & a(2) &0& \dots\\ 
0 & a(2) & d(3) & a(3) &\dots\\ 
0 &0& a(3) & d(4) &\dots\\  
\vdots &\vdots&\vdots&\vdots&\ddots 
\end{pmatrix}.
\end{equation}  
To treat the Jacobi matrices representing the quantum Rabi model we have to consider entries $\accol{d(k)}_{k=1}^{\infty}$ and $\accol{a(k)}_{k=1}^{\infty}$ of the form
\begin{equation}\label{J'}
\begin{cases}    
d(k)=k+(-1)^k\rho,&\\  
a(k)=a_1k^{1/2},&   
\end{cases} 
\end{equation}
where $\rho\in\D{R}$ and $a_1>0$ are some constants.  

The Jacobi matrix \eqref{J} defines the self-adjoint operator $J$ that acts on $l^2(\D{N}^*)$ by
\begin{equation}\label{J''}
(Jx)(k)=d(k)x(k)+a(k)x(k+1)+a(k-1)x(k-1) 
\end{equation} 
where, by convention, $x(0)=0$ and $a(0)=0$. We denote by $\D{N}^*$ the set of positive integers and by $l^2(\D{N}^*)$ the Hilbert space of square summable complex sequences $x=(x(k))_{k=1}^{\infty}$ equipped with the scalar product $\scal{x,y}=\sum_{k=1}^{\infty}\overline {x(k)}y(k)$. The operator $J$ is defined on $\C{D}\coloneqq\bigl\lbrace x\in l^2(\D{N}^*):\sum_{k=1}^{\infty}d(k)^2\abs{x(k)}^2<\infty\bigr\rbrace$.  

Under our assumptions, in particular under \eqref{J'}, the self-adjoint operator $J$ is bounded from below with compact resolvent. Its spectrum is therefore discrete and one can find an orthogonal basis $\accol{w_n}_{n=1}^{\infty}$ such that $Jw_n=\lambda_n(J)w_n$ where $\accol{\lambda_n(J)}_{n=1}^{\infty}$ is the non-decreasing sequence of its eigenvalues:
\[
\lambda_1(J)\leq\dots\leq\lambda_n(J)\leq\lambda_{n+1}(J)\leq\dots
\]
The aim of this paper is to describe the asymptotic behavior of $\lambda_n(J)$ when $n\to\infty$.

\begin{theorem}[Quantum Rabi model]\label{thm:11} 
Let $J$ be defined by \eqref{J''} with $\accol{a(k)}_{k=1}^{\infty}$ and $\accol{d(k)}_{k=1}^{\infty}$ given by \eqref{J'}. Let $\lambda_n(J)$ denote the $n$-th eigenvalue of $J$. Then, for any $\varepsilon>0$ we have the large $n$ asymptotic formula
\begin{subequations}\label{ER0}
\begin{equation}\label{E0}
\lambda_n(J)=n-a_1^2+\F{r}(n)+\ord(n^{-1/2+\varepsilon}),
\end{equation}
where
\begin{equation}\label{R0} 
\F{r}(n)=(-1)^n\rho\,\frac{\cos\bigl(4a_1\sqrt{n}-\tfrac{\pi}{4}\bigr)}{\sqrt{2\pi a_1}}\,n^{-1/4}.
\end{equation}
\end{subequations}  
\end{theorem} 
 
\subsection{Comments} 
 
In this section $J$ denotes the Jacobi operator defined by \eqref{J''}-\eqref{J'}.
\subsubsection{}  

The three-term asymptotics \eqref{ER0} improves the two-term asymptotics proved by Yanovich \cite{Y} (see also an earlier version in \cite{Tur2}):
\begin{equation}\label{Y0}
\lambda_n(J)=n-a_1^2+\ord(n^{-1/16}).
\end{equation}
\subsubsection{}  

The large $n$ behavior of $\lambda_n(J)$ was already investigated by Schmutz \cite{Schm}. Let $J_0$ denote the operator defined by \eqref{J''}-\eqref{J'} when $\rho=0$, i.e.,
\begin{equation}  \label{J0}
(J_0x)(k)=kx(k)+a_1\sqrt{k}\,x(k+1)+a_1\sqrt{k-1}\,x(k-1).
\end{equation}
As noticed by Schmutz, $J_0$ can be diagonalized in the canonical basis by means of the Bogoliubov transformation: 
\[
L_0\coloneqq\eul^{\ii B}J_0\eul^{-\ii B}=\diag(k-a_1^2)_{k=1}^\infty
\]
defined by $B=a_1(\hat a-\hat a^\dagger)$, where $\hat a$ and $\hat a^\dagger$ are the annihilation and creation operators, respectively (see Appendix~\ref{appendix}). Since $J=J_0+V$ with $V=\diag\left((-1)^k\rho\right)$, its eigenvalues are the same as those of
\[
L\coloneqq\eul^{\ii B}J\eul^{-\ii B}=L_0+\tilde V,
\]
where $\tilde V\coloneqq\eul^{\ii B}V\eul^{-\ii B}$. Then, decay properties of the entries of $\tilde V$ should allow to expect the large $n$ behavior
\begin{equation}  \label{schmutz}
\lambda_n(J)=\lambda_n(L)\approx\lambda_n(L_0)=n-a_1^2.
\end{equation}
\subsubsection{}  

The asymptotic behavior of the matrix elements $\tilde V(j,k)$ expressed by means of Laguerre polynomials was considered by I.~D.~Feranchuk et al \cite{FKU} who proposed to apply the $0$th order perturbation theory (see the book \cite{FILU}) and suggested the following improvement of \eqref{schmutz}:
\begin{equation} \label{GRWA}
\lambda_n(J)=\lambda_n(L)\approx\lambda_n(L_0)+\tilde V(n,n). 
\end{equation} 
The approximation \eqref{GRWA} was discovered independently by Irish~\cite{Ir}. Following Irish, \eqref{GRWA} is called the Generalized Rotating-Wave Approximation (GRWA) in the physical literature. We observe that one can prove that
\begin{equation} \label{diag-appr}
\tilde V(n,n)=\F{r}(n)+\ord(n^{-1/2+\varepsilon}),
\end{equation} 
where $\F{r}(n)$ is given by \eqref{R0}.
\subsubsection{}  

The first step of our analysis uses an analog of the Bogoliubov transformation. In order to simplify the remainder estimates with respect to the large parameter $n$ we work with Jacobi operators indexed by $n$ and write 
\[
\eul^{\ii B_n}J_n\eul^{-\ii B_n}=L_{0,n}+\tilde V_n+R_n 
\]
where $L_{0,n}$ is diagonal, $\norm{R_n}=\ord(n^{-1/2})$, and
\begin{equation} \label{diag-approx}
\tilde V_n(n,n)=\F{r}(n)+\ord(n^{-1/2+\varepsilon}).
\end{equation} 
The definitions of $J_n$, $L_{0,n}$, and $\tilde V_n$ are given in Section~\ref{sec:22}. Propositions~\ref{prop:21} and \ref{prop:23} ensure the fact that the asymptotic formula for $J$ can be reduced to an analogous formula for $J_n$ and $L_n=L_{0,n}+\tilde V_n$.
\subsubsection{}  

The asymptotic behavior of $\lambda_n(J_n)$ is deduced from the trace estimate described in Section~\ref{sec:23} by means of a Tauberian type result \cite{BZ5}*{Proposition 11.1} slightly adapted in Proposition~\ref{prop:25}. We begin the proof of the trace estimate in Section~\ref{sec:4'} by reducing the problem to large $n$ estimates of a Dyson expansion similarly as in \cite{BZ5}. We notice that Section~\ref{sec:4'} is the only part of this paper where we rewrite proofs from \cite{BZ5} in a slightly more general form. We complete the proof by an analysis of the Dyson expansion in Sections~\ref{sec:5}-\ref{sec:7} and to perform this analysis we need to use a certain number of auxiliary results from \cite{BZ5}*{Section 10}. In order to avoid unnecessary overlaps we refer to \cite{BZ5} for the proofs of these auxiliary results. 
\subsubsection{}  

In Section~\ref{sec:xJC} we describe a class of more general type of operators for which we can obtain an analogous large $n$ asymptotic formula. Following \cite{BNS} we replace the sequence $\rho (-1)^k$ by a general sequence of period $N$ and we give the corresponding asymptotic formula in Theorem~\ref{thm:12}. Since in practice the proofs of Theorem~\ref{thm:11} and Theorem~\ref{thm:12} require the same arguments, we chose to present the proof in the more general framework, i.e., for the class of operators described in Section~\ref{sec:xJC}. For readers interested only in the result of Theorem~\ref{thm:11} we indicate that the only simplification with respect to Theorem~\ref{thm:12} consists in the fact that the proof of \eqref{diag-appr} is simpler in the case of period $N=2$. Indeed, an additional symmetry of this case allows us to express an approximation of $\tilde V_n(k,k)$ by oscillating integrals with very simple phase functions (see Section~\ref{sec:42}) and to obtain \eqref{diag-appr} immediately from the stationary phase formula. Thus the proof of Theorem~\ref{thm:11} ends in Section~\ref{sec:72}. In the case of Theorem~\ref{thm:12} the proof of \eqref{diag-appr} involves more complicated phase functions and is given in Section~\ref{sec:73}.
\subsubsection{}  

Our approach works the same way in the proofs of Theorems~\ref{thm:11} and \ref{thm:12}. Therefore, it does not distinguish whether or not the corresponding model is integrable in the sense of Braak \cite{Bra11}. For this reason, it makes no contribution to the Braak conjecture.
 
\subsection{Quantum Rabi type models} \label{sec:xJC}

We consider the following assumptions on the entries of $J$:   

\begin{assumption-1*}
There exist constants $0<\gamma\leq\frac{1}{2}$, $C$, $C'$, $C''$, and $c>0$ such that
\begin{align*}
&ck^{\gamma}\leq a(k)\leq Ck^{\gamma},\tag{H1a}\\
&\abs{\delta a(k)}\leq C'k^{\gamma-1},\tag{H1b}\\
&\abs{\delta^2 a(k)}\leq C''k^{\gamma-2}\tag{H1c}
\end{align*}
for any $k\in\D{N}^*$. Here, $\delta a(k)\coloneqq a(k+1)-a(k)$ and $\delta^2a(k)\coloneqq a(k+2)-2a(k+1)+a(k)$.
\end{assumption-1*}

\begin{remark*}
(H1) is satisfied if $a(k)$ has the large $k$ behavior
\[  
a(k)=a_1k^{\gamma}+a_1'k^{\gamma-1}+\ord(k^{\gamma-2}).
\]
\end{remark*}

\begin{assumption-2*}
The diagonal entries of $J$ are of the form
\begin{equation}\label{dk}
d(k)=k+v(k)\tag{H2a}
\end{equation}
where $v\colon\D{N}^*\to\D{R}$ is real-valued and periodic of period $N$, i.e.: 
\begin{equation}\label{period}
v(k+N)=v(k)\text{ for any }k\in\D{N}^*.\tag{H2b}
\end{equation}
Moreover, we assume
\begin{equation}\label{rho-N}
\rho_N<
\begin{cases} 
\frac{1}{2}&\text{if }N=2,\\
\frac{1}{\pi\sqrt{N}}&\text{if }N\geq 3,
\end{cases}
\tag{H2c}
\end{equation}
where
\begin{equation}  \label{rhoN} 
\rho_N=\rho_N(v)\coloneqq\max_{1\leq k\leq N}\abs{v(k)-\av{v}}.
\end{equation}
Here $\av{v}\coloneqq\frac{1}{N}\sum_{1\leq k\leq N}v(k)$ denotes the ``mean value'' of $v$.
\end{assumption-2*}

To compare with the hypotheses of Theorem~\ref{thm:11} we reformulate these as follows:

\begin{assumption-0*}
The diagonal and off-diagonal entries of $J$ are of the form
\[
\begin{cases}
d(k)=k+v(k),&\text{with }v(k)=(-1)^k\rho,\\
a(k)=a_1k^{\gamma},&\text{with }\gamma=\frac{1}{2},
\end{cases}
\]
where $\rho$ is a real constant. In particular, $v$ is periodic of period $N=2$, $\av{v}=0$, and $\rho_N=\abs{\rho}$.
\end{assumption-0*}

\begin{remark*}
(H0) is a special case of (H1) and (H2), except that there is no restriction on $\rho_N=\abs{\rho}$.
\end{remark*}

Let $v$ be as in (H2). By $N$-periodicity we can expand it as follows: 
\begin{equation}\label{vk} 
v(k)=\alpha_0+\sum_{m=1}^{\lfloor N/2\rfloor}\alpha_m\cos\tfrac{2m\pi k}{N}+\sum_{m=1}^{\lfloor(N-1)/2\rfloor}\tilde\alpha_m\sin\tfrac{2m\pi k}{N}, 
\end{equation} 
where 
\begin{itemize}
\item
$\lfloor s\rfloor\coloneqq\max\accol{k\in\D{Z}:k\leq s}$ is the integer part of $s$, 
\item
all coefficients $\alpha_0=\av{v}$, $\alpha_m$, and $\tilde\alpha_m$ are real constants.  
\end{itemize}

\begin{theorem}[Quantum Rabi type model]\label{thm:12} 
Let $J$ be defined by \eqref{J''} with $\accol{a(k)}_{k=1}^{\infty}$ and $\accol{d(k)}_{k=1}^{\infty}$ satisfying assumptions \emph{(H1)} and \emph{(H2)}, respectively. Then, for any $\varepsilon>0$ we have the large $n$ asymptotic formula:
\begin{subequations}\label{ER2}
\begin{equation}\label{E2}
\lambda_n(J)=n+a(n-1)^2-a(n)^2+\alpha_0+\F{r}(n)+\ord(n^{-\gamma+\varepsilon}).
\end{equation} 
where $\alpha_0=\av{v}$ and
\begin{equation}\label{R2}
\F{r}(n)=\sum_{m=1}^{\lfloor N/2\rfloor}\alpha_m\F{r}_m(n)+\sum_{m=1}^{\lfloor(N-1)/2\rfloor}\tilde\alpha_m\tilde{\F{r}}_m(n),
\end{equation} 
with $\alpha_m,\tilde\alpha_m$ as in \eqref{vk} and $\F{r}_m(n),\tilde{\F{r}}_m(n)$ defined by 
\begin{align}\label{Rm}
\F{r}_m(n)\coloneqq\frac{\cos\bigl(4a(n)\sin\tfrac{m\pi}{N}-\tfrac{\pi}{4}\bigr)}{\sqrt{2\pi a(n)\sin\tfrac{m\pi}{N}}}\cos\Bigl(\tfrac{2m\pi n}{N}+2a(n)\delta a(n)\sin\tfrac{2m\pi}{N}\Bigr),\\   
\label{tRm}
\tilde{\F{r}}_m(n)\coloneqq\frac{\cos\bigl(4a(n)\sin\tfrac{m\pi}{N}-\tfrac{\pi}{4}\bigr)}{\sqrt{2\pi a(n)\sin\tfrac{m\pi}{N}}}\sin\Bigl(\tfrac{2m\pi n}{N}+2a(n)\delta a(n)\sin\tfrac{2m\pi}{N}\Bigr).
\end{align}
\end{subequations}  
\end{theorem}  

\begin{remark*}
For $N=2$, the expression of $\F{r}(n)$ reduces to
\[
\F{r}(n)=\rho\,\F{r}_1(n)=(-1)^n\rho\frac{\cos(4a(n)-\tfrac{\pi}{4})}{\sqrt{2\pi a(n)}}\,.
\]
Moreover, in the case of the quantum Rabi model, $a(n)=a_1\sqrt{n}$, hence $a(n-1)^2-a(n)^2=-a_1^2$. Then, \eqref{ER2} becomes the asymptotic formula \eqref{ER0} in Theorem~\ref{thm:11}.
\end{remark*}

\begin{assumption*}
Further on, we make the assumption 
\begin{equation}\label{alpha0}
\av{v}=\alpha_0=0. 
\end{equation} 
Indeed, since $\lambda_n(J)=\alpha_0+\lambda_n(J-\alpha_0I)$ it suffices to prove Theorem~\ref{thm:12} for $J-\alpha_0I$.
\end{assumption*}   

\subsection{Plan of the paper}   \label{sec:14}

As in \cite{BZ5} the main ingredient of our approach is a trace estimate (Proposition \ref{prop:trace}). In Section \ref{sec:2} we show the implication
\begin{equation}\label{144}
\begin{rcases} 
\text{Propositions~\ref{prop:21} \& ~\ref{prop:23}}\\
\text{Propositions~\ref{prop:trace} \& ~\ref{prop:25}}\\
\text{Lemmas~\ref{lem:22} \& ~\ref{lem:26}}
\end{rcases}\!\implies\!\text{Theorems~\ref{thm:11} \& ~\ref{thm:12}}.
\end{equation} 
Section 2 gives the proofs of Propositions~\ref{prop:21}, \ref{prop:23}, and \ref{prop:25}, and of Lemma~\ref{lem:26}. Notice that Proposition~\ref{prop:25} was proved in \cite{BZ5} and Propositions~\ref{prop:21} and \ref{prop:23} were proved in \cite{BZ5} under Assumptions (H1) and (H2). Thus, it remains to prove Lemma~\ref{lem:22} and the trace estimate from Proposition \ref{prop:trace}.

The proof of Lemma~\ref{lem:22} under Assumption (H0) is given in Section~\ref{sec:4}. The proof of Lemma~\ref{lem:22} under Assumptions (H1) and (H2) is given in Section~\ref{sec:73}.
        
The remaining part of the paper is devoted to the proof of the trace estimate (Proposition \ref{prop:trace}). This result is a refinement of a less precise trace estimate \cite{BZ5}*{Proposition 5.2} and is obtained from the analysis of a suitable evolution $t\to U_n(t)$ based on Fourier transform, as in \cite{BZ5}*{Section 6}. This reduction is presented in Section \ref{sec:4'} where we give details which are more involved than in \cite{BZ5}. More precisely, in Section \ref{sec:4'} we state Proposition \ref{prop:4'} which gives $\ord(n^{-\gamma+\varepsilon})$ estimates for the diagonal entries from the Neumann series expansion of $t\to U_n(t)$ and we show that Proposition~\ref{prop:4'} implies Proposition~\ref{prop:trace}. The proof of Proposition \ref{prop:4'} is given in Section~\ref{sec:63} and is based on approximations by oscillatory integrals (Lemma \ref{lem:61}). In Section \ref{sec:7} we observe that the construction of these approximations was already made in \cite{BZ5}*{Section 10} and give the proof of the regularity properties claimed in Lemma~\ref{lem:61}.
   
Concerning the proof of Proposition~\ref{prop:4'} we observe that the principal difficulty consists in the control of oscillatory integrals with phase functions depending on parameters. In particular these phase functions can be identically zero for some values of the parameters but an additional integration allows us to neglect the contribution of these bad cases. More precisely the phase functions appear with a large parameter proportional to $n^{\gamma}$ (see Section~\ref{sec:611}) and the results of \cite{BZ5} were based on the fact that the decay of the corresponding oscillatory integrals is of order $n^{-\gamma/2}$. To obtain the results described in this paper we apply the formula of the asymptotic expansion for oscillatory integrals stated in Lemma~\ref{lem:41}. In Section~\ref{sec:63} we investigate the special structure of the main term (of order $n^{-\gamma/2}$) and error terms (of order $n^{-\gamma}$) and we manage to control their dependence on parameters by using an auxiliary estimate proved in Section~\ref{sec:5}.  

\subsection{Notations}  \label{sec:15}

Throughout the paper, we use the following notations:
\begin{enumerate}[\textbullet]
\item
$\C{B}(\C{H})$ is the algebra of linear bounded operators on a Hilbert space $\C{H}$.
\item   
$\D{N}=\{0,1,\dots\}$ is the set of non-negative integers, $\D{N}^* =\{1,2,\dots\}$ is the set of positive integers.
\item
$l^2(\D{Z})$ is the Hilbert space of square-summable complex sequences $x\colon\D{Z}\to\D{C}$ equipped with the scalar product $\scal{x,y}=\sum_{k\in\D{Z}}\overline{x(k)}y(k)$ and with the norm $\norm{x}_{l^2(\D{Z})}\coloneqq\sqrt{\scal{x,x}}$.
\item
$\accol{\vece_k}_{k\in\D{Z}}$ is the canonical basis of $l^2(\D{Z})$, i.e., $\vece_k(j)=\delta_{k,j}$.
\item
$H(j,k)\coloneqq\scal{\vece_j,H\vece_k}$, $j,k\in\D{Z}$ denote the matrix elements of an operator $H$ acting on $l^2(\D{Z})$ and defined on its canonical basis.
\item
$l^2(\D{N}^*)$ is the Hilbert space of square-summable complex sequences $x\colon\D{N}^*\to\D{C}$ equipped with the scalar product $\scal{x,y}=\sum_{k\in\D{N}^*}\overline{x(k)}y(k)$ and the norm $\norm{x}_{l^2(\D{N}^*)}\coloneqq\sqrt{\scal{x,x}}$. It can be identified with the closed subspace of $l^2(\D{Z})$ generated by $\accol{\vece_n}_{n\in\D{N}^*}$, i.e., with the subspace defined by the conditions $x(k)=0$ for any $k\leq 0$.
\end{enumerate}

We use specific notations for some operators acting on $l^2(\D{Z})$:
\begin{enumerate}[\textbullet]
\item
The shift $S\in\C{B}(l^2(\D{Z}))$ is defined by $(Sx)(k)=x(k-1)$,  $k\in\D{Z}$. Thus, $S\vece_k=\vece_{k+1}$.  
\item
$\Lambda$ acts on $l^2(\D{Z})$ by $(\Lambda x)(k)=kx(k)$, $k\in\D{Z}$ for any $x$ such that $(kx(k))_{k\in\D{Z}}\in l^2(\D{Z})$.
\item
For any $q\colon\D{Z}\to\D{C}$ we define the linear operator $q(\Lambda)$ by functional calculus, i.e., $q(\Lambda)$ is the closed operator acting on $l^2(\D{Z})$ and characterized by $q(\Lambda)\vece_k=q(k)\vece_k$, $k\in\D{Z}$. 
\item
If $L$ is a self-adjoint operator which is bounded from below with compact resolvent its spectrum is discrete and we denote
\[
\lambda_1(L)\leq\dots\leq\lambda_k(L)\leq\lambda_{k+1}(L)\leq\dots
\] 
its eigenvalues, enumerated in non-decreasing order, counting multiplicities.
\end{enumerate} 

Finally, we also use the following notations:
\begin{enumerate}[\textbullet]
\item
$\C{S}(\D{R})$ denotes the Schwartz class of rapidly decreasing functions $\chi\colon\D{R}\to\D{C}$.
\item
The Fourier transform $\hat\chi$ of a function $\chi\in\C{S}(\D{R})$ is defined by
\[ 
\hat\chi(t)\coloneqq\int_{-\infty}^{\infty}\chi(\lambda)\eul^{-\ii t\lambda}\frac{\dd\lambda}{2\pi}.
\] 
\item
$\cercle$ denotes the unit circle $\accol{z\in\D{C}:\abs{z}=1}$.
\item
$\tau_{\omega}\colon\cercle\to\cercle$, where $\omega\in\D{R}$ denotes the translation $\ek\mapsto\eul^{\ii(\xi-\omega)}$.
\item
$\class^m(\cercle)$, $m=0,1,2,\dots$ is the space of functions $b\colon\cercle\to\D{C}$ of class $\class^m$ equipped with the norm 
\[ 
\norm{b}_{\class^m(\cercle)}\coloneqq\max_{0\leq k\leq m}\sup_{\xi\in\D{R}}\bigl\lvert\partial_{\xi}^k{b(\ek)}\bigr\rvert.
\]    
\end{enumerate}

Throughout the paper $n\in\D{N}^*$ is the large parameter involved in the asymptotics \eqref{E0} or \eqref{E2}. All error estimates are considered with respect to $n\in\D{N}^*$ and some statements will be established only for $n\geq n_0$, for some large enough constant $n_0$.
\section{Scheme of the proof of Theorems~\ref{thm:11} and \ref{thm:12}}\label{sec:2} 
\subsection{Plan of Section 2}\label{sec:21} 

In Section \ref{sec:22} we introduce auxiliary operators $J_n$, $\tilde V_n$, and $L_n$ and state Propositions \ref{prop:21} and \ref{prop:23} ensuring that $\lambda_n(J)$ is well approximated for large $n$ by suitable eigenvalues of $J_n$ and $L_n$. Moreover, we state Lemma~\ref{lem:22} which gives the asymptotics of the $n$-th diagonal entry of $\tilde V_n$. In Section \ref{sec:23} we first state Proposition \ref{prop:trace} which gives a trace estimate for $L_n$. To derive from this estimate the asymptotic behavior of eigenvalues of $L_n$ we prove Lemma~\ref{lem:26} which allows us to apply Proposition \ref{prop:25}. Finally in Section \ref{sec:24} we check the implication \eqref{144}.

\begin{remark*}
Proposition~\ref{prop:25} was already proved in \cite{BZ5}. Propositions \ref{prop:21} and \ref{prop:23} were proved in \cite{BZ5} under assumptions (H1) and (H2). Lemma~\ref{lem:26} is proved in Section~\ref{sec:233}. Lemma~\ref{lem:22} is proved in Section~\ref{sec:4} under assumption (H0) and in Section~\ref{sec:7} under assumptions (H1) and (H2). Sections~\ref{sec:4'}, \ref{sec:5}, \ref{sec:6}, and \ref{sec:63} are devoted to the proof of Proposition~\ref{prop:trace} (trace estimate).
\end{remark*}

\subsection{Auxiliary operators $\BS{J_n}$, $\BS{\tilde V_n}$, and $\BS{L_n}$} \label{sec:22}
\subsubsection{Cut-off function}\label{sec:220}

These operators were already introduced in \cite{BZ5}. Their definition depends on the choice of a cut-off function $\theta_0\in\class^{\infty}(\D{R})$ such that
\begin{subequations} \label{222}
\begin{equation}\label{220} 
\begin{cases} 
\theta_0(t)=1&\text{if }\abs{t}\leq\frac{1}{6}\,,\\ 
\theta_0(t)=0&\text{if }\abs{t}\geq\frac{1}{5}\,,\\ 
0\leq\theta_0(t)\leq 1&\text{otherwise}.
\end{cases}  
\end{equation} 
From now on such a cut-off is fixed and for any $\tau>0$ we denote
\begin{equation} \label{223}
\theta_{\tau,n}(s)\coloneqq\theta_0\!\left(\frac{s-n}\tau\right).
\end{equation}
\end{subequations} 
Finally we define $v_n,\,a_n\colon\D{Z}\to\D{R}$ by 
\begin{subequations} \label{avn}
\begin{align}   \label{vn}
v_n(k)&\coloneqq v(k)\theta_{n,n}(k)^2,\\
\label{an}
a_n(k)&\coloneqq\left(a(n)+(k-n)\delta a(n)\right)\theta_{2n,n}(k). 
\end{align}
\end{subequations}

\subsubsection{Operators $J_n$}\label{sec:221}
 
As in \cite{BZ5}*{Section 2.2} we define the self-adjoint operator $J_n$ on $l^2(\D{Z})$ by 
\begin{subequations}  \label{Jn}
\begin{equation} \label{Jn+}
(J_nx)(k)=(k+v_n(k))x(k)+a_n(k)x(k+1)+a_n(k-1)x(k-1)
\end{equation} 
for any $x$ such that $(kx(k))_{k\in\D{Z}}\in l^2(\D{Z})$. Using notations of Section~\ref{sec:15} we can write
\begin{equation}\label{27} 
J_n=\Lambda+v_n(\Lambda)+a_n(\Lambda)S^{-1}+Sa_n(\Lambda),  
\end{equation}
\end{subequations} 
to compare with the similar expression for $J$. Each operator $J_n$ is a finite rank perturbation of $\Lambda$, hence its spectrum $\spec(J_n)$ is discrete and can be written
\begin{equation}  \label{Jnev}
\spec(J_n)=\accol{\lambda_k(J_n)}_{k\in\D{Z}}
\end{equation}
where $(\lambda_k(J_n))_{k\in\D{Z}}$ denotes the non-decreasing sequence of eigenvalues of $J_n$ counted with their multiplicities, well labeled up to translation.

\begin{proposition}\label{prop:21} 
Let $J$ and $J_n$ be defined by \eqref{J''} and \eqref{Jn}, respectively. Assume we are in one of the following two cases:
\begin{enumerate}[\rm(a)]
\item   \label{H0}
\emph{(H0)} is satisfied.
\item \label{H12}
\emph{(H1)} for some $0<\gamma\leq\frac{1}{2}$ and \emph{(H2)} are satisfied.
\end{enumerate}
Let $n_0\in\D{N}$ be large enough. Then, for any $n\geq n_0$ we can enumerate the eigenvalues of $J_n$ as in \eqref{Jnev} so that we have the large $n$ estimate
\begin{equation}  \label{J-Jn}
\lambda_n(J)=\lambda_n(J_n)+\ord(n^{3\gamma-2}),
\end{equation}
where $\gamma=\frac{1}{2}$ in case \eqref{H0}.
\end{proposition} 

\begin{proof} 
Case~\eqref{H12} is already proven in \cite{BZ5}*{Proposition 12.1}. Case~\eqref{H0} requires a new proof, since in that case there is no restriction on $\rho_N=\abs{\rho}$.

Let $J_n^+$ be the restriction of $J_n$ to the subspace $l^2(\D{N}^*)$ which is invariant under $J_n$. The operator $J_n^+$ is self-adjoint and bounded from below with compact resolvent. Its spectrum is discrete: $\sigma(J_n^+)=\accol{\lambda_k(J_n^+)}_{k\geq 1}$, where $\lambda_1(J_n^+)\leq\dots\leq\lambda_k(J_n^+)\leq\lambda_{k+1}(J_n^+)\leq\dots$ denote its eigenvalues, enumerated in non-decreasing order, counting multiplicities. Since $J_n\vece_k=k\vece_k$ for $k\leq 0$ we can write
\begin{equation}  \label{260}
\sigma(J_n)=\accol{k\in\D{Z}:k\leq 0}\cup\sigma(J_n^+)=\accol{k\in\D{Z}:k\leq 0}\cup\accol{\lambda_k(J_n^+)}_{k\geq 1}.
\end{equation}
\begin{stepone*}
We will show the estimate
\begin{equation} \label{Jnk-k}
\sup_{n,k\geq 1}\abs{\lambda_k(J_n^+)-k}<\infty.
\end{equation}
Let $J_{0,n}^+$ denote the operator acting on $l^2(\D{N}^*)$ by
\[
(J_{0,n}^+x)(k)=kx(k)+a_n(k)x(k+1)+a_n(k-1)x(k-1).  
\] 
Since $J_n^+-J_{0,n}^+=\diag\left(v_n(k)\right)_{k=1}^\infty$ with $\abs{v_n(k)}\leq\abs{\rho}$, the min-max principle applies and gives
\begin{equation} \label{Jnk-J0nk}
\sup_{n,k\geq 1}\abs{\lambda_k(J_n^+)-\lambda_k(J_{0,n}^+)}<\infty.
\end{equation}
Hence, to get \eqref{Jnk-k} it suffices to show
\begin{equation} \label{J0nk-k}
\sup_{n,k\geq 1}\abs{\lambda_k(J_{0,n}^+)-k}<\infty.
\end{equation}
In \cite{BZ5}*{Proposition 3.1} we proved the large $n$ estimate
\begin{equation} \label{E3}   
\sup_{k\in\D{N}^*}\,\abs{\lambda_k(J_{0,n}^+)-l_n(k)}=\ord(n^{3\gamma-2}),
\end{equation}   
where
\begin{equation} \label{300}
l_n(k)\coloneqq k+a_n(k-1)^2-a_n(k)^2,\quad k\geq 1.   
\end{equation} 
Hence, to obtain \eqref{J0nk-k} it suffices to show
\begin{equation} \label{lnk-k}
\sup_{n,k\geq 1}\abs{l_n(k)-k}<\infty.
\end{equation}
To prove \eqref{lnk-k} we denote
\begin{equation} \label{a1nk}
a_{1,n}(k)\coloneqq l_n(k)-k=a_n(k-1)^2-a_n(k)^2.
\end{equation}
By estimate $a_n(n-1)^2-a_n(n)^2=a(n-1)^2-a(n)^2+\ord(n^{2\gamma-2})$ from \cite{BZ5}*{(2.5)} we get
\[
l_n(n)=n+a(n-1)^2-a(n)^2+\ord(n^{2\gamma-2}). 
\]
In case~\eqref{H0}, $\gamma=\frac{1}{2}$ and $a(n)=a_1\sqrt{n}$, and thus we obtain
\begin{equation}   \label{lnn}
a_{1,n}(n)=l_n(n)-n=-a_1^2+\ord(n^{-1}).
\end{equation}
Moreover, in \cite{BZ5}*{Section 3.3} we proved the estimate
\begin{equation}  \label{da1n}
\abs{\delta a_{1,n}(k)}\leq Cn^{-1},
\end{equation}
and thus
\begin{equation}  \label{a1nk-a1nn}
\abs{a_{1,n}(k)-a_{1,n}(n)}\leq\abs{k-n}Cn^{-1},
\end{equation}
By definition, $a_{1,n}(k)=0$ for any $k\geq 2n$, and \eqref{a1nk-a1nn} ensures
\[
\sup_{n,k\geq 1}\abs{a_{1,n}(k)-a_{1,n}(n)}<\infty.
\]
Using \eqref{lnn} which implies $a_{1,n}=\ord(1)$ we finally get \eqref{lnk-k}:
\[
\sup_{n,k\geq 1}\abs{a_{1,n}(k)}<\infty.
\]
\end{stepone*}
\begin{steptwo*}
Now we will prove that 
\begin{equation} \label{tJnk-k}
\sup_{n,k\geq 1}\abs{\lambda_k(\tilde J_n^+)-k}<\infty.
\end{equation}
where the operator $\tilde J_n^+$ acts on $l^2(\D{N}^*)$ by
\begin{subequations}  \label{tildeJn}
\begin{equation}   \label{tJn}
(\tilde J_n^+x)(k)=d_n(k)x(k)+\tilde a_n(k)x(k+1)+\tilde a_n(k-1)x(k-1),
\end{equation}
with
\begin{equation}  \label{tan}
\tilde a_n(k)=
\begin{cases}
a(k)&\text{if }n-C_1(n+1)^{\gamma}\leq k\leq n+C_1(n+1)^{\gamma},\\
a_n(k)&\text{otherwise},
\end{cases}
\end{equation}
\end{subequations}
and with a large enough constant $C_1$. Estimate \cite{BZ5}*{(12.8)} ensures that
\begin{equation}  \label{eq:1232}
\sup_{k\in\D{N}^*}\abs{\lambda_k(\tilde J_n^+)-\lambda_k(J_n^+)}=\ord(n^{3\gamma-2}).
\end{equation}
Thus, \eqref{Jnk-k} and \eqref{eq:1232} for $\gamma=1/2$ give \eqref{tJnk-k}.
\end{steptwo*}
\begin{stepthree*}
Taking $C_1>3$ in the definition \eqref{tildeJn} we have the property
\[ 
n-3n^{\gamma}\leq k\leq n+3n^{\gamma}\implies J\vece_k=\tilde J_n^+\vece_k.
\] 
Then \cite{BZ5}*{Proposition 12.5} applies and for any $\nu>0$ there exists $n(\nu)$ such that the inequalities
\begin{subequations}  \label{card}
\begin{align}    \label{card1}
&\card\accol{k\in\D{N}^*:\lambda'+\lambda^{-\nu}<\lambda_k(\tilde J_n^+)\leq\lambda-\lambda^{-\nu}}\leq\card\accol{k\in\D{N}^*:\lambda'<\lambda_k(J)\leq\lambda},\\\label{card2}
&\card\accol{k\in\D{N}^*:\lambda'<\lambda_k(J)\leq\lambda}\leq\card\accol{k\in\D{N}^*:\lambda'-\lambda^{-\nu}<\lambda_k(\tilde J_n^+)\leq\lambda+\lambda^{-\nu}},
\end{align}
\end{subequations}
hold if $n-2n^{\gamma}\leq\lambda'<\lambda\leq n+2n^{\gamma}$ and $n\geq n(\nu)$. Let $\accol{\Delta_n^\kappa}_{n\geq 1}$, $0<\kappa<1$ be the family of intervals defined by
\[
\Delta_n^\kappa=(n-a_1^2-\kappa,n-a_1^2+\kappa\rbrack.
\]
Using \eqref{card1} with $\lambda'=n-a_1^2-\frac{3}{4}$, $\lambda=n-a_1^2+\frac{3}{4}$, and \eqref{card2} with $\lambda'=n-a_1^2-\frac{1}{4}$, $\lambda=n-a_1^2+\frac{1}{4}$ we get that for some $n_0$ and any $n\geq n_0$ we have the inequalities
\begin{subequations}  \label{est}
\begin{align}   \label{est1}
&\card\accol{k\in\D{N}^*:\lambda_k(\tilde J_n^+)\in\Delta_n^{1/2}}\leq\card\accol{k\in\D{N}^*:\lambda_k(J)\in\Delta_n^{3/4}},\\\label{est2}
&\card\accol{k\in\D{N}^*:\lambda_k(\tilde J_n^+)\in\Delta_n^{1/2}}\geq\card\accol{k\in\D{N}^*:\lambda_k(J)\in\Delta_n^{1/4}}.
\end{align}
\end{subequations}
\end{stepthree*}
\begin{stepfour*}
Yanovich estimate \eqref{Y0} gives $\lambda_n(J)=n-a_1^2+\osmall(1)$. Hence, for any $0<\kappa<1$ and $n(\kappa)$ large enough
\begin{equation} \label{lanJ}
n\geq n(\kappa)\implies\spec(J)\cap\Delta_n^\kappa=\accol{\lambda_n(J)}.
\end{equation}
Thus we can find $n_1\geq n_0$ such that the right hand sides of \eqref{est} are both equal to $1$ for $n\geq n_1$. Hence, inequalities \eqref{est} imply
\begin{equation}   \label{nkappa}
n\geq n_1\implies\card\accol{k\in\D{N}^*:\lambda_k(\tilde J_n^+)\in\Delta_n^{1/2}}=1.
\end{equation}
Let $k(n)=n+m(n)$, $m(n)\in\D{Z}$, $n\geq n_1$, be the unique $k$ such that $\lambda_k(\tilde J_n^+)\in\Delta_n^{1/2}$:
\[
n\geq n_1\implies\sigma(\tilde J_n^+)\cap\Delta_n^{1/2}=\accol{\lambda_{n+m(n)}(\tilde J_n^+)}.
\]
The eigenvalues $\lambda_{n+m(n)}(\tilde J_n^+)$ are of multiplicity one. Moreover, by \eqref{tJnk-k} there exist $\tilde n_0$, $m_0\in\D{N}^*$ large enough such that
\begin{equation} \label{m0}
n\geq\tilde n_0\implies\abs{m(n)}\leq m_0.
\end{equation}
\end{stepfour*}
\begin{stepfive*}
Using \eqref{card2} with $\lambda'=\lambda_n(J)-\lambda_n(J)^{-\nu}$, $\lambda=\lambda_n(J)$ we find that, for any fixed $\nu>0$,
\begin{equation}  \label{nu}
\lambda_{n+m(n)}(\tilde J_n^+)=\lambda_n(J)+\ord(n^{-\nu}).
\end{equation}
Then, using \eqref{eq:1232} with $\gamma=\frac{1}{2}$ and \eqref{nu} with $\nu=\frac{1}{2}$ we get
\begin{equation}  \label{Jn+Jn}
\lambda_{n+m(n)}(J_n^+)=\lambda_n(J)+\ord(n^{-1/2}).
\end{equation}
Let $n\geq\tilde n_1$ with $\tilde n_1$ large enough. Then, according to \eqref{260} and \eqref{m0} we can label the eigenvalues of $J_n$ in nondecreasing order, counting multiplicity, so that 
\begin{equation}  \label{Jnn}
\lambda_n(J_n)=\lambda_{n+m(n)}(J_n^+). 
\end{equation}
The proof of \eqref{J-Jn} is then completed by combining \eqref{Jn+Jn} with \eqref{Jnn}.\qedhere
\end{stepfive*}
\end{proof}

\subsubsection{Operators $\tilde V_n$}\label{sec:222}
 
As in \cite{BZ5}*{Section 5.2} we define self-adjoint operators $B_n,\tilde V_n\in\C{B}(l^2(\D{Z}))$ by 
\begin{subequations}\label{vtBn}
\begin{align}\label{Bn} 
B_n&\coloneqq\ii\left(a_n(\Lambda)S^{-1}-Sa_n(\Lambda)\right),\\
\label{vtilde}
\tilde V_n&\coloneqq\eul^{\ii B_n}v_n(\Lambda)\eul^{-\ii B_n}, 
\end{align}
\end{subequations}   
where $\accol{a_n(k)}_{k\in\D{Z}}$ and $\accol{v_n(k)}_{k\in\D{Z}}$ are defined in \eqref{avn}, and we denote by $\accol{g_n(k)}_{k\in\D{Z}}$ the sequence of diagonal entries of $\tilde V_n$:
\begin{equation}\label{gnk} 
g_n(k)\coloneqq\tilde V_n(k,k).  
\end{equation} 
Notice that $\tilde V_n$ and $g_n$ depend on both sequences $\accol{a(k)}_{k=1}^{\infty}$, $\accol{v(k)}_{k=1}^{\infty}$, and on the cut-off $\theta_0$.

\begin{lemma}\label{lem:22}
Let $\tilde V_n$ be defined by \eqref{vtBn} and assume we are in one of the following two cases:
\begin{enumerate}[\rm(a)]
\item  \label{H022}
\emph{(H0)} is satisfied.
\item  \label{H1222}
\emph{(H1)} and \emph{(H2)} with $\av{v}=0$ are satisfied.
\end{enumerate}
Then for any $\varepsilon>0$ one has the large $n$ asymptotics
\begin{equation}\label{gnn} 
g_n(n)=\F{r}(n)+\ord(n^{-\gamma+\varepsilon}).
\end{equation} 
where $\F{r}(n)$ is given by \eqref{R0} in case \eqref{H022} and by \eqref{R2} in case \eqref{H1222}.
\end{lemma} 

\begin{proof}  
See Section \ref{sec:43} for case (a) and Section \ref{sec:73} for case (b). 
\end{proof}
  
\subsubsection{Operators $L_n$}\label{sec:223} 

As in \cite{BZ5}*{Section 5.2} we introduce operators $L_n$ acting on $l^2(\D{Z})$ by
\begin{subequations}\label{Lnln}
\begin{equation}\label{Ln} 
L_n\coloneqq l_n(\Lambda)+\tilde V_n,
\end{equation}
where $\tilde V_n$ is defined by \eqref{vtBn} and $l_n(k)$ is as in \eqref{300}:
\begin{equation}\label{lnk} 
l_n(k)\coloneqq k+a_n(k-1)^2-a_n(k)^2.  
\end{equation} 
\end{subequations}
Since $l_n(\Lambda)$ is a diagonal operator with discrete spectrum and $\tilde V_n$ is bounded, the spectrum of $L_n$ is discrete and can be written
\begin{equation}  \label{Lnev}
\spec(L_n)=\accol{\lambda_k(L_n)}_{k\in\D{Z}}
\end{equation}
where $(\lambda_k(L_n))_{k\in\D{Z}}^{\infty}$ denotes the non-decreasing sequence of eigenvalues of $L_n$ counted with their multiplicities, well-labeled up to translation. Moreover, the subspace $l^2(\D{N}^*)$ is invariant by $B_n$, hence by $\tilde V_n$ and $L_n$, and $L_n\vece_k=k\vece_k$ if $k\leq 0$.

\begin{proposition}\label{prop:23} 
Let $J_n$ and $L_n$ be defined by \eqref{Jn} and \eqref{Lnln}, respectively. We assume we are in one of the following two cases:
\begin{enumerate}[\rm(a)]
\item  \label{H023}
\emph{(H0)} is satisfied.
\item  \label{H1223}
\emph{(H1)} for some $0<\gamma\leq\frac{1}{2}$ and \emph{(H2)} are satisfied.
\end{enumerate}
Then, the eigenvalues of $L_n$ can be enumerated in nondecreasing order, counting multiplicity, so that one has the large $n$ estimate
\begin{equation}  \label{J_n-L_n}
\lambda_n(J_n)=\lambda_n(L_n)+\ord(n^{3\gamma-2}),
\end{equation}
where $\gamma=\frac{1}{2}$ in case \eqref{H023}.
\end{proposition} 

\begin{proof} 
We indeed have $\norm{J_n-L_n}=\ord(n^{3\gamma-2})$, see \cite{BZ5}*{proof of Proposition 5.1}. Hence, we can enumerate the eigenvalues of $L_n$ in \eqref{Lnev} so that
\[
\sup_{k\in\D{Z}}\abs{\lambda_k(J_n)-\lambda_k(L_n)}=\ord(n^{3\gamma-2}).\qedhere
\]
\end{proof}

\begin{summary*} 
In both cases \eqref{H023} and \eqref{H1223}, Propositions \ref{prop:21} and \ref{prop:23} imply the large $n$ estimate $\lambda_n(J)=\lambda_n(L_n)+\ord(n^{3\gamma-2})$. Since $\gamma\leq\frac{1}{2}$ we also get
\begin{equation}\label{summary}
\lambda_n(J)=\lambda_n(L_n)+\ord(n^{-\gamma}).
\end{equation}
\end{summary*}

\subsection{Trace estimate and its consequences}\label{sec:23} 
\subsubsection{The trace estimate}\label{sec:231}
 
Further on we denote
\begin{equation}\label{4'1} 
\tilde l_n(k)\coloneqq l_n(k)+g_n(k)
\end{equation}
where $l_n$ and $g_n$ are given by \eqref{lnk} and \eqref{gnk}, respectively. Then, for $\chi\in\C{S}(\D{R})$ we consider
\begin{subequations}\label{tG0n}
\begin{equation}\label{G0n}
\tilde{\C{G}}_n^0\coloneqq\sum_{j\in\D{Z}}\Bigl(\chi\bigl(\lambda_{n+j}(L_n)-l_n(n)\bigr)-\chi\bigl(\tilde l_n(n+j)-l_n(n)\bigr)\Bigr)    
\end{equation}          
with $L_n$ as in \eqref{Ln}. Writing $k=j+n$ in \eqref{G0n} we get the expression
\begin{equation}\label{4'0}
\tilde{\C{G}}_n^0=\sum_{k\in\D{Z}}\Bigl(\chi(\lambda_k(L_n)-l_n(n))-\chi(\tilde{l}_n(k)-l_n(n))\Bigr).    
\end{equation}  
\end{subequations}
Introducing the diagonal operators 
\begin{subequations}\label{4'23}
\begin{align}\label{4'2} 
L_{0,n}&\coloneqq l_n(\Lambda),\\ 
\label{4'3} 
\tilde L_{0,n}&\coloneqq\tilde l_n(\Lambda)=L_{0,n}+g_n(\Lambda) 
\end{align}
\end{subequations} 
we see that the r.h.s.\ of \eqref{4'0} is the trace of an operator:    
\begin{equation}\label{G0n'}
\tilde{\C{G}}_n^0=\tr\Bigl(\chi(L_n-l_n(n))-\chi(\tilde L_{0,n}-l_n(n))\Bigr).  
\end{equation} 
Notice that $\tilde{\C{G}}_n^0$ depends on $\chi\in\C{S}(\D{R})$. It also depends on $\accol{a(k)}_{k=1}^{\infty}$, $\accol{v(k)}_{k=1}^{\infty}$, and $\theta_0$.

\begin{warning*}
This trace $\tilde{\C{G}}_n^0$ differs from the trace $\C{G}_n^0$ considered in \cite{BZ5}*{formulas (5.10)}, that uses $L_{0,n}$ instead of $\tilde L_{0,n}$, and thus does not involve $g_n$. 
\end{warning*}

\begin{proposition}[trace estimate]\label{prop:trace}   
Let $\tilde{\C{G}}_n^0$ be the trace defined by \eqref{G0n} under the additional assumption that $\chi\in\C{S}(\D{R})$ has Fourier transform with compact support. Then, under assumption \emph{(H0)} or assumptions \emph{(H1)} and \emph{(H2)} and for any $\varepsilon>0$ one has the large $n$ estimate  
\begin{equation}\label{tr}
\tilde{\C{G}}_n^0=\ord(n^{-\gamma+\varepsilon}). 
\end{equation} 
\end{proposition} 

\begin{proof} 
See Section \ref{sec:4'4} where the proof is reduced to that of Proposition \ref{prop:4'}. See also Section \ref{sec:63} where the proof of Proposition \ref{prop:4'} is given.
\end{proof}  

\subsubsection{Comparison of the asymptotic behavior of two sequences}\label{sec:232}
 
As before $(l_n(k))_{k\in\D{Z}}$ is defined by \eqref{lnk} where $(a_n(k))_{k\in\D{Z}}$ is given by \eqref{an} under assumption~(H1).  

\begin{proposition}\label{prop:25}   
Let $(l_n(k))_{k\in\D{Z}}$ be defined by \eqref{lnk} under assumption\emph{(H1)} for some $0<\gamma\leq\frac{1}{2}$. For each $j\in\D{Z}$ let $(r_n^0(j))_{n=1}^{\infty}$ and $(r_n^1(j))_{n=1}^{\infty}$ be real valued sequences such that
\[
\sup_{j\in\D{Z},\,n\geq 1}\left(\abs{r_n^0(j)}+\abs{r_n^1(j)}\right)<\infty.
\]
Assume they also satisfy 
\begin{subequations}\label{232ab}
\begin{align}\label{232a}
&\sup_{\abs{j}\leq n^{\gamma_0}}\abs{r_n^i(j+N)-r_n^i(j)}\leq Cn^{\gamma-1}\qquad(i=0,1),\\
\label{232b}
&\sup_{n>n_0}\sup_{\abs{j}\leq n^{\gamma_0}}\left(\abs{r_n^0(j)}+\abs{r_n^1(j)}\right)\leq\rho' 
\end{align}
\end{subequations}
for some $\gamma_0>0$, $n_0\in\D{N}$, and $\rho'>0$ such that
\begin{equation}\label{232}
\rho'<
\begin{cases} 
\frac{1}{2}&\text{when }N=2,\\
\frac{1}{\pi\sqrt{N}}&\text{when }N\geq 3. 
\end{cases}   
\end{equation}
Assume moreover that for some $\varepsilon>0$
\begin{equation}\label{232c}
\sum_{j\in\D{Z}}\Bigl(\chi(l_n(n+j)+r_n^1(j)-l_n(n))-\chi(l_n(n+j)+r_n^0(j)-l_n(n))\Bigr)=\ord(n^{-\gamma+\varepsilon})
\end{equation}    
holds for any $\chi\in\C{S}(\D{R})$ whose Fourier transform has compact support. Then 
\begin{equation}\label{232d}
r_n^1(0)=r_n^0(0)+\ord(n^{-\gamma+\varepsilon}). 
\end{equation} 
\end{proposition} 

\begin{proof}
It suffices to adapt the proof of \cite{BZ5}*{Proposition 11.1} as follows. Remove the first two lines, define $G_n^{\chi}$ as the l.h.s.\ of \eqref{232c} and replace the error terms $\ord(n^{-\gamma/2}\ln n)$ by $\ord(n^{-\gamma+\varepsilon})$. 
\end{proof}
 
\subsubsection{Application of Proposition \ref{prop:25}}\label{sec:233}

We will apply Proposition~\ref{prop:25} to the case where the two sequences $(r_n^0(j))_{n=1}^\infty$ and $(r_n^1(j))_{n=1}^\infty$ are given by
\begin{subequations}\label{2330}
\begin{align}\label{233}
r_n^0(j)&\coloneqq g_n(n+j),\\
\label{233'}
r_n^1(j)&\coloneqq\lambda_{n+j}(L_n)-l_n(n+j). 
\end{align} 
\end{subequations}

\begin{proposition}\label{prop:26}
Let $L_n$ and $\tilde l_n(k)$ be given by \eqref{Ln} and \eqref{4'1}, respectively. We assume we are in one of the following two cases:
\begin{enumerate}[\rm(a)]
\item  \label{H026}
\emph{(H0)} is satisfied, $\gamma=\frac{1}{2}$.
\item  \label{H1226}
\emph{(H1)} for some $0<\gamma\leq\frac{1}{2}$ and \emph{(H2)} are satisfied.
\end{enumerate}
Then, for any $\varepsilon>0$ one has the large $n$ estimate
\begin{equation}\label{233''}
\lambda_n(L_n)=\tilde l_n(n)+\ord(n^{-\gamma+\varepsilon}).
\end{equation} 
\end{proposition}

\begin{proof}
We first show that Proposition~\ref{prop:25} applies to the case where $r_n^0$ and $r_n^1$ are given by \eqref{2330}. In this case the l.h.s.\ of \eqref{232c} is $\tilde{\C{G}}_n^0$ as defined by \eqref{G0n} and thus the trace estimate \eqref{tr} in Proposition \ref{prop:trace} says that assumption~\eqref{232c} is satisfied for any $\varepsilon>0$ and any $\chi\in\C{S}(\D{R})$ whose Fourier transform has compact support. That the conditions \eqref{232ab} are also satisfied is proven in Lemma~\ref{lem:26} below. Thus, Proposition~\ref{prop:25} applies and the assertion \eqref{232d} is exactly \eqref{233''} since in that case $r_n^1(0)-r_n^0(0)=\lambda_n(L_n)-\tilde l_n(n)$.
\end{proof}

\begin{lemma}  \label{lem:26} 
The sequences $r_n^0$ and $r_n^1$ given by \eqref{2330} satisfy the conditions \eqref{232ab} in each of the following two cases:
\begin{enumerate}[\rm(a)]
\item \label{H027}
\emph{(H0)} is satisfied, $N=2$, $\gamma=\frac{1}{2}$, $\gamma_0=\frac{1}{4}$, any $0<\rho'<\frac{1}{2}$, and $n_0\in\D{N}$ large enough.
\item \label{H1227}
\emph{(H1)} and \emph{(H2)} are satisfied, $0<\gamma\leq\frac{1}{2}$, $\gamma_0=\gamma$, any $\rho'>\rho_N$ satisfying \eqref{232}, and $n_0\in\D{N}$ large enough.
\end{enumerate}
\end{lemma} 

\begin{proof} 
We first prove \eqref{232b}.
\begin{stepone*}[estimate of $r_n^0(j)$ in cases \eqref{H027} and \eqref{H1227}]
First note that $r_n^0(j)=g_n(n+j)$. Then, taking $t_1=0$ and $j+n$ instead of $j$ in \cite{BZ5}*{Lemma 6.3 (i)} we get
\begin{equation}\label{233b}
\sup_{\abs{j}\leq n^{\gamma}}\abs{r_n^0(j)}=\sup_{\abs{j}\leq n^{\gamma}}\abs{g_n(n+j)}=\ord(n^{-\gamma/2}).   
\end{equation}
We indeed have the relation $g_{1,n,j}(0)=\ii g_n(j)$ with $g_{1,n,j}$ defined in \cite{BZ5}*{(6.16)}. 
\end{stepone*}
\begin{steptwo*}[estimate of $r_n^1(j)$ in case \eqref{H027}]
It suffices to show 
\begin{equation}  \label{rn1j}
\sup_{\abs{j}\leq n^{1/4}}\abs{r_n^1(j)}=\ord(n^{-1/16}).
\end{equation}
To obtain \eqref{rn1j} we introduce the quantity 
\begin{equation}   \label{trn1}
\tilde r_n^1(j)\coloneqq\lambda_{n+j}(J)-(n+j-a_1^2) 
\end{equation}
and observe that the Yanovich estimate \eqref{Y0} ensures 
\begin{equation}  \label{trn1j}
\sup_{\abs{j}\leq n^{1/4}}\abs{\tilde r_n^1(j)}=\ord(n^{-1/16}). \end{equation}
Due to \eqref{trn1j} to complete the proof of \eqref{rn1j} it suffices to show the estimate 
\begin{equation}  \label{rn1-trn1}
\sup_{\abs{j}\leq n^{1/4}}\abs{r_n^1(j)-\tilde r_n^1(j)}=\ord(n^{-1/2}). 
\end{equation}
In order to get \eqref{rn1-trn1} we first show the  estimate 
\begin{equation}  \label{Jn-Jnj}
\sup_{\abs{j}\leq n^{1/4}}\abs{\lambda_n(J)-\lambda_n(J_{n-j})}=\ord(n^{-1/2}).
\end{equation}
For this purpose we observe that the proof given in \cite{BZ5}*{Section 12.5} still holds if $\tilde J_n^+$ is replaced by $\tilde J_{n-j}^+$ with $\abs{j}\leq n^{1/4}$ and all estimates are uniform with respect to $j$. Hence we can replace $J_n$ by $J_{n-j}$ in \eqref{J-Jn}, and \eqref{Jn-Jnj} is proved. Moreover, replacing $n$ by $n+j$ we can write \eqref{Jn-Jnj} in the form
\begin{equation}  \label{lanj}
\sup_{\abs{j}\leq n^{1/4}}\abs{\lambda_{n+j}(J)-\lambda_{n+j}(J_n)}=\ord(n^{-1/2}).
\end{equation}
Then we complete the proof of \eqref{rn1-trn1} by showing   
\begin{equation} \label{lnnj}
\sup_{\abs{j}\leq n^{1/4}}\abs{l_n(n+j)-(n+j-a_1^2)}=\ord(n^{-3/4}).
\end{equation}
In order to show \eqref{lnnj} we consider $a_{1,n}(k)=l_n(k)-k$ as in \eqref{a1nk}. Using \eqref{da1n} we get
\[
\abs{l_n(n+j)-l_n(n)-j}=\abs{a_{1,n}(n+j)-a_{1,n}(n)}\leq\abs{j}\,Cn^{-1},
\]
hence
\[
\sup_{\abs{j}\leq n^{1/4}}\abs{l_n(n+j)-l_n(n)-j}=\ord(n^{-3/4}),
\]
and \eqref{lnnj} follows by using \eqref{lnn} in the last estimate.
\end{steptwo*}
\begin{stepthree*}[estimate of $r_n^1(j)$ in case \eqref{H1227}]
In \cite{BZ5}*{Section 11.3} where $r_n^1$ is denoted by $r_n$ (see \cite{BZ5}*{(11.16)}) we have shown the estimate
\begin{equation}\label{233a}
\sup_{j\in\D{Z}}\,\abs{r_n^1(j)}\leq\rho_N+C_1n^{3\gamma-2}. 
\end{equation} 
We indeed have the relation $g_{1,n,j}(0)=\ii g_n(j)$ with $g_{1,n,j}$ defined in \cite{BZ5}*{(6.16)}. Using \eqref{233a}, \eqref{233b} and taking $\gamma_0=\gamma$ we can estimate the l.h.s.\ of \eqref{232b} by $\rho_N+C_1n^{3\gamma-2}+C_2n^{-\gamma/2}$. Moreover, by assumption \eqref{rho-N} on $\rho_N$ we can choose $\rho'>\rho_N$ satisfying \eqref{232}. We conclude that \eqref{232b} holds if $n_0$ satisfies $C_1n_0^{3\gamma-2}+C_2n_0^{-\gamma/2}\leq\rho'-\rho_N$, and that is possible since $0<\gamma<\frac{2}{3}$.
\end{stepthree*}
We now prove \eqref{232a}. 
\begin{stepfour*}[proof of \eqref{232a} for $i=0$]
Since $g_n(k+N)=\scal{\vece_k,S^{-N}\tilde V_nS^N\vece_k}$ it suffices to prove      
\begin{equation}\label{73a}
\norm{S^{-N}\tilde V_n S^N-\tilde V_n}=\ord(n^{\gamma-1}).
\end{equation}
In order to show \eqref{73a} we first observe that $S^{-N}v(\Lambda)S^N=v(\Lambda+N)=v(\Lambda)$ and $\norm{S^{-N}\theta_{n,n}(\Lambda)S^N-\theta_{n,n}(\Lambda)}=\norm{\theta_0((\Lambda+N)/n-I)-\theta_0(\Lambda/n-I)}=\ord(n^{-1})$ ensure 
\begin{equation}\label{73a'}
\norm{S^{-N}v_n(\Lambda)S^N-v_n(\Lambda)}=\ord(n^{\gamma-1}).
\end{equation}
Similarly, $\norm{S^{-N}a_{n}(\Lambda)S^N-a_{n}(\Lambda)}=\norm{a_n(\Lambda+N)-a_n(\Lambda)}=\ord(n^{\gamma-1})$ implies 
\begin{equation}\label{73a''}
\norm{S^{-N}\eul^{\pm\ii B_n}S^N-\eul^{\pm\ii B_n}}=\ord(n^{\gamma-1})
\end{equation}  
and \eqref{73a} follows from \eqref{73a'} and \eqref{73a''}.\end{stepfour*}
\begin{stepfive*}[proof of \eqref{232a} for $i=1$ in case \eqref{H1227}]
In \cite{BZ5}*{Section 11.3} where $r_n^1$ is denoted by $r_n$ we already checked that \eqref{232a} holds for $i=1$. 
\end{stepfive*}
\begin{stepsix*}[proof of \eqref{232a} for $i=1$ in case \eqref{H027}]
We observe that if $0<\kappa<1$ then combining \eqref{lanJ} and \eqref{lanj} we can choose $n_\kappa$ large enough to ensure 
\begin{equation}  \label{kappa}
n\geq n_\kappa\text{ and }\abs{j}\leq n^{1/4}\implies\spec(J_n)\cap 
(n+j-a_1^2-\kappa,n+j-a_1^2+\kappa\rbrack=\accol{\lambda_{n+j}(J_n)} \end{equation}
and \eqref{kappa} allows us to obtain \eqref{232a} for $i=1$ following the proof given in \cite{BZ5}*{Section 11.3}.\qedhere
\end{stepsix*}
\end{proof}
  
\subsection{Proof of Theorems \ref{thm:11} and \ref{thm:12}}\label{sec:24} 

Recall that Propositions~\ref{prop:21} and \ref{prop:23} give the large~$n$ estimate $\lambda_n(J)=\lambda_n(L_n)+\ord(n^{-\gamma})$. Combining this estimate \eqref{summary} with estimate \eqref{233''} of $\lambda_n(L_n)$ in Proposition \ref{prop:26} we obtain
\begin{equation}\label{24}
\lambda_n(J)=\tilde l_n(n)+\ord(n^{-\gamma+\varepsilon}) 
\end{equation}
for any $\varepsilon>0$. The desired estimates \eqref{E0} and \eqref{E2} follow from \eqref{24}, using the estimate \eqref{gnn} of $g_n(n)$, and from
\[
l_n(n)-n=a_n(n-1)^2-a_n(n)^2=a(n-1)^2-a(n)^2+\ord(n^{2\gamma-2})
\]
whose last estimate comes from \cite{BZ5}*{Section 2.3, (2.5)}.\qed

\section{Proof of Lemma \ref{lem:22} \textup{(a)}}\label{sec:4}

In section \ref{sec:41} we prove a stationary phase formula for some type of oscillatory integral (Lemma \ref{lem:41}). Then we assume that the diagonal and off-diagonal entries of $J$ satisfy (H1) and $d(k)=k+v(k)$ with $v(k)=(-1)^k\rho$. In section \ref{sec:42} we prove an approximation result of $g_n(n)$ by an oscillatory integral of the above type (Lemma \ref{lem:42}). Finally, in section \ref{sec:43} we derive the asymptotics \eqref{gnn} of $g_n(n)$:
\[
\begin{rcases} 
\text{Lemma~\ref{lem:41}}\\
\text{Lemma~\ref{lem:42}}
\end{rcases}\!\implies\!\text{Lemma~\ref{lem:22} (a)}.
\]
\subsection{Stationary phase formula}\label{sec:41}

\begin{lemma}\label{lem:41} 
For $b\in\class^2(\cercle)$, $\eta_0\in\D{R}$ and $\mu>0$ we consider the oscillatory integral 
\begin{equation} \label{eq:oscill}
\E{I}(b,\mu,\eta_0)\coloneqq\int_0^{2\pi}\eul^{\ii\mu\cos(\eta-\eta_0)}b(\nk)\,\frac{\dd\eta}{2\pi}.  
\end{equation}
If we write 
\begin{subequations}  \label{eq:stat.phase}
\begin{equation} \label{eq:stat.phase.formula}
\E{I}(b,\mu,\eta_0)=\sum_{\kappa=\pm 1}\frac{\eul^{\ii\kappa(\mu-\pi/4)}}{\sqrt{2\pi\mu}}\,b(\kappa\eul^{\ii\eta_0})+r_b(\mu,\eta_0) 
\end{equation} 
then the remainder $r_b(\mu,\eta_0)$ satisfies the estimate
\begin{equation} \label{eq:stat.phase.remainder}
\abs{r_b(\mu,\eta_0)}\leq\frac{C_0}{\mu}\norm{b}_{\class^2(\cercle)}
\end{equation}
\end{subequations}
for some constant $C_0$. 
\end{lemma} 

\begin{proof} 
Let $\chi\in\class^\infty(\cercle)$ be such that $\chi(\eul^{\ii(\xi+\eta_0)})=1$ if $\abs{\xi}\leq\pi/4$ and $\chi(\eul^{\ii(\xi+\eta_0)})=0$ if $3\pi/4\leq\abs{\xi}\leq 5\pi/4$. Since the integrands are $2\pi$-periodic the change of variable $\eta=\xi+\eta_0$ gives
\[
\E{I}(\chi b,\mu,\eta_0)=\int_0^{2\pi}\eul^{\ii\mu\cos\xi}(\chi b)(\eul^{\ii(\xi+\eta_0)})\,\frac{\dd\xi}{2\pi}\,. 
\] 
Denoting $b_+(\xi)\coloneqq b(\eul^{\ii(\xi+\eta_0)})$ and $\chi_+(\xi)\coloneqq\chi(\eul^{\ii(\xi+\eta_0)})$, we can express
\[
\E{I}(\chi b,\mu,\eta_0)=\int_{-\pi}^{\pi}\eul^{\ii\mu\cos\xi}(\chi_+b_+)(\xi)\,\frac{\dd\xi}{2\pi}\,. 
\] 
Let $b_-(\xi)\coloneqq b(\eul^{\ii(\xi+\eta_0-\pi)})$ and $\chi_-(\xi)\coloneqq 1-\chi(\eul^{\ii(\xi+\eta_0-\pi)})$. We still have $\chi_-(\xi)=1$ if $\abs{\xi}\leq\pi/4$ and $\chi_-(\xi)=0$ if $3\pi/4\leq\abs{\xi}\leq 5\pi/4$. Then the change of variable $\eta=\xi+\eta_0-\pi$ gives  
\[
\E{I}((1-\chi)b,\mu,\eta_0)=\int_{-\pi}^{\pi}\eul^{-\ii\mu\cos\xi}(\chi_-b_-)(\xi)\,\frac{\dd\xi}{2\pi}\,.
\]
We have $\croch{-\pi,\pi}\cap\supp\chi_{\pm}\subset\croch{-\frac{3\pi}{4},\frac{3\pi}{4}}$. Next we observe that $\abs{\xi}\leq 3\pi/4$ allows us to write 
\[
b_{\pm}(\xi)=b_{\pm}(0)+q_{\pm}(\xi)\xi=b_{\pm}(0)+\tilde q_{\pm}(\xi)\sin\xi 
\]
with $\tilde q_{\pm}(\xi)\coloneqq q_{\pm}(\xi)\frac{\xi}{\sin\xi}$. Moreover, $\chi_{\pm}(0)=1$ and the standard stationary phase formula ensures 
\[
\left\lvert\int_{-\pi}^{\pi}\eul^{\pm\ii\mu\cos\xi}\chi_{\pm}(\xi)b_{\pm}(0)\,\dd\xi-\frac{\eul^{\pm\ii(\mu-\pi/4)}}{\sqrt{2\pi\mu}}b_{\pm}(0)\right\rvert\leq\frac{C_{\chi_{\pm}}}{\mu}\abs{b_{\pm}(0)}.
\]
Then writing $\eul^{\pm\ii\mu\cos\xi}\sin\xi=\frac{\pm\ii}{\mu}\partial_{\xi}\eul^{\pm\ii\mu\cos\xi}$ and integrating by parts we obtain 
\begin{equation}\label{53}
\int_{-\pi}^{\pi}\tilde q_{\pm}(\xi)\sin\xi\eul^{\pm\ii\mu\cos(\xi)} \chi_{\pm}(\xi)\,\dd\xi=\frac{\pm\ii}{\mu}\int_{-\pi}^{\pi}\eul^{\pm\ii\mu\cos\xi}\partial_{\xi}\bigl((\tilde q_{\pm}\chi_{\pm})(\xi)\bigr)\,\dd\xi\,.
\end{equation} 
Since the absolute value of the right hand side of \eqref{53} can be estimated by $\frac{C_1}{\mu}\norm{b_{\pm}}_{\class^2(\D{R})}$ the proof is complete. 
\end{proof}
 
\subsection{Approximation of $\BS{g_n(n)}$ by an oscillatory integral} \label{sec:42}
  
Recall that $g_n(k)$, $k\in\D{Z}$ is defined in \eqref{gnk} as the $k$-th diagonal entry of $\tilde V_n\coloneqq\eul^{\ii B_n}v_n(\Lambda)\eul^{-\ii B_n}$. We define $\tilde\varphi_n\colon\D{Z}\times\cercle\to\D{C}$ by 
\begin{equation}  \label{tphi}
\tilde\varphi_n(k,\ek)\coloneqq -4\left(a(n)+(k-n)\delta a(n)\right)(\sin\xi+\delta a(n)\sin 2\xi)
\end{equation} 

\begin{lemma}\label{lem:42} 
We assume that the diagonal entries of $J$ are of the form $d(k)=k+(-1)^k\rho$ and the off-diagonal entries $a(k)$ satisfy \emph{(H1)} for some $0<\gamma\leq\frac{1}{2}$. Let $g_n(k)$, $k\in\D{Z}$ be defined by \eqref{gnk}. If $\F{g}_n(k)$, $k\in\D{Z}$ is defined by  
\begin{subequations}\label{L42}
\begin{equation}\label{L42b} 
\F{g}_n(k)\coloneqq(-1)^k\rho\,\int_0^{2\pi}\eul^{\ii\tilde\varphi_n(k,\ek)}\,\frac{\dd\xi}{2\pi}, 
\end{equation}
with $\tilde\varphi_n$ as above, then 
\begin{equation}\label{L42c}
\sup_{\abs{k-n}\leq n^{\gamma}}\,\abs{g_n(k)-\F{g}_n(k)}=\ord(n^{-\gamma}\ln n). 
\end{equation}
\end{subequations}
\end{lemma} 

\begin{proof} 
As in \cite{BZ5} we denote
\[
\Theta_n\coloneqq\theta_{n,n}(\Lambda)=\theta_0(\Lambda/n-I).
\] 
Then it is easy to check the estimate $\norm{\croch{B_n,\Theta_n}}=\ord(n^{\gamma-1})$. Writing 
\[
\croch{\eul^{\pm\ii B_n},\Theta_n}=\int_0^1\eul^{\pm\ii tB_n}\,\croch{\pm\ii B_n,\Theta_n}\,\eul^{\pm\ii(1-t)B_n}\dd t
\]
we deduce 
\begin{equation}\label{42a}
\norm{\croch{\eul^{\pm\ii B_n},\Theta_n}}=\ord(n^{\gamma-1}).
\end{equation} 
We recall that $\tilde V_n=\eul^{\ii B_n}\Theta_nv(\Lambda)\Theta_n\eul^{-\ii B_n}$ and observe that \eqref{42a} ensures
\begin{equation}\label{42b}
\norm{\tilde V_n-\Theta_n\eul^{\ii B_n}v(\Lambda)\eul^{-\ii B_n}\Theta_n}=\ord(n^{\gamma-1}).
\end{equation} 
Further on we assume that $\abs{k-n}\leq n^{\gamma}$. Using $v(\Lambda)=\rho\,\eul^{\ii\pi\Lambda}$ and \eqref{42b} we obtain 
\begin{equation}\label{42c}
g_n(k)=\rho\,\scal{\vece_k,\eul^{\ii B_n}\eul^{\ii\pi\Lambda}\eul^{-\ii B_n}\vece_k}+\ord(n^{\gamma-1}).
\end{equation} 
However, $\eul^{\ii\pi\Lambda}S^{\pm 1}\eul^{-\ii\pi\Lambda}=-S^{\pm 1}$ implies $\eul^{\ii\pi\Lambda}B_n\eul^{-\ii\pi\Lambda}=-B_n$, hence  
\begin{equation}\label{42d}
\eul^{\ii B_n}\eul^{\ii\pi\Lambda}=\eul^{\ii\pi\Lambda}\eul^{-\ii B_n}.  
\end{equation} 
Using \eqref{42c} and \eqref{42d} we obtain 
\[
g_n(k)=\rho\eul^{\ii\pi k}\,\scal{\vece_k,\eul^{-2\ii B_n}\vece_k} +\ord(n^{\gamma-1})
\]
Let $Q_n^t$ be the operators introduced in \cite{BZ5}*{proof of Proposition 8.1}. For $t\in\croch{-2,2}$ and $k\in\D{N}^*$ we have 
\begin{equation} \label{Qn} 
\abs{\scal{\vece_k,\Theta_n \vece^{\ii tB_n}\vece_k}-\scal{\vece_k, Q_n^t\vece_k}}\leq\norm{\Theta_n\eul^{\ii tB_n}-Q_n^t}\leq Cn^{\gamma-1}\ln n.
\end{equation}  
Moreover,
\begin{equation} \label{Qn'} 
Q_n^t(k,k)=\theta_n(k)\int_0^{2\pi}\eul^{\ii\tilde\psi_n^t(k,\ek)}\,\frac{\dd\xi}{2\pi}
\end{equation} 
where ${\tilde\psi}_n^t$ is given by \cite{BZ5}*{(8.5a)}. Observe now that $\tilde\varphi_n$ given by \eqref{tphi} coincides with $\tilde\psi_n^t$ for $t=-2$. Thus, to complete the proof of \eqref{L42c} it suffices to use \eqref{Qn} and \eqref{Qn'} with $t=-2$.
\end{proof} 

\subsection{End of the proof of Lemma \ref{lem:22} (a)}\label{sec:43}

By definitions \eqref{L42b} and \eqref{tphi} we have 
\begin{equation*}\label{Fgjn2} 
\F{g}_n(n)=(-1)^n\rho\,\int_0^{2\pi}\eul^{-4\ii a(n)\sin\xi}\, b_n(\ek)\,\frac{\dd\xi}{2\pi},
\end{equation*}
where
\[
b_n(\ek)\coloneqq\eul^{-4\ii a(n)\delta a(n)\sin 2\xi}.
\]
Thus, using notation \eqref{eq:oscill} we can write
\[
\F{g}_n(n)=(-1)^n\rho\,\E{I}(b_n,4a(n),-\pi/2).
\]
By (H1a) and (H1b) with $0<\gamma\leq\frac{1}{2}$ we have $\norm{b_n}_{\class^2(\cercle)}=\ord(1)$. We also have $\ord(a(n)^{-1})=\ord(n^{-\gamma})$ by (H1a) and $b_n(\pm\eul^{-\ii\pi/2})=1$. Then the stationary phase formula of Lemma \ref{lem:41} gives
\begin{align*} 
\F{g}_n(n)&=(-1)^n\rho\sum_{\kappa=\pm 1}\frac{\eul^{\ii\kappa(4a(n)-\pi/4)}}{2\sqrt{2\pi a(n)}}\,b_n(\kappa\eul^{-\ii\pi/2})+\ord(a(n)^{-1})\\
&=(-1)^n\rho\frac{\cos(4a(n)-\pi/4)}{\sqrt{2\pi a(n)}}+\ord(n^{-\gamma})\\
&=\F{r}(n)+\ord(n^{-\gamma})
\end{align*}
with $\F{r}(n)$ as in \eqref{R0}. This estimate, together with the estimate
\[
\abs{g_n(n)-\F{g}_n(n)}=\ord(n^{-\gamma}\ln n)
\]
from Lemma \ref{lem:42}, gives $g_n(n)=\F{r}(n)+\ord(n^{-\gamma+\varepsilon})$ for any $\varepsilon>0$, i.e.\ estimate \eqref{gnn}, in case (a).\qed

\section{Trace estimate: a first reduction}\label{sec:4'}

In this section we reduce the proof of Proposition \ref{prop:trace} to that of Proposition~\ref{prop:4'} using the representation of functions of operators by means of Fourier transform. This representation allows us to investigate the quantity $\tilde{\C{G}}_n^0$ using  the Neumann series of  
\begin{equation}\label{Un}
U_n(t)\coloneqq\eul^{-\ii tl_n(\Lambda)}\eul^{\ii tL_n}\quad(t\in\D{R}).
\end{equation} 
The proof of this reduction is based on Lemma \ref{lem:4'} which develops our ideas from \cite{BZ5}*{Section 6} and allows us to conclude in Section \ref{sec:4'4}:
\[
\begin{rcases} 
\text{Proposition~\ref{prop:4'}}\\
\text{Lemma~\ref{lem:4'}}
\end{rcases}\!\implies\!\text{Proposition~\ref{prop:trace}}.
\]

\subsection{Properties of the evolution $\BS{U_n(t)}$}\label{sec:4'2}

Using \eqref{Un} and $L_n-l_n(\Lambda)=L_n-L_{0,n}=\tilde V_n$, we get   
\[
-\ii\,\partial_tU_n(t)=H_n(t)U_n(t)
\]
where
\begin{equation}  \label{Hnt}
H_n(t)\coloneqq\eul^{-\ii tL_{0,n}}\tilde V_n\eul^{\ii tL_{0,n}}. 
\end{equation}
The Neumann series gives the expansion  
\begin{equation*}  \label{U=}
U_n(t)=I+\ii\int_0^tH_n(t_1)\,\dd t_1+\sum_{\nu=2}^{\infty}\,\ii^{\nu}\int_0^t\!\dd t_1\dots\int_0^{t_{\nu-1}}H_n(t_1)\dots H_n(t_{\nu})\,\dd t_{\nu}.
\end{equation*} 
For $\nu\geq 1$ and $\ut=(t_1,\dots,t_{\nu})\in\D{R}^{\nu}$ we denote
\[
H_n(\ut)\coloneqq H_n(t_1)\dots H_n(t_{\nu})
\]
and
\begin{equation}  \label{gnu}
g_{\nu,n,j}(\ut)\coloneqq\ii^{\nu}H_n(\ut)(j,j).
\end{equation} 
Note that $\scal{\vece_j,H_n(t)\vece_j}=\scal{\eul^{\ii tL_{0,n}}\vece_j,\,\tilde V_n\eul^{\ii tL_{0,n}}\vece_j}=\scal{\eul^{\ii tl_{n}(j)}\vece_j,\,\tilde V_n\eul^{\ii tl_{n}(j)}\vece_j}=\scal{\vece_j,\,\tilde V_n\vece_j}$, i.e.  
\begin{equation}\label{un1}
H_n(t)(j,j)=g_n(j).
\end{equation}
For $t\in\D{R}$ we denote
\begin{equation}\label{unj} 
u_{n,j}(t)\coloneqq U_n(t)(j,j).
\end{equation} 
Then, using \eqref{un1}, we get the expansion 
\begin{equation}\label{u'nj}
\partial_tu_{n,j}(t)=\ii g_n(j)+\sum_{\nu=2}^{\infty}u_{\nu,n,j}(t),
\end{equation} 
where 
\begin{align*}
u_{2,n,j}(t)&\coloneqq-\int_0^t\scal{\vece_j,H_n(t)H_n(t_2)\vece_j}\,\dd t_2,\\ 
u_{\nu,n,j}(t)&\coloneqq\ii^{\nu}\int_0^t\!\dd t_2\dots\int_0^{t_{\nu-1}}\scal{\vece_j,\,H_n(t)H_n(t_2)\dots H_n(t_{\nu})\vece_j}\,\dd t_{\nu}\text{ for }\nu\geq 3.
\end{align*}

\begin{proposition}  \label{prop:4'}
Let $t_0>0$ and $\varepsilon>0$.

\emph{(a)}
We can find $C>0$ such that
\begin{subequations} \label{pr32}
\begin{equation}\label{pr32c} 
\sup_{\substack{\abs{j-n}\leq n^{\gamma}\\-t_0\leq t\leq t_0}}\abs{u_{2,n,j}(t)}\leq Cn^{-\gamma+5\varepsilon}.  
\end{equation} 
 
\emph{(b)}
If $\varepsilon<1/8$, then we can find $\tilde C>0$ such that the estimate
\begin{equation}\label{pr32d} 
\sup_{\abs{j-n}\leq n^{\gamma}}\,\int_{-t_0}^{t_0}\!\!\dd t_{\nu-1}\left\lvert\int_0^{t_{\nu-1}}\!\scal{\vece_j,\, H_n(t_1)\dots H_n(t_{\nu})\vece_j}\,\dd t_{\nu}\right\rvert\leq\tilde C^{\nu}n^{-\gamma+5\varepsilon}    
\end{equation}
\end{subequations}
holds whenever $3\leq\nu\leq n^{\varepsilon}$ and $t_1,\dots,t_{\nu-2}\in\croch{-t_0,t_0}$. 
\end{proposition} 

\begin{proof} 
See Section \ref{sec:63}.
\end{proof}

\begin{remark*}
In \cite{BZ5}*{estimate (6.17b)} the constant in the right hand side should be $C^{\nu}$ instead of $C$.
\end{remark*}

\subsection{Use of the Fourier transform}\label{sec:4'3}

In this section we prove the desired trace estimate provided assumption~\eqref{L4'} is satisfied, and in the next section we show that Proposition~\ref{prop:4'} precisely implies this assumption. 

\begin{lemma}\label{lem:4'} 
Let $g_n$ be defined by \eqref{gnk} and $u_{n,j}$ by \eqref{unj}. Let $\varepsilon>0$ be fixed and assume that for every $t_0>0$ one has the estimate 
\begin{equation}  \label{L4'}
\sup_{\substack{\abs{j-n}\leq n^{\gamma}\\-t_0\leq t\leq t_0}}\left\lvert\partial_tu_{n,j}(t)-\ii g_n(j)\right\rvert=\ord(n^{-\gamma+5\varepsilon}).
\end{equation} 
Let $\tilde{\C{G}}_n^0$ be defined by \eqref{G0n} by means of a function $\chi\in\C{S}(\D{R})$ whose Fourier transform has compact support. We have then the estimate 
\begin{equation}  \label{L4''}
\tilde{\C{G}}_n^0=\ord(n^{-\gamma+6\varepsilon}).
\end{equation}    
\end{lemma}  

\begin{proof} 
It consists in four steps.
\begin{stepone*}
Let $\theta_0$ be as in \eqref{220} and $\theta_{n^{\gamma},n}(s)\coloneqq\theta_0(\frac{s-n}{n^{\gamma}})$ according to \eqref{223}.  We claim that
\begin{subequations} \label{l33}
\begin{align}  \label{l33a}
\norm{(I-\theta_{n^{\gamma},n}(\tilde L_{0,n}))\chi(L_n-l_n(n))}_{{\C{B}}_1(l^2(\D{Z}))}=\ord(n^{-\gamma}),\\  
\label{l33b}
\norm{(I-\theta_{n^{\gamma},n}(\tilde L_{0,n}))\chi(\tilde L_{0,n}-l_n(n))}_{{\C{B}}_1(l^2(\D{Z}))}=\ord(n^{-\gamma}),
\end{align}
\end{subequations}  
where $\norm{T}_{{\C{B}}_1(l^2(\D{Z}))}=\tr\sqrt{T^*T}$ is the trace class norm on the algebra $\C{B}_1(l^2(\D{Z}))$ of trace class operators on $l^2(\D{Z})$. It suffices to apply \cite{BZ5}*{Proof of Lemma 6.1} with $\tilde L_{0,n}$ instead of $L_{0,n}$. 
\end{stepone*}

\begin{steptwo*}
The assertions \eqref{l33} of Step 1 ensure that  
\begin{equation}\label{dtr}
\tilde{\C{G}}_n^0-\tilde{\C{G}}_n=\ord(n^{-\gamma})
\end{equation}   
holds with 
\begin{equation*}\label{tr'}
\tilde{\C{G}}_n\coloneqq\tr\Bigl(\theta_{n^{\gamma},n}(\tilde L_{0,n})\bigl(\chi(L_n-l_n(n))-\chi(\tilde L_{0,n}-l_n(n))\bigr)\Bigr).
\end{equation*} 
Thus it remains to prove $\tilde{\C{G}}_n=\ord(n^{-\gamma+6\varepsilon})$.
\end{steptwo*}

\begin{stepthree*}
Let $t_0>0$ be such that $\supp\hat\chi\subset\croch{-t_0,t_0}$. Then the inverse Fourier formula  
\begin{equation*}\label{212} 
\chi(\lambda)=\int_{-\infty}^{\infty}\hat\chi(t)\eul^{\ii t\lambda}\dd t=\int_{-t_0}^{t_0}\hat\chi(t)\eul^{\ii t\lambda}\dd t
\end{equation*} 
allows us to express  
\[ 
\chi(L_n-l_n(n))-\chi(\tilde L_{0,n}-l_n(n))=\int_{-t_0}^{t_0}\hat\chi(t)\,\eul^{-\ii tl_n(n)}\bigl(\eul^{\ii tL_n}-\eul^{\ii t\tilde L_{0,n}}\bigr)\,\dd t
\]
and
\[  
\tilde{\C{G}}_n=\int_{-t_0}^{t_0}\hat\chi(t)\,\eul^{-\ii tl_n(n)}\,\tr\bigl(\theta_{n^{\gamma},n}(\tilde L_{0,n})\eul^{\ii tL_{0,n}}(U_n(t)-\eul^{\ii tg_n(\Lambda)})\bigr)\,\dd t.
\] 
We thus have $\tilde{\C{G}}_n=\sum_{j\in\D{Z}}\tilde{\C{G}}_n(j)$ with 
\[ 
\tilde{\C{G}}_n(j)\coloneqq\int_{-t_0}^{t_0}\hat\chi(t)\,\eul^{\ii t/2}\,\eul^{\ii t(l_n(j)-l_n(n)-1/2)}\theta_{n^{\gamma},n}(\tilde l_n(j))\bigl(u_{n,j}(t)-\eul^{\ii tg_n(j)}\bigr)\,\dd t. 
\] 
Integrating by parts as in \cite{BZ5}*{Section 6.3} we find
\[
\tilde{\C{G}}_n(j)=\ii\tilde{\C{G}}_{1,n}(j)+\ii\tilde{\C{G}}_{2,n}(j)
\]
with 
\begin{align*}
\tilde{\C{G}}_{1,n}(j)&=\int_{-t_0}^{t_0}\hat\chi(t)\,\eul^{\ii t(l_n(j)-l_n(n))}\frac{\theta_{n^{\gamma},n}(\tilde l_n(j))}{l_n(j)-l_n(n)-\tfrac{1}{2}}\,\partial_t\bigl({u_{n,j}(t)-\eul^{\ii tg_n(j)}}\bigr)\,\dd t,\\
\tilde{\C{G}}_{2,n}(j)&=\int_{-t_0}^{t_0}\partial_t\bigl(\hat\chi(t)\,\eul^{\ii t/2}\bigr)\,\eul^{\ii t(l_n(j)-l_n(n)-1/2)}\frac{\theta_{n^{\gamma},n}(\tilde l_n(j))}{l_n(j)-l_n(n)-\tfrac{1}{2}}\,\bigl(u_{n,j}(t)-\eul^{\ii tg_n(j)}\bigr)\,\dd t
\end{align*} 
and we can estimate 
\begin{align*}
\abs{\tilde{\C{G}}_{1,n}(j)}&\leq C\,\frac{\theta_{n^{\gamma},n}(\tilde l_n(j))}{1+\abs{j-n}}\,\sup_{-t_0\leq t\leq t_0}\left\lvert\partial_t\left({u_{n,j}(t)-\eul^{\ii tg_n(j)}}\right)\right\rvert,\\
\abs{\tilde{\C{G}}_{2,n}(j)}&\leq C\,\frac{\theta_{n^{\gamma},n}(\tilde l_n(j))}{1+\abs{j-n}}\,\sup_{-t_0\leq t\leq t_0}\left\lvert u_{n,j}(t)-\eul^{\ii tg_n(j)}\right\rvert.
\end{align*}
\end{stepthree*}

\begin{stepfour*}[last step]  
Since $\abs{\tilde l_n(j)-j}\leq C$ we can find $n_0$ such that $\theta_{n^{\gamma},n}(\tilde l_n(j))\neq 0$ implies $\abs{j-n}\leq n^{\gamma}$ 
for $n\geq n_0$. Combining this fact with  
\[
\sup_{-t_0\leq t\leq t_0}\left\lvert u_{n,j}(t)-\eul^{\ii tg_n(j)}\right\rvert\leq\abs{t_0}\sup_{-t_0\leq t\leq t_0}\left\lvert\partial_t\left(u_{n,j}(t)-\eul^{\ii tg_n(j)}\right)\right\rvert  
\] 
we can estimate 
\begin{equation}\label{Gn} 
\abs{\tilde{\C{G}}_n}\leq\sum_{\abs{j-n}\leq n^{\gamma}}\frac{C_0}{1+\abs{j-n}}\,\sup_{-t_0\leq t\leq t_0}\left\lvert\partial_t\left(u_{n,j}(t)-\eul^{\ii tg_n(j)}\right)\right\rvert.
\end{equation}  
By \eqref{233b} we have the estimate $\left\lvert\eul^{\ii tg_n(j)}-1\right\rvert\leq\abs{tg_n(j)}\leq C\abs{t}n^{-\gamma/2}$ for $\abs{j-n}\leq n^{\gamma}$, and thus   
\begin{equation}\label{egnk2} 
\sup_{\substack{\abs{j-n}\leq n^{\gamma}\\-t_0\leq t\leq t_0}}\left\lvert\partial_t\eul^{\ii tg_n(j)}-\ii g_n(j)\right\rvert=\ord(n^{-\gamma}).
\end{equation}  
Therefore combining assumption \eqref{L4'} with \eqref{egnk2} we find 
\begin{equation}\label{4'}  
\sup_{\substack{\abs{j-n}\leq n^{\gamma}\\-t_0\leq t\leq t_0}}\left\lvert  \partial_t\left(u_{n,j}(t)-\eul^{\ii tg_n(j)}\right)\right\rvert=\ord(n^{-\gamma+5\varepsilon}). 
\end{equation}   
Using
\[
\sum_{\abs{k}\leq n^{\gamma}}\frac{1}{1+\abs{k}}\leq 1+2\ln(n)=\ord(n^{\varepsilon})
\] 
and \eqref{4'} we can estimate the r.h.s.\ of \eqref{Gn} by $\ord(n^{-\gamma+6\varepsilon})$, hence \eqref{L4''} follows from \eqref{dtr}.\qedhere
\end{stepfour*}
\end{proof}
  
\subsection{Proof of Proposition~\ref{prop:4'}$\BS\implies$Proposition~\ref{prop:trace}}\label{sec:4'4}
  
Since $\varepsilon>0$ is arbitrary, it is clear that it suffices to prove $\tilde{\C{G}}_n^0=\ord(n^{-\gamma+6\varepsilon})$ instead of \eqref{tr}. Thus by Lemma \ref{lem:4'} it only remains to check that the assertions of Proposition \ref{prop:4'} imply estimate \eqref{L4'}, i.e.,
\[
\sup_{\substack{\abs{j-n}\leq n^{\gamma}\\-t_0\leq t\leq t_0}}\left\lvert\partial_tu_{n,j}(t)-\ii g_n(j)\right\rvert=\ord(n^{-\gamma+5\varepsilon}).
\]
We first note that \eqref{u'nj} and \eqref{un1} give the expansion
\begin{equation} \label{340}
\partial_tu_{n,j}(t)-\ii g_n(j)=\sum_{\nu=2}^{\infty}u_{\nu,n,j}(t).
\end{equation}
We then observe that \eqref{pr32d} for $\nu=3$ yields
\begin{equation}  \label{nu=3}
\sup_{\substack{\abs{j-n}\leq n^{\gamma}\\-t_0\leq t\leq t_0}}\abs{u_{3,n,j}(t)}\leq\tilde C^3n^{-\gamma+5\varepsilon}.   
\end{equation}
For any $4\leq\nu<n^{\varepsilon}$ and $t\in\croch{-t_0,t_0}$, estimate \eqref{pr32d} gives 
\begin{equation}  \label{nu>3}
\sup_{\abs{j-n}\leq n^{\gamma}}\abs{u_{\nu,n,j}(t)}\leq\tilde C^{\nu}n^{-\gamma+5\varepsilon}\int_{\Delta_t}\!\dd t_2\dots\int_{\Delta_{t_{\nu-3}}}\!\!\dd t_{\nu-2}=\tilde C^{\nu}n^{-\gamma+5\varepsilon}\frac{\abs{t}^{\nu-3}}{(\nu-3)!}\,,
\end{equation}
where $\Delta_t\coloneqq\croch{0,t}$ when $t\geq 0$ and $\croch{t,0}$ when $t\leq 0$. Therefore, by using \eqref{pr32c}, \eqref{nu=3}, and \eqref{nu>3} we get  
\begin{subequations} \label{34}
\begin{equation}\label{34a}
\sup_{\substack{\abs{j-n}\leq n^{\gamma}\\-t_0\leq t\leq t_0}}\sum_{2\leq\nu<n^{\varepsilon}}\abs{u_{\nu,n,j}(t)}\leq Cn^{-\gamma+5\varepsilon}+\sum_{3\leq\nu<n^{\varepsilon}}\frac{\tilde C^{\nu}t_0^{\nu-3}n^{-\gamma+5\varepsilon}}{(\nu-3)!}\leq\bigl(C+\tilde C^3\eul^{\tilde Ct_0}\bigr)n^{-\gamma+5\varepsilon}.
\end{equation}
To complete the proof it remains to consider indices $\nu\geq n^{\varepsilon}$. We observe that $\norm{H_n(t)}=\norm{\tilde V_n}=\norm{V_n}\leq\rho_N$. Therefore,
\[
\abs{\scal{\vece_j,H_n(t)H_n(t_2)\dots H_n(t_\nu)\vece_j}}\leq\rho_N^{\nu}
\]
and 
\[
\abs{u_{\nu,n,j}(t)}\leq\int_{\Delta_t}\!\dd t_2\dots\int_{\Delta_{t_{\nu-1}}}\!\!\abs{\scal{\vece_j,\,H_n(t)H_n(t_2)\dots H_n(t_{\nu})\vece_j}}\,\dd t_{\nu}\leq\rho_N^{\nu}\frac{\abs{t}^{\nu-1}}{(\nu-1)!}.
\]
We thus get
\begin{equation}\label{34a'} 
\sup_{\substack{\abs{j-n}\leq n^{\gamma}\\-t_0\leq t\leq t_0}}\sum_{\nu\geq n^{\varepsilon}}\abs{u_{\nu,n,j}(t)}\leq\sum_{\nu\geq n^{\varepsilon}}\frac{\rho_N^{\nu}t_0^{\nu-1}}{(\nu-1)!}\leq\frac{\rho_N\eul^{\rho_Nt_0}(\rho_Nt_0)^{\lfloor n^{\varepsilon}\rfloor-1}}{(\lfloor n^{\varepsilon}\rfloor-1)!}=\ord(n^{-m}) 
\end{equation} 
\end{subequations}
for any integer $m$. Estimates \eqref{34} with \eqref{340} show that assumption \eqref{L4'} in Lemma~\ref{lem:4'} is valid.\qed
\section{Estimate of some oscillatory integrals}\label{sec:5}
\subsection{Main result}\label{sec:51}

In this section we consider oscillatory integrals of the following type: 
\begin{equation}\label{51a} 
\C{J}(b,t_1,t_2,\zeta,\mu)\coloneqq\int_{t_1}^{t_2}\frac{\eul^{\ii\mu\sqrt{4\sin^2(t/2)+\zeta^2}}}{(4\sin^2(t/2)+\zeta^2)^{1/4}}\, b(t)\,\dd t,
\end{equation}
where $b\in\class^1(\D{R})$, $\mu\in\D{R}^*$, $\zeta\geq 0$ and $t_1\leq t_2$ are real numbers. Notice that in the case $\zeta=0$ the function $t\to\abs{2\sin(t/2)}^{-1/2}$ is Lebesgue integrable 
on any bounded interval of $\D{R}$ and so \eqref{51a} is still well defined. 

\begin{lemma}\label{lem:51} 
Let $\Delta_0\subset\D{R}$ be a bounded interval. Then there is a constant $C_{\Delta_0}>0$ such that for any interval $\croch{t_1,t_2}\subset\Delta_0$, $\mu\in\D{R}^*$, $\zeta\geq 0$ and $b\in\class^1(\D{R})$ one has the estimate 
\begin{equation}\label{51c}  
\abs{\C{J}(b,t_1,t_2,\zeta,\mu)}\leq C_{\Delta_0}\frac{1+\sqrt{\zeta}}{\sqrt{\abs{\mu}}}\C{M}(b,\croch{t_1,t_2}),  
\end{equation}  
where for any bounded interval $\Delta\subset\D{R}$
\begin{equation}\label{Mcal}  
\C{M}(b,\Delta)\coloneqq\sup_{t\in\Delta}\abs{b(t)}+\int_{\Delta}\abs{b'(t)}\dd t.   
\end{equation}  
\end{lemma}

This lemma is used in the proof of Proposition \ref{prop:4'}. It serves in Section~\ref{sec:634} to prove estimate \eqref{644a} of some oscillatory integral.

\begin{proof} 
The proof is given in the next subsections. It is based on van der Corput lemma. 
\end{proof} 

\begin{lemma}[van der Corput]\label{lem:52}
Assume that $h_0\colon(t_1,t_2)\to\D{R}$ is smooth and its second derivative satisfies $h_0''(t)\geq c_0$ for $t_1<t<t_2$ and some constant $c_0>0$. Assume also that $\mu_0\in\D{R}^*$ and $b_0\in\class^1\bigl((t_1,t_2)\bigr)$, and consider the oscillatory integral
\begin{equation}\label{vdc}
\E{J}(b_0,t_1,t_2,h_0,\mu_0)\coloneqq\int_{t_1}^{t_2}\eul^{\ii\mu_0 h_0(t)}b_0(t)\,\dd t.
\end{equation}
Then there is a constant $C_0$ depending only on $c_0$ such that we have the estimate     
\begin{equation*}\label{51e}
\left\lvert\E{J}(b_0,t_1,t_2,h_0,\mu_0)\right\rvert\leq\frac{C_0}{\sqrt{\abs{\mu_0}}}\left(\abs{b_0(t_1)}+\int_{t_1}^{t_2}\abs{b_0'(t)}\,\dd t\right)\leq\frac{C_0}{\sqrt{\abs{\mu_0}}}\,\C{M}(b_0,\croch{t_1,t_2}).
\end{equation*}
\end{lemma}

\begin{proof} 
See \cite{St}*{Section VIII.1.2, p.~354}. 
\end{proof} 

\subsection{Proof of Lemma \ref{lem:51} in case $\BS{\Delta_0=\croch{0,2\pi/3}}$}\label{sec:52} 
\subsubsection{Change of variable}\label{sec:521} 

By our assumption $\croch{t_1,t_2}\subset\croch{0,2\pi/3}$ and the change of variable 
\begin{equation}\label{52a}
t=2\arcsin(s/2) 
\end{equation}
parametrize $\croch{0,2\pi/3}$ by $s\in\croch{0,\sqrt{3}}$. If $t\in\croch{t_1,t_2}$ then $s\in\croch{s_1,s_2}$ where $s_i\coloneqq 2\sin(t_i/2)$, $i=1,2$. In particular, $\sqrt{4\sin^2(t/2)+\zeta^2}=\sqrt{s^2+\zeta^2}$ and the change of variable \eqref{52a} gives
\begin{equation*}\label{52b} 
\C{J}(b,t_1,t_2,\zeta,\mu)=\C{J}_1(b_1,s_1,s_2,\zeta,\mu),
\end{equation*}  
where  
\begin{equation*}\label{52b'} 
\C{J}_1(b_1,s_1,s_2,\zeta,\mu)\coloneqq\int_{s_1}^{s_2}\frac{\eul^{\ii\mu\sqrt{s^2+\zeta^2}}}{(s^2+\zeta^2)^{1/4}}\,b_1(s)\,\dd s
\end{equation*}
with
\begin{equation*}\label{52a'}
b_1(s)\coloneqq\frac{b(2\arcsin(s/2))}{\sqrt{1-s^2/4}}\,.
\end{equation*} 
Since there is a constant $C_0$ such that $\C{M}(b_1,\croch{s_1,s_2})\leq C_0\C{M}(b,\croch{t_1,t_2})$, to get \eqref{51c} for any interval $\croch{t_1,t_2}\subset\croch{0,2\pi/3}$, $\mu\in\D{R}^*$, $\zeta\geq 0$, and $b\in\class^1(\croch{t_1,t_2})$ it suffices to prove the following

\begin{statement*}
There is a constant $C>0$ such that for any $\mu\in\D{R}^*$, $\zeta\geq 0$, $\croch{s_1,s_2}\subset\croch{0,\sqrt{3}}$, and $b_1\in\class^1(\croch{s_1,s_2})$ we have the estimate
\begin{equation}\label{52c} 
\abs{\C{J}_1(b_1,s_1,s_2,\zeta,\mu)}\leq\frac{C}{\sqrt{\abs{\mu}}}\,\C{M}(b_1,\croch{s_1,s_2}).  
\end{equation}
\end{statement*}

To prove this statement we distinguish three cases: $\zeta\geq s_2$, $\zeta\leq s_1$, and $s_1\leq\zeta\leq s_2$.

\subsubsection{Proof of \eqref{52c} in case $\zeta\geq s_2$}\label{sec:522} 

Using the notations 
\begin{equation*}\label{52d}
h_1(s,\zeta)\coloneqq\zeta\sqrt{\zeta^2+s^2}  
\end{equation*} 
and \eqref{vdc} from Lemma \ref{lem:52} we can write  
\begin{align*} 
\C{J}_1(b_1,s_1,s_2,\zeta,\mu)&=\sqrt{\zeta}\int_{s_1}^{s_2}\eul^{\ii\mu\zeta^{-1}h_1(s,\zeta)}h_1(s,\zeta)^{-1/2}b_1(s)\,\dd s\\
&=\sqrt{\zeta}\times\E{J}\bigl(b_1\tilde h_1(\,\cdot\,,\zeta),s_1,s_2,h_1(\,\cdot\,,\zeta),\mu\zeta^{-1}\bigr)
\end{align*}
where
\begin{equation*}\label{52e'}
\tilde h_1(s,\zeta)\coloneqq h_1(s,\zeta)^{-1/2}=\frac{1}{\sqrt{\zeta}(\zeta^2+s^2)^{1/4}}\,.  
\end{equation*} 
Next we observe that $\zeta\geq s$ ensures 
\[
\partial^2_sh_1(s,\zeta)=\frac{\zeta^3}{(s^2+\zeta^2)^{3/2}}\geq\frac{\zeta^3}{(2\zeta^2)^{3/2}}=2^{-3/2}
\] 
and thus we can apply Lemma \ref{lem:52}. It gives the estimate   
\begin{equation*}\label{52e}  
\abs{\C{J}_1(b_1,s_1,s_2,\zeta,\mu)}\leq\frac{C_0\zeta}{\sqrt{\abs{\mu}}}\C{M}\bigl(b_1\tilde h_1(\,\cdot\,,\zeta),\croch{s_1,s_2}\bigr).  
\end{equation*}
Thus to get \eqref{52c} it suffices to show 
\begin{equation*}\label{52f}  
\C{M}\bigl(b_1\tilde h_1(\,\cdot\,,\zeta),\croch{s_1,s_2}\bigr)\leq 2\frac{\C{M}(b_1,\croch{s_1,s_2})}{\zeta}\,. 
\end{equation*} 
For this purpose we first observe that $h_1(s,\zeta)\coloneqq\zeta\sqrt{\zeta^2+s^2}\geq\zeta^2$. Then 
\begin{equation}\label{52g}
0<\tilde h_1(s,\zeta)\leq\frac{1}{\zeta}\,, 
\end{equation} 
hence 
\begin{equation*}\label{52z}
\sup_{s_1\leq s\leq s_2}\abs{b_1(s)\tilde h_1(s,\zeta)}\leq\sup_{s_1\leq s\leq s_2}\frac{\abs{b_1(s)}}{\zeta}\leq\frac{\C{M}(b_1,\croch{s_1,s_2})}{\zeta}\,. 
\end{equation*}
Next we claim that 
\begin{equation}\label{52g'}  
\int_{s_1}^{s_2}\abs{\partial_s\tilde h_1(s,\zeta)}\,\dd s\leq\frac{1}{\zeta}\,.
\end{equation}
Indeed, since $\partial_s\tilde h_1(s,\zeta)\leq 0$, we can estimate the left hand side of \eqref{52g'} using \eqref{52g}:
\[
\int_{s_1}^{s_2}(-\partial_s\tilde h_1)(s,\zeta)\,\dd s=\tilde h_1(s_1,\zeta)-\tilde h_1(s_2,\zeta)\leq\frac{1}{\zeta}\,.
\]
Finally it remains to show 
\begin{equation}\label{52h}  
\int_{s_1}^{s_2}\abs{\partial_s(b_1\tilde h_1)(s,\zeta)}\,\dd s\leq\frac{\C{M}(b_1,\croch{s_1,s_2})}{\zeta}\,.
\end{equation} 
The left hand side of \eqref{52h} is indeed $\int_{s_1}^{s_2}\abs{b_1'(s)\tilde h_1(s,\zeta)+b_1(s)\partial_s\tilde h_1(s,\zeta)}\,\dd s$. Thus, using \eqref{52g} and \eqref{52g'} we get 
\begin{align*}  
\int_{s_1}^{s_2}\abs{\partial_s(b_1\tilde h_1)(s,\zeta)}\,\dd s&\leq\sup_{s_1\leq s\leq s_2}\abs{\tilde h_1(s,\zeta)}\int_{s_1}^{s_2}\abs{b_1'(s)}\,\dd s+\sup_{s_1\leq s\leq s_2}\abs{b_1(s)}\int_{s_1}^{s_2}\abs{\partial_s\tilde h_1(s,\zeta)}\,\dd s\\
&\leq\frac{1}{\zeta}\Bigl(\,\sup_{s_1\leq s\leq s_2}\abs{b_1(s)}+\int_{s_1}^{s_2}\abs{b_1'(s)}\,\dd s\Bigr)=\frac{\C{M}(b_1,\croch{s_1,s_2})}{\zeta}\,.  
\end{align*} 

\subsubsection{Proof of \eqref{52c} in case $\zeta\leq s_1$}\label{sec:523} 

We denote $\tilde s_i\coloneqq(s_i^2+\zeta^2)^{1/4}$, $i=1,2$ and consider the change of variable 
\begin{equation*}\label{53a} 
s=\sqrt{\tilde s^4-\zeta^2}\text{ for }\tilde s\in\croch{\tilde s_1,\tilde s_2}
\end{equation*}  
which gives $\sqrt{s^2+\zeta^2}=\tilde s^2$. By applying this change of variable to the integral
\[
\C{J}_1(b_1,s_1,s_2,\zeta,\mu)=\int_{s_1}^{s_2}\eul^{\ii\mu\sqrt{s^2+\zeta^2}}(s^2+\zeta^2)^{-1/4}b_1(s)\,\dd s
\]
we find
\[
\C{J}_1(b_1,s_1,s_2,\zeta,\mu)=\int_{\tilde s_1}^{\tilde s_2}\eul^{\ii\mu\tilde s^2}\,b_2(\tilde s,\zeta)\tilde h_2(\tilde s,\zeta)\,\dd\tilde s=\E{J}\bigl(b_2\tilde h_2(\,\cdot\,,\zeta),\tilde s_1,\tilde s_2,h_0,\mu\bigr),
\]
with $b_2(\tilde s,\zeta)\coloneqq b_1(\sqrt{\tilde s^4-\zeta^2})$, $h_0(\tilde s)=\tilde s^2$ and
\begin{equation*}\label{53c'}
\tilde h_2(\tilde s,\zeta)\coloneqq\frac{1}{\tilde s}\,\partial_{\tilde s}(\sqrt{\tilde s^4-\zeta^2})=\frac{2\tilde s^2}{\sqrt{\tilde s^4-\zeta^2}}.  
\end{equation*}     
Then Lemma \ref{lem:52} applies and gives
\begin{equation}\label{53c} 
\abs{\C{J}_1(b_1,s_1,s_2,\zeta,\mu)}\leq\frac{C_0}{\sqrt{\abs{\mu}}}\C{M}(b_2\tilde h_2,\croch{\tilde s_1,\tilde s_2}).  
\end{equation}
We observe that $s\geq s_1\geq\zeta$ ensures $\tilde s\geq(2\zeta^2)^{1/4}$, hence $\zeta\leq\tilde s^2/\sqrt{2}$. Using the fact that $\tilde h_2(s,\zeta)$ is increasing with $\zeta$ we have the estimate        
\begin{equation}\label{53e}
s\geq s_1\geq\zeta\implies\tilde h_2(\tilde s,\zeta)\leq\tilde h_2(\tilde s,\tilde s^2/\sqrt{2})=2\sqrt{2}.    
\end{equation}        
Moreover, $\partial_{\tilde s}\tilde h_2(\tilde s,\zeta)=-4\zeta^2\tilde s(\tilde s^4-\zeta^2)^{-3/2}\leq 0$ and as in the proof of estimate \eqref{52g'} we find
\begin{equation}\label{53f}
\int_{\tilde s_1}^{\tilde s_2}\abs{\partial_{\tilde s}\tilde h_2(\tilde s,\zeta)}\,\dd\tilde s=\tilde h_2(\tilde s_1,\zeta)-\tilde h_2(\tilde s_2,\zeta)\leq 2\sqrt{2}.
\end{equation}
Estimate \eqref{53c} leads to the desired estimate \eqref{52c} if we prove\begin{equation*}\label{53g}  
\C{M}(b_2\tilde h_2,\croch{\tilde s_1,\tilde s_2})\leq 4\sqrt{2}\C{M}(b_2,\croch{\tilde s_1,\tilde s_2})=4\sqrt{2}\C{M}(b_1,\croch{s_1,s_2}). 
\end{equation*} 
The last equality is easy. For the inequality we first observe that by \eqref{53e} 
\[
\sup_{\tilde s_1\leq\tilde s\leq\tilde s_2}\abs{b_2\tilde h_2(\tilde s,\zeta)}\leq 2\sqrt{2}\sup_{\tilde s_1\leq\tilde s\leq\tilde s_2}\abs{b_2(\tilde s,\zeta)}\leq 2\sqrt{2}\C{M}(b_2,\croch{\tilde s_1,\tilde s_2}).
\]
Moreover, using \eqref{53e} and \eqref{53f} we find
\begin{align*}  
\int_{\tilde s_1}^{\tilde s_2}\abs{\partial_{\tilde s}(b_2\tilde h_2)(\tilde s,\zeta)}\,\dd\tilde s&\leq\sup_{\tilde s_1\leq\tilde s\leq\tilde s_2}\abs{\tilde h_2(\tilde s,\zeta)}\,\int_{\tilde s_1}^{\tilde s_2}\abs{\partial_{\tilde s}b_2(\tilde s,\zeta)}\,\dd\tilde s+\sup_{\tilde s_1<\tilde s<\tilde s_2}\abs{b_2(\tilde s,\zeta)}\,\int_{\tilde s_1}^{\tilde s_2}\abs{\partial_{\tilde s}\tilde h_2(\tilde s,\zeta)}\,\dd\tilde s\\
&\leq 2\sqrt{2}\int_{\tilde s_1}^{\tilde s_2}\abs{\partial_{\tilde s}b_2(\tilde s,\zeta)}\,\dd\tilde s+2\sqrt{2}\sup_{\tilde s_1<\tilde s<\tilde s_2}\abs{b_2(\tilde s,\zeta)}\\
&=2\sqrt{2}\C{M}(b_2,\croch{\tilde s_1,\tilde s_2}).
\end{align*} 

\subsubsection{Proof of \eqref{52c} in case $s_1<\zeta<s_2$}\label{sec:524} 

This case reduces to the previous ones. For $s_1<\zeta<s_2$ we indeed have 
\begin{align*}
&\C{J}_1(b_1,s_1,s_2,\zeta,\mu)=\C{J}_1(b_1,s_1,\zeta,\zeta,\mu)+\C{J}_1 (b_1,\zeta,s_2,\zeta,\mu),\\ 
&\C{M}(b_1,\croch{s_1,\zeta})+\C{M}(b_1,\croch{\zeta,s_2})\leq 2\C{M}(b_1,\croch{s_1,s_2}).
\end{align*} 

\subsection{Proof of Lemma \ref{lem:51} in case $\BS{\Delta_0=\croch{2\pi/3,\pi}}$}\label{sec:53} 

Denoting  
\[
h(t,\zeta)\coloneqq\sqrt{4\sin^2(t/2)+\zeta^2}
\]
we can write
\begin{align*}
\C{J}_1(b,t_1,t_2,\zeta,\mu)&\coloneqq\int_{t_1}^{t_2}\frac{\eul^{\ii\mu\sqrt{4\sin^2(t/2)+\zeta^2}}}{(4\sin^2(t/2)+\zeta^2)^{1/4}}\,b(t)\,\dd t=\int_{t_1}^{t_2}\eul^{\ii\mu h(t,\zeta)}\frac{b(t)}{h(t,\zeta)^{1/2}}\,\dd t\\
&=\int_{t_1}^{t_2}\eul^{\ii\frac{\mu}{1+\zeta}(1+\zeta)h(t,\zeta)}\,\tilde b(t,\zeta)\,\dd t=\E{J}\bigl(\tilde b,t_1,t_2,\tilde h,\tilde\mu\bigr),
\end{align*}
where $\tilde b(t,\zeta)\coloneqq b(t)h(t,\zeta)^{-1/2}$, $\tilde\mu\coloneqq\mu/(1+\zeta)$, $\tilde h(t,\zeta)\coloneqq(1+\zeta)h(t,\zeta)$, and $\E{J}$ is as in \eqref{vdc}. We get
\[
\partial_th(t,\zeta)=\frac{\sin t}{h(t,\zeta)}\,,\quad -\partial_t^2h(t,\zeta)=-\frac{\cos t}{h(t,\zeta)}+\frac{\sin^2t}{h(t,\zeta)^3}.
\]
Since $2\pi/3\leq t\leq\pi$ implies $1/2\leq-\cos t\leq 1$, we get  
\[   
2\pi/3\leq t\leq\pi\implies-\partial_t^2\tilde h(t,\zeta)\geq-\frac{(1+\zeta)\cos t}{h(t,\zeta)}\geq\frac{1+\zeta}{2\sqrt{4+\zeta^2}}\geq\frac{1}{4}\,. 
\]
Hence, Lemma \ref{lem:52} applies to estimate $\E{J}\bigl(\tilde b,t_1,t_2,\tilde h,\tilde\mu\bigr)$. To get \eqref{51c} it suffices to show that there exists a constant $\tilde C$, independent of $\zeta$, $t_1$, $t_2$ such that
\[
\C{M}(\tilde b,\croch{t_1,t_2})\leq\tilde C\C{M}(b,\croch{t_1,t_2}).
\]
For $2\pi/3\leq t\leq\pi$ we have $2\sin(t/2)\geq\sqrt{3}$ so that $h(t,\zeta)\geq\sqrt{3+\zeta^2}\geq\sqrt{3}$. Hence,
\[
\sup_{t_1\leq t\leq t_2}\abs{\tilde b(t,\zeta)}=\sup_{t_1\leq t\leq t_2}\frac{\abs{b(t)}}{\sqrt{h(t,\zeta)}}\leq\frac{1}{\sqrt[4]{3}}\sup_{t_1\leq t\leq t_2}\abs{b(t)}.
\]
Moreover,
\[
\abs{\partial_t\tilde b(t,\zeta)}=\left\lvert\frac{b'(t)}{h(t,\zeta)^{1/2}}-\frac{b(t)\sin t}{2h(t,\zeta)^{5/2}}\right\rvert\leq\frac{1}{\sqrt[4]{3}}\abs{b'(t)}+\frac{1}{2\sqrt[4]{3^5}}\abs{b(t)}\leq\frac{1}{\sqrt[4]{3}}(\abs{b'(t)}+\abs{b(t)}).
\]
Thus, using $t_2-t_1\leq\pi/3$,
\[
\int_{t_1}^{t_2}\abs{\partial_t\tilde b(t,\zeta)}\,\dd t\leq\frac{1}{\sqrt[4]{3}}\int_{t_1}^{t_2}\abs{b'(t)}\,\dd t+\frac{\pi}{3\sqrt[4]{3}}\sup_{t_1\leq t\leq t_2}\abs{b(t)}.
\]
Summarizing,
\[
\C{M}(\tilde b,\croch{t_1,t_2})\leq\frac{1}{\sqrt[4]{3}}\Bigl(1+\frac{\pi}{3}\Bigr)\C{M}(b,\croch{t_1,t_2}).
\]

\subsection{Proof of Lemma \ref{lem:51}: last steps}\label{sec:54} 
\subsubsection{Proof of \eqref{51c} in case $\Delta_0=\croch{0,\pi}$}\label{sec:541} 

By Sections \ref{sec:52} and \ref{sec:53} we know that \eqref{51c} holds if $2\pi/3\notin(t_1,t_2)$. The remaining case $t_1<2\pi/3<t_2$ can be deduced from the previous ones by using the additivity properties
\begin{subequations} \label{54} 
\begin{align}\label{54'}
&\C{J}(b,t_1,t_2,\zeta,\mu)=\C{J}(b,t_1,t_*,\zeta,\mu)+\C{J}(b,t_*,t_2,\zeta,\mu),\\ 
\label{54''}
&\C{M}(b,\croch{t_1,t_*})+\C{M}(b,\croch{t_*,t_2})\leq 2\C{M}(b,\croch{t_1,t_2}) 
\end{align}
\end{subequations}
with $t_*=2\pi/3$.

\subsubsection{Proof of \eqref{51c} in case $\Delta_0=\croch{-\pi,0}$}\label{sec:542} 

This case reduces to the previous one by using the symmetry $t\leadsto -t$. We indeed have $\C{J}(b,t_1,t_2,\zeta,\mu)=\C{J}(\check{b},-t_2,-t_1,\zeta,\mu)$ and $\C{M}(b,\croch{t_1,t_2})=\C{M}(\check{b},\croch{-t_2,-t_1})$ where $\check{b}(t)=b(-t)$.

\subsubsection{Proof of \eqref{51c} in case $\Delta_0=\croch{-\pi,\pi}$}\label{sec:543} 

By Sections \ref{sec:541} and \ref{sec:542} we know that \eqref{51c} holds if $0\notin(t_1,t_2)$. The remaining case $0\in(t_1,t_2)$ can be deduced from the two previous ones by using \eqref{54} with $t_*=0$.

\subsubsection{Proof of \eqref{51c} for $\Delta_0=\croch{(2k-1)\pi,(2k+1)\pi}$, $k\in\D{Z}$}\label{sec:544}
 
This case reduces to the previous one by translation $t\mapsto\hat t\coloneqq t-2k\pi$. We indeed have $\C{J}(b,t_1,t_2,\zeta,\mu)=\C{J}(\hat b,\hat t_1,\hat t_2,\zeta,\mu)$ and $\C{M}(b,\croch{t_1,t_2})=\C{M}(\hat b,\croch{\hat t_1,\hat t_2})$ where $\hat b(t)\coloneqq b(t+2k\pi)$.

\subsubsection{Proof of \eqref{51c} for arbitrary $\Delta_0$}\label{sec:545}

We know \eqref{51c} holds if $(t_1,t_2)\cap(2\D{Z}+1)\pi=\emptyset$. If $(t_1,t_2)\cap(2\D{Z}+1)\pi=\accol{t_*^1,t_*^2,\dots,t_*^k}$ where $t_*^2=t_*^1+2\pi,\dots,t_*^k=t_*^1+2(k-1)\pi$, $k\in\D{Z}$ then \eqref{51c} follows by using repeatedly properties \eqref{54}:
\begin{align*}
&\C{J}(b,t_1,t_2,\zeta,\mu)=\C{J}(b,t_1,t_*^1,\zeta,\mu)+\C{J}(b,t_*^1,t_*^2,\zeta,\mu)+\dots+\C{J}(b,t_*^k,t_2,\zeta,\mu),\\ 
&\C{M}(b,\croch{t_1,t_*^1})+\C{M}(b,\croch{t_*^1,t_*^2})+\dots+\C{M}(b,\croch{t_*^k,t_2})\leq(k+1)\C{M}(b,\croch{t_1,t_2}).\tag*{\qed}
\end{align*}
 
\section{Approximation by oscillatory integrals} \label{sec:6} 
\subsection{Decomposition of $\BS{H_n(\ut)}$ into components $\BS{H_n^{\uw,\ut}}$}\label{sec:60}

We assume $\alpha_0=0$ in \eqref{vk} and denote $\Omega^*=\accol{2\pi m/N}_{m=1}^{N-1}$. Thus, we can expand 
\[ 
v(\Lambda)=\sum_{\omega\in\Omega^*}c_{\omega}\eul^{\ii\omega\Lambda}    
\] 
where $c_{\omega}\in\D{C}$ are constants. Since $H_n(t)\coloneqq\eul^{-\ii tL_{0,n}}\tilde V_n\eul^{\ii tL_{0,n}}$ with $\tilde V_n=\eul^{\ii B_n}(\theta_{n,n}^2v)(\Lambda)\eul^{-\ii B_n}$ we can expand $H_n(t)$ as follows:
\[ 
H_n(t)=\eul^{-\ii tL_{0,n}}\eul^{\ii B_n}(\theta_{n,n}^2v)(\Lambda)\eul^{-\ii B_n}\eul^{\ii tL_{0,n}}=\sum_{\omega\in\Omega^*}c_{\omega}H_n^{\omega,t}     
\] 
with 
\[ 
H_n^{\omega,t}\coloneqq\eul^{-\ii tL_{0,n}}\eul^{\ii B_n} \theta_{n,n}^2(\Lambda)\eul^{\ii\omega\Lambda}\eul^{-\ii B_n}\eul^{\ii tL_{0,n}}. 
\]  
More generally, for any integer $\nu\geq 1$, $\uw=(\omega_1,\dots,\omega_{\nu})\in(\Omega^*)^{\nu}$, and $\ut=(t_1,\dots,t_{\nu})\in\D{R}^{\nu}$,
\begin{equation} \label{601}
H_n(\ut)\coloneqq H_n(t_1)\dots H_n(t_{\nu})=\sum_{\uw\in(\Omega^*)^{\nu}}c_{\uw}H_n^{\uw,\ut}
\end{equation}
where $c_{\uw}\coloneqq c_{\omega_1}\dots c_{\omega_{\nu}}$ and
\begin{equation}   \label{60}
H_n^{\uw,\ut}\coloneqq H_n^{\,\omega_1,t_1}\dots H_n^{\,\omega_{\nu},t_{\nu}}.
\end{equation} 
Let $g_{\nu,n,j}(\ut)\coloneqq\ii^{\nu}H_n(\ut)(j,j)$ be as in \eqref{gnu}. Using \eqref{601} we can expand $g_{\nu,n,j}(\ut)$ as follows:
\begin{subequations}\label{61}
\begin{equation}\label{61a}
g_{\nu,n,j}(\ut)=\ii^{\nu}\sum_{\uw\in(\Omega^*)^{\nu}}\!c_{\uw}\,g_n^{\uw,\ut}(j),
\end{equation} 
where
\begin{equation}\label{61a'}
g_n^{\uw,\ut}(j)\coloneqq H_n^{\uw,\ut}(j,j).
\end{equation} 
\end{subequations}
\subsection{Approximation of $\BS{g_n^{\uw,\ut}(j)}$ by an oscillatory integral}\label{sec:61} 

For $j\in\croch{n-n^{\gamma},n+n^{\gamma}}$ we approximate the $j$-th diagonal entry $g_n^{\uw,\ut}(j)$ of $H_n^{\uw,\ut}$ by an oscillatory integral of type
\begin{equation}\label{61b}
\F{g}_n^{\uw,\ut}(j)\coloneqq\eul^{\ii j\abs{\uw}_1}\int_0^{2\pi}\eul^{\ii\psi_{n,1}^{\uw,\ut\,}(\nk)}b_n^{\,\uw,\ut}(j,\nk)\,\frac{\dd\eta}{2\pi}\,. 
\end{equation} 
Here, $\abs{(\omega_1,\dots,\omega_{\nu})}_1\coloneqq\omega_1+\dots+\omega_{\nu}$, the phase $\psi_{n,1}^{\uw,\ut\,}\colon\cercle\to\D{R}$ is defined in the next subsection, and $b_n^{\,\uw,\ut}(j,\,\cdot\,)\colon\cercle\to\D{C}$ is chosen as indicated in Lemma \ref{lem:61} below.

\subsubsection{Definition of $\psi_{n,1}^{\uw,\ut\,}$}\label{sec:611} 

For $(\uw,\ut)\in(\Omega^*)^{\nu}\times\D{R}^{\nu}$ the phase $\psi_{n,1}^{\uw,\ut\,}$ is given by
\begin{equation}\label{L61a}
\psi_{n,1}^{\uw,\ut}(\ek)\coloneqq 2a(n)\,\Im\bigl(z(\uw;\ut)\ek \bigr),
\end{equation} 
where $z(\uw;\ut)$ will be defined by induction on $\nu$.

We first assume $\nu=1$. For $\omega\in\Omega$, $t\in\D{R}$ we define
\begin{equation}\label{61c}
z(\omega;t)\coloneqq\bigl(\eul^{-\ii{\omega}}-1\bigr)\eul^{-\ii t}=-2\sin\tfrac{\omega}{2}\,\eul^{\ii\pi/2-\ii\omega/2-\ii t}. 
\end{equation}
Thus the definition of $\psi_{n,1}^{\omega,t}\colon\cercle\to\D{R}$ is as in \cite{BZ5}*{(9.10b)}: 
\begin{equation}\label{61cc}
\psi_{n,1}^{\omega,t}(\ek)\coloneqq 2a(n)\,\Im\bigl(z(\omega;t)\ek \bigr)=-4a(n)\sin\tfrac{\omega}{2}\,\cos(\xi-t-\tfrac{\omega}{2}).
\end{equation}
Moreover, if $\tau_{\omega}\colon\cercle\to\cercle$ is the translation $\ek\mapsto\eul^{\ii(\xi-\omega)}$ and $\psi_{n,1}^{-\omega,t}\coloneqq\psi_{n,1}^{2\pi-\omega,t}$ we have the relation
\begin{equation}  \label{61cd}
\psi_{n,1}^{\omega,t}=-\psi_{n,1}^{-\omega,t}\circ\tau_{\omega}.
\end{equation}

Assuming now $\nu\geq 2$ and using induction with respect to $\nu$,  we define
\begin{equation}\label{61c'}
z(\uw;\ut)=z(\uw';\ut')\eul^{-\ii\omega_{\nu}}+z(\omega_{\nu};t_{\nu}),  
\end{equation}         
where 
\begin{equation}\label{61d}
\begin{split} 
\uw&=(\uw',\omega_{\nu})\in(\Omega^*)^{\nu}=(\Omega^*)^{\nu-1}\times\Omega^*,\\ 
\ut&=(\ut',t_{\nu})\in\D{R}^{\nu}=\D{R}^{\nu-1}\times\D{R}.
\end{split}
\end{equation}  
By \eqref{61c'} and \eqref{61cd} we observe that $\psi_{n,1}^{\uw,\ut}$ can also be defined by induction as in \cite{BZ5}*{(10.7) and (10.6)}:
\begin{equation}\label{L61a'}
\psi_{n,1}^{\uw,\ut}=\psi_{n,1}^{\uw',\ut'}\circ\tau_{\omega_{\nu}}+\psi_{n,1}^{\omega_{\nu},t_{\nu}}=(\psi_{n,1}^{\uw',\ut'}-\psi_{n,1}^{-\omega_{\nu},t_{\nu}})\circ\tau_{\omega_{\nu}}, 
\end{equation}       

\subsubsection{Approximation of $g_n^{\uw,\ut}(j)$ by $\F{g}_n^{\uw,\ut}(j)$}\label{sec:612} 

\begin{lemma}\label{lem:61} 
Let $t_0>0$ and $0<\varepsilon<\frac{1}{8}$ be fixed. Let $\F{g}_n^{\uw,\ut}(j)$ be the oscillatory integral defined by \eqref{61b} for $\abs{j-n}\leq n^{\gamma}$ with $\psi_{n,1}^{\uw,\ut}$ as in \eqref{L61a}. If $\hat C=\hat C(t_0,\varepsilon)>0$ is large enough, then for any $\nu\in\D{N}^*$, $(\uw,\ut)\in(\Omega^*)^{\nu}\times\croch{-t_0,t_0}^{\nu}$, $n\geq\nu^{\frac{1}{\varepsilon}}$, and $j\in\croch{n-n^{\gamma},n+n^{\gamma}}$ we can choose $b_n^{\,\uw,\ut}(j,\,\cdot\,)\colon\cercle\to\D{C}$ in \eqref{61b} such that $\F{g}_n^{\uw,\ut}(j)$ satisfies
\begin{equation}\label{L61d}
\sup_{\abs{j-n}\leq n^{\gamma}}\abs{g_n^{\uw,\ut}(j)-\F{g}_n^{\uw,\ut}(j)}\leq\hat Cn^{-\gamma+4\varepsilon},
\end{equation}
with
\begin{subequations}\label{L61}
\begin{align}\label{L61f}
&\abs{b_n^{\,\uw,\ut}(j,\nk)}=1,\\   
\label{L61c}
&\norm{b_n^{\,\uw,\ut}(j,\,\cdot\,)}_{\class^2(\cercle)}\leq\hat Cn^{4\varepsilon},\\   
\label{L61e}
&\norm{\partial_{t_{\nu}}b_n^{\,\uw,\ut}(j,\,\cdot\,)}_{\class^0(\cercle)}\leq\hat Cn^{\varepsilon},\quad\nu\geq 2.  
\end{align}
\end{subequations}
\end{lemma} 

\begin{proof} 
See Section \ref{sec:72}, in particular Section \ref{sec:71} and \eqref{71''} for the actual choice of $b_n^{\,\uw,\ut}(j,\,\cdot\,)$.
\end{proof} 

\subsection{Properties of $\BS{z(\uw;\ut)}$}\label{sec:62}

Let $z(\uw;\ut)$, $\uw\in(\Omega^*)^{\nu}$, $\ut\in\D{R}^{\nu}$ be as in Section \ref{sec:611}. We write 
\begin{equation}\label{62}
z(\uw;\ut)=\abs{z(\uw;\ut)}\eul^{\ii\alpha(\uw;\ut)}, 
\end{equation} 
where $0\leq\alpha(\uw;\ut)<2\pi$ is the argument of $z(\uw;\ut)$. If $\nu\geq 2$ we write $\uw=(\uw',\omega_{\nu})\in(\Omega^*)^{\nu-1}\times\Omega^*$, $\ut=(\ut',t_{\nu})\in\D{R}^{\nu-1}\times\D{R}$ as in \eqref{61d} and define  
\begin{align}\label{62'} 
\hat z(\uw;\ut')&\coloneqq\frac{z(\uw';\ut')}{2\sin(\omega_{\nu}/2)},\\   
\label{62''}
\hat\alpha(\uw;\ut')&\coloneqq\alpha(\uw';\ut')-\omega_{\nu}/2-\pi/2.    
\end{align} 
For the next lemma we also introduce the function $\F{h}\colon\D{R}\times\D{C}\to\D{R}_+$ by
\begin{equation}\label{L62a}
\F{h}(t,z)\coloneqq\sqrt{4\abs{z}\sin^2(t/2)+(1-\abs{z})^2}.
\end{equation} 

\begin{lemma}\label{lem:63} 
We assume $\nu\geq 2$. Then for $\uw=(\uw',\omega_{\nu})\in(\Omega^*)^{\nu-1}\times\Omega^*$ and $\ut=(\ut',t_{\nu})\in\D{R}^{\nu-1}\times\D{R}$ we have the relation   
\begin{equation}\label{L62b}
\abs{z(\uw;\ut)}=2\sin(\omega_{\nu}/2)\,\F{h}(t_{\nu}+\hat\alpha(\uw;\ut'),\hat z(\uw;\ut')),  
\end{equation} 
where $\F{h}$ is given by \eqref{L62a}, $\hat z(\uw;\ut')$ by \eqref{62'}, and $\hat\alpha(\uw;\ut')$ by \eqref{62''}. Moreover, 
\begin{equation}\label{L62d}
\abs{z(\uw;\ut)}\geq\frac{2\sin(\omega_{\nu}/2)}{\pi}\abs{t_{\nu}+\hat\alpha(\uw;\ut')}_{2\pi} 
\end{equation}
where $\abs{t}_{2\pi}\coloneqq\dist(t,2\pi\D{Z})$ and   
\begin{equation}\label{L62c}
\left\lvert\partial_{t_{\nu}}\eul^{\ii\alpha(\uw;\,\ut)}\right\rvert\leq\frac{6}{\abs{z(\uw;\ut)}}.
\end{equation} 
\end{lemma} 

\begin{proof} 
The proof consists of four steps. 
\begin{stepone*}
We first claim that for $t\in\D{R}$ and $z\in\D{C}$
\begin{equation} \label{62a'}  
\left\lvert\abs{z}-\eul^{-\ii t}\right\rvert=\F{h}(t,z).\end{equation}   
Indeed, the left hand side of \eqref{62a'} is 
\begin{equation*} \label{62a''} 
\sqrt{(\abs{z}-\cos t)^2+\sin^2t}=\sqrt{2\abs{z}(1-\cos t)+(1-\abs{z})^2} 
\end{equation*}  
which is $\F{h}(t,z)$ defined by \eqref{L62a} due to $1-\cos t=2\sin^2(t/2)$.
\end{stepone*}    

\begin{steptwo*}
Now we will show \eqref{L62b}. Combining \eqref{61c}, \eqref{61c'} with \eqref{62} we obtain  
\begin{equation} \label{62aa} 
z(\uw;\ut)=\abs{z(\uw;\ut')}\eul^{\ii\alpha(\uw';\ut')-\ii\omega_{\nu}}-2\sin(\omega_{\nu}/2)\,\eul^{-\ii(t_{\nu}+\omega_{\nu}/2-\pi/2)}. 
\end{equation}  
Using \eqref{62'} and \eqref{62''} in the right hand side of \eqref{62aa} we find 
\[ 
z(\uw;\ut)=2\sin(\omega_{\nu}/2)\,\left(\abs{\hat z(\uw;\ut')}-\eul^{-\ii(t_{\nu}+\hat\alpha(\uw;\ut'))}\right)\eul^{\ii\alpha(\uw';\ut')-\ii\omega_{\nu}},  
\]
hence 
\begin{equation} \label{62a}  
\abs{z(\uw;\ut)}=2\sin(\omega_{\nu}/2)\,\left\lvert\abs{\hat z(\uw;\ut')}-\eul^{-\ii(t_{\nu}+\hat\alpha(\uw;\ut'))}\right\rvert.  
\end{equation}    
Denoting $z\coloneqq\hat z(\uw;\ut')$, $t\coloneqq t_{\nu}+\hat\alpha(\uw;\ut')$ and using \eqref{62a'} we can express the right hand side of \eqref{62a} in the form   
\[ 
2\sin(\omega_{\nu}/2)\left\lvert\abs{z}-\eul^{-\ii t}\right\rvert= 2\sin(\omega_{\nu}/2)\,\F{h}(t,z),  
\] 
which completes the proof of \eqref{L62b}. 
\end{steptwo*}

\begin{stepthree*}
Here we will show \eqref{L62d}. Combining 
\begin{equation} \label{h}
\F{h}(t,z)\geq\sqrt{(4\abs{z}+(1-\abs{z})^2)\sin^2(t/2)}=(1+\abs{z})\abs{\sin(t/2)} 
\end{equation}  
with \eqref{L62b} we obtain 
\begin{equation*} \label{hh}
\abs{z(\uw;\ut)}\geq 2\sin(\omega_{\nu}/2)\,\left\lvert\sin\tfrac{1}{2}(t_{\nu}+\hat\alpha(\uw;\ut'))\right\rvert.   
\end{equation*}  
Thus \eqref{L62d} follows from $\abs{\sin(t/2)}\geq\frac{1}{\pi}\abs{t}_{2\pi}$. 
\end{stepthree*}  

\begin{stepfour*}
Finally we show \eqref{L62c}. We begin by writing
\begin{equation*} \label{alpha'}
\partial_{t_{\nu}}\frac{z(\uw;\ut)}{\abs{z(\uw;\ut)}}=\frac{1}{\abs{z(\uw;\ut)}}\left(\partial_{t_{\nu}}z(\uw;\ut)-\frac{z(\uw;\ut)}{\abs{z(\uw;\ut)}}\partial_{t_{\nu}}\abs{z(\uw;\ut)}\right).
\end{equation*} 
Since 
$\abs{\partial_{t_{\nu}}z(\uw;\ut)}=\abs{\partial_{t_{\nu}} z(\omega_{\nu};t_{\nu})}=2\sin(\omega_{\nu}/2)\leq 2$, the proof of \eqref{L62c} is completed if we show 
\begin{equation} \label{62b}
\left\lvert\partial_{t_{\nu}}\abs{z(\uw;\ut)}\right\rvert\leq 4.
\end{equation}
However using \eqref{L62b} we obtain  
\[
\partial_{t_{\nu}}\abs{z(\uw;\ut)}=2\sin(\omega_{\nu}/2)\,\partial_{t_{\nu}}\F{h}(t_{\nu}+\hat\alpha(\uw;\ut'),\hat z(\uw;\ut'))
\] 
Then using \eqref{h} we have   
\[
\abs{\partial_t\F{h}(t,z)}=\frac{\abs{z}\,\abs{\sin t}}{\F{h}(t,z)}\leq\frac{\abs{z}\,\abs{\sin t}}{(1+\abs{z})\abs{\sin(t/2)}}
\] 
and \eqref{62b} results from the inequality $\abs{\sin t}\leq 2\left\lvert\sin\frac{t}{2}\right\rvert$.\qedhere
\end{stepfour*}
\end{proof}  

\section{Proof of Proposition \ref{prop:4'}.}\label{sec:63} 
\subsection{First reductions}\label{sec:6'1}

In this section $0<\varepsilon<\frac{1}{8}$ and $t_0>0$ are fixed. As in \eqref{60} we denote by $H_n^{\uw,\ut}$ the operator $H_n^{\,\omega_1,t_1}\!\dots H_n^{\,\omega_{\nu},t_{\nu}}$ for $(\uw,\ut)\in(\Omega^*)^{\nu}\times\D{R}^{\nu}$ and by $g_n^{\uw,\ut}(j)$ its $j$-th diagonal entry.

We first observe that instead of \eqref{pr32} it suffices to prove the estimates 
\begin{subequations}  \label{pr32'}
\begin{alignat}{2} \label{pr32c'}  
\left\lvert\int_{\Delta_1}g_n^{\uw,\ut}(j)\,\dd t_2\right\rvert&\leq Cn^{-\gamma+5\varepsilon}&\quad&(\nu=2),\\
\label{pr32d'}
\int_{\Delta_2}\left\lvert\int_{\Delta_1}g_n^{\uw,\ut}(j)\,\dd t_{\nu}\right\rvert\dd t_{\nu-1}&\leq Cn^{-\gamma+5\varepsilon}&&(3\leq\nu\leq n^{\varepsilon})   
\end{alignat}
\end{subequations}
for any intervals $\Delta_1,\,\Delta_2\subset\croch{-t_0,t_0}$, $(\uw,\ut)\in(\Omega^*)^{\nu}\times\croch{-t_0,t_0}^{\nu}$, and $j\in\croch{n-n^{\gamma},n+n^{\gamma}}$.

Further on, $\F{g}_n^{\uw,\ut}(j)$ is given by \eqref{61b} with $b_n^{\,\uw,\ut}(j,\,\cdot\,)$ as in Lemma \ref{lem:61}, $\abs{j-n}\leq n^{\gamma}$, $(\uw,\ut)\in(\Omega^*)^{\nu}\times\croch{-t_0,t_0}^{\nu}$, and $2\leq\nu\leq n^{\varepsilon}$. Then due to Lemma \ref{lem:61}, instead of proving estimates \eqref{pr32'} it suffices to prove these ones:
\begin{subequations}  \label{pr32''}
\begin{alignat}{2} \label{pr32c''} 
\left\lvert\int_{\Delta_1}\F{g}_n^{\uw,\ut}(j)\,\dd t_2\right\rvert&\leq Cn^{-\gamma+5\varepsilon}&\quad&(\nu=2),\\
\label{pr32d''} 
\int_{\Delta_2}\left\lvert\int_{\Delta_1}\F{g}_n^{\uw,\ut}(j)\,\dd t_{\nu}\right\rvert\dd t_{\nu-1}&\leq Cn^{-\gamma+5\varepsilon}&&(3\leq\nu\leq n^{\varepsilon}).
\end{alignat}
\end{subequations} 
To estimate $\int_{\Delta_1}\F{g}_n^{\uw,\ut}(j)\,\dd t_{\nu}$ it suffices to consider the following two cases:
\begin{enumerate}[1)]
\item
$\Delta_1\subset\croch{-n^{-\gamma},n^{-\gamma}}+2\pi\D{Z}-\hat\alpha(\uw;\ut')$,
\item
$\Delta_1\subset\croch{n^{-\gamma},2\pi-n^{-\gamma}}+2\pi\D{Z}-\hat\alpha(\uw;\ut')$,
\end{enumerate}
where $\hat\alpha(\uw;\ut')\coloneqq\alpha(\uw';\ut')-\omega_{\nu}/2-\pi/2$ as in \eqref{62''}.

\subsection{Case 1)}\label{sec:6'2}

We assume $\nu\geq 2$. The definition of $H_n^{\uw,\ut}$ shows that its diagonal entries satisfy $\abs{g_n^{\uw,\ut}(j)}\leq 1$. Then, since $\abs{j-n}\leq n^{\gamma}$, estimate \eqref{L61d} from Lemma \ref{lem:61} applies and gives
\begin{equation} \label{630a} 
\abs{\F{g}_n^{\uw,\ut}(j)}\leq C_1n^{4\varepsilon}
\end{equation} 
for some constant $C_1>0$. Since $\Delta_1$ is a subinterval of $\croch{-t_0,t_0}$ satisfying 
\begin{equation} \label{63aa'} 
\Delta_1\subset\croch{-n^{-\gamma},n^{-\gamma}}+2\pi\D{Z}-\hat\alpha(\uw;\ut'),
\end{equation}  
its length satisfies $\abs{\Delta_1}\leq 2n^{-\gamma}$ and using \eqref{630a} we get the estimate
\begin{equation} \label{63aa} 
\left\lvert\int_{\Delta_1}\F{g}_n^{\uw,\ut}(j)\,\dd t_{\nu}\right\rvert\leq C_0n^{-\gamma+4\varepsilon}   
\end{equation}    
for some constant $C_0>0$. Thus, in case 1) the proof of \eqref{pr32''} is completed.
\subsection{Case 2)}\label{sec:6'3}

We now assume $\Delta_1$ is a subinterval of $\croch{-t_0,t_0}$ satisfying 
\begin{equation} \label{63aa''} 
\Delta_1\subset\croch{n^{-\gamma},2\pi-n^{-\gamma}}+2\pi\D{Z}-\hat\alpha(\uw;\ut').
\end{equation}   
If $t_{\nu}\in\Delta_1$, then $\abs{t_{\nu}+\hat\alpha(\uw;\ut')}_{2\pi}\geq n^{-\gamma}$ and $\abs{z(\uw;\ut)}\geq(2/\pi)\sin(\pi/N)n^{-\gamma}$ by estimate \eqref{L62d} from Lemma \ref{lem:63}. In particular, $z(\uw;\ut)\neq 0$. 

\subsubsection{Use of the stationary phase formula} \label{sec:6'4}

Writing $z(\uw;\ut)=\abs{z(\uw;\ut)}\eul^{\ii\alpha(\uw;\ut)}$ in \eqref{L61a} we find
\[
\psi_{n,1}^{\uw,\ut}(\ek)=2a(n)\abs{z(\uw;\ut)}\sin(\xi+\alpha(\uw;\ut)).
\]
Using this expression in \eqref{61b} we get  
\begin{equation} \label{63bb}
\F{g}_n^{\uw,\ut}(j)=\eul^{\ii j\abs{\uw}_1}\int_0^{2\pi}\eul^{2\ii a(n)\abs{z(\uw;\ut)}\sin(\eta+\alpha(\uw;\ut))}b_n^{\,\uw,\ut}(j,\nk)\,\frac{\dd\eta}{2\pi}\,.    
\end{equation} 
This oscillatory integral is of the type $\E{I}(b,\mu,\eta_0)$ considered in Lemma \ref{lem:41} for $b=b_n^{\,\uw,\ut}(j,\,\cdot\,)$, $\mu=2a(n)\abs{z(\uw;\ut)}$, and $\eta_0=\pi/2-\alpha(\uw;\ut)$. Since $z(\uw;\ut)\neq 0$ the stationary phase formula \eqref{eq:stat.phase} applies:    
\begin{subequations}\label{63cr}
\begin{equation} \label{63c}
\eul^{-\ii j\abs{\uw}_1}\F{g}_n^{\uw,\ut}(j)=\sum_{\kappa=\pm 1}\frac{\eul^{\ii\kappa(2a(n)\abs{z(\uw;\ut)}-\pi/4)}}{2\sqrt{\pi a(n)\abs{z(\uw;\ut)}}}\,b_n^{\,\uw,\ut}(j,\kappa\ii\eul^{-\ii\alpha(\uw;\ut)})+\F{r}_n^{\,\uw,\ut}(j) 
\end{equation}  
with  
\begin{equation} \label{63c'}
\abs{\F{r}_n^{\,\uw,\ut}(j)}\leq\frac{C_0}{a(n)\abs{z(\uw;\ut)}} 
\norm{b_n^{\,\uw,\ut}(j,\,\cdot\,)}_{\class^2(\cercle)}. 
\end{equation}
\end{subequations} 
Integrating \eqref{63c} we find 
\begin{equation} \label{63cc}
\eul^{-\ii j\abs{\uw}_1}\int_{\Delta_1}\F{g}_n^{\uw,\ut}(j)\,\dd t_{\nu}=\eul^{-\ii\pi/4}\C{J}_{n,j,+}^{\uw,\ut'}(\Delta_1)+\eul^{\ii\pi/4}\C{J}_{n,j,-}^{\uw,\ut'}(\Delta_1)+\int_{\Delta_1}\F{r}_n^{\,\uw,\ut}(j)\,\dd t_{\nu}  
\end{equation}   
where
\begin{equation} \label{64d}
\C{J}_{n,j,\pm}^{\uw,\ut'}(\Delta_1)\coloneqq\int_{\Delta_1}\frac{\eul^{\pm 2\ii a(n)\abs{z(\uw;\ut)}}}{2\sqrt{\pi a(n)\abs{z(\uw;\ut)}}}\,b_n^{\,\uw,\ut}(j,\pm\ii\eul^{-\ii\alpha(\uw;\ut)})\,\dd t_{\nu}.  
\end{equation}  

\subsubsection{Remainder estimate}  \label{sec:6'45}

We claim that the remainder in \eqref{63cc} can be estimated by
\begin{equation} \label{64e}
\left\lvert\int_{\Delta_1}\F{r}_n^{\,\uw,\ut}(j)\,\dd t_{\nu}\right\rvert\leq Cn^{-\gamma+5\varepsilon}.    
\end{equation} 
Indeed,  (H1a), \eqref{L62d} and \eqref{L61c} allow us to derive from estimate \eqref{63c'} of $\abs{\F{r}_n^{\,\uw,\ut}(j)}$ this new one  
\[ 
\abs{\F{r}_n^{\,\uw,\ut}(j)}\leq\frac{Cn^{4\varepsilon}}{n^{\gamma}\abs{t_{\nu}+\hat\alpha(\uw;\ut')}_{2\pi}},
\]  
and to get \eqref{64e} it suffices to observe that   
\[ 
\int_{\Delta_1}\frac{\dd t_{\nu}}{\abs{t_{\nu}+{\hat\alpha}(\uw;\ut')}_{2\pi}}\leq 2\int_{n^{-\gamma}}^\pi\frac{\dd t}{t}=\ord(\ln n).
\] 
To complete the proof of \eqref{pr32''} in case 2) we will prove the following estimates: 
\begin{subequations}  \label{pr32cd}
\begin{alignat}{2} \label{pr32c'''} 
\left\lvert\C{J}_{n,j,\pm}^{\uw,\ut'}(\Delta_1)\right\rvert&\leq Cn^{-\gamma+2\varepsilon}&\quad&(\nu=2),\\
\label{pr32d'''} 
\int_{\Delta_2}\left\lvert\C{J}_{n,j,\pm}^{\uw,\ut'}(\Delta_1)\right\rvert\dd t_{\nu-1}&\leq Cn^{-\gamma+2\varepsilon}&&(3\leq\nu\leq n^{\varepsilon}). 
\end{alignat}  
\end{subequations}
According to \eqref{63cc} these estimates together with the remainder estimate \eqref{64e} actually imply \eqref{pr32''} in case 2).
\subsubsection{Transformation of $\C{J}_{n,j,\pm}^{\uw,\ut'}(\Delta_1)$}  \label{sec:6'5}

If the new variable $t=t_{\nu}+\hat\alpha(\uw;\ut')$ is introduced, then \eqref{L62b} becomes $\abs{z(\uw;\ut)}=2\sin(\omega_{\nu}/2)\,\F{h}(t,\hat z(\uw;\ut'))$ and \eqref{64d} takes the form
\begin{subequations}  \label{630}
\begin{equation} \label{63e}
\C{J}_{n,j,\pm}^{\uw,\ut'}(\Delta_1)=\int_{\Delta_1+\hat\alpha(\uw;\ut')}\frac{\eul^{\pm \ii\tilde\mu_n(\omega_{\nu})\F{h}(t,\hat z(\uw;\ut'))}}{\sqrt{2\pi\tilde\mu_n(\omega_{\nu})\F{h}(t,\hat z(\uw;\ut'))}}\,b_{n,j,\pm}^{\,\uw,\ut'}(t)\,\dd t 
\end{equation}   
where  
\begin{align} \label{63f}
\tilde\mu_n(\omega_{\nu})&\coloneqq 4a(n)\sin(\omega_{\nu}/2),\\ 
\label{63j}
b_{n,j,\pm}^{\,\uw,\ut'}(t)&\coloneqq b_n^{\,\uw,(\ut',t-\hat\alpha(\uw;\ut'))}\bigl(j,\pm\ii\eul^{-\ii\alpha(\uw;(\ut',t-\hat\alpha(\uw;\ut'))}\bigr).
\end{align}
\end{subequations}
Our goal is to get \eqref{pr32cd}. We first observe that in the case $\nu=2$ we have $\abs{z(\uw';\ut')}=\abs{z(\omega_1;t_1)}=2\sin(\omega_1/2)>0$ for $\omega_1\in\Omega^*$, hence $\hat z(\uw;\ut')\neq 0$. In the case $\nu\geq 3$ we have to estimate the integral $\int_{\Delta_2}\abs{\C{J}_{n,j,\pm}^{\uw,\ut'}(\Delta_1)}\dd t_{\nu-1}$ and we observe that in this integral we can forget the $t_{\nu-1}$ such that $\hat z(\uw;\ut')=0$, or equivalently, $z(\uw';\ut')=0$ because there are only finitely many. By \eqref{L62d} they indeed satisfy $t_{\nu-1}\in 2\pi\D{Z}-\hat\alpha(\uw';\ut'')$ where $\ut'=(\ut'',t_{\nu-1})$. So, we henceforth assume $\hat z(\uw;\ut')\neq 0$ when estimating $\C{J}_{n,j,\pm}^{\uw,\ut'}(\Delta_1)$.

Our next step is to write \eqref{63e} as an integral of the type considered in Section \ref{sec:5}. For this purpose we denote by $\zeta\equiv\zeta(\hat z)$ the nonnegative number associated to $\hat z\in\D{C}^*$ by 
\begin{equation} \label{zeta} 
\zeta\coloneqq\left\lvert\abs{\hat z}^{-1/2}-\abs{\hat z}^{1/2}\right\rvert.
\end{equation}   
Since $(1-\abs{\hat z})^2=\zeta^2\abs{\hat z}$, we can write $\F{h}(t,\hat z)$, $\hat z\neq 0$ as follows:  
\begin{equation} \label{zeta'} 
\F{h}(t,\hat z)\coloneqq\sqrt{4\abs{\hat z}\sin^2(t/2)+(1-\abs{\hat z})^2}=\sqrt{\abs{\hat z}}\,\sqrt{4\sin^2(t/2)+\zeta^2}.   
\end{equation} 
Then for $\mu>0$, $\zeta\geq 0$, and $\Delta_1'$ a bounded interval we introduce the integral 
\begin{equation} \label{63} 
\tilde{\C{J}}_{n,j,\pm}^{\uw,\ut'}(\zeta,\mu,\Delta_1')\coloneqq\frac{1}{\sqrt{2\pi\mu}}\int_{\Delta_1'}\frac{\eul^{\pm\ii\mu\sqrt{4\sin^2(t/2)+\zeta^2}}}{(4\sin^2(t/2)+\zeta^2)^{1/4}}\,b_{n,j,\pm}^{\,\uw,\ut'}(t)\,\dd t.
\end{equation}
Due to \eqref{zeta'} we have 
\begin{subequations} \label{63t}
\begin{equation} \label{63g}
\C{J}_{n,j,\pm}^{\uw,\ut'}(\Delta_1)=\tilde{\C{J}}_{n,j,\pm}^{\uw,\ut'}(\zeta,\mu,\Delta_1')
\end{equation} 
where $\Delta_1'=\Delta_1+\hat\alpha(\uw;\ut')$ and
\begin{align} \label{63h} 
&\zeta\equiv\zeta(\uw;\ut')\coloneqq\left\lvert\abs{\hat z(\uw;\ut')}^{1/2}-\abs{\hat z(\uw;\ut')}^{-1/2}\right\rvert\\ 
\label{63h'}   
&\mu\equiv\mu_n(\uw;\ut')\coloneqq\tilde\mu_n(\omega_{\nu})\,\abs{\hat z(\uw;\ut')}^{1/2}  
\end{align}
\end{subequations} 
Thus it remains to investigate $\tilde{\C{J}}_{n,j,\pm}^{\uw,\ut'}(\zeta,\mu,\Delta_1')$ with $\zeta$ and $\mu$ as in \eqref{63h}, and \eqref{63h'}, and 
\begin{equation} \label{643f}
\Delta_1'\subset\croch{n^{-\gamma},2\pi-n^{-\gamma}}+2\pi\D{Z}.
\end{equation}  
\subsubsection{Estimate of $\tilde{\C{J}}_{n,j,\pm}^{\uw,\ut'}(\zeta,\mu,\Delta_1')$}  \label{sec:634}
   
In this section $\zeta$ and $\mu$ are given by \eqref{63h} and \eqref{63h'} and $\Delta_1'$ is an interval satisfying \eqref{643f}. We claim that there is a constant $C>0$ such that  
\begin{equation} \label{644a}
\abs{\tilde{\C{J}}_{n,j,\pm}^{\uw,\ut'}(\zeta,\mu,\Delta_1')}\leq C\frac{n^{2\varepsilon}}{a(n)}\,\left(1+\abs{\hat z(\uw;\ut')}^{-3/4}\right).
\end{equation} 
For this purpose we first observe that $\tilde{\C{J}}_{n,j,\pm}^{\uw,\ut'}(\zeta,\mu,\Delta_1')$ coincides with the oscillatory integral $\C{J}(b,t_1,t_2,\zeta,\pm\mu)$ defined in \eqref{51a} for $b=\frac{1}{\sqrt{2\pi\mu}}\,b_{n,j,\pm}^{\,\uw,\ut'}(t)$ and $\Delta_1'=\croch{t_1,t_2}$. Since $\mu=\mu_n(\uw;\ut')>0$, Lemma \ref{lem:51} applies and gives
\begin{equation}  \label{644b}
\abs{\tilde{\C{J}}^{\uw,\ut'}_{n,j,\pm}(\zeta,\mu,\Delta_1')}\leq C\,\frac{1+\zeta^{1/2}}{\mu}\,\C{M}(b_{n,j,\pm}^{\,\uw,\ut'},\Delta_1').
\end{equation} 
Since 
\[
\frac{1+\zeta^{1/2}}{\mu}\leq\frac{1+\abs{\hat z(\uw;\ut')}^{1/4}+\abs{\hat z(\uw;\ut')}^{-1/4}}{4a(n)\sin(\omega_{\nu}/2)\abs{\hat z(\uw;\ut')}^{1/2}}\leq\frac{1+\abs{\hat z(\uw;\ut')}^{-3/4}}{2a(n)\sin(\omega_{\nu}/2)} 
\]
we deduce \eqref{644a} from \eqref{644b} if we show the estimate 
\begin{equation} 
\C{M}(b_{n,j,\pm}^{\,\uw,\ut'},\Delta_1')\leq Cn^{\varepsilon}\ln n. 
\end{equation} 
However, using definition \eqref{63j} we see that relation \eqref{L61f} from Lemma \ref{lem:61} gives
\[
\abs{b_{n,j,\pm}^{\,\uw,\ut'}(t)}=1.
\]
Moreover, by estimate \eqref{L61e} from Lemma \ref{lem:61} and estimates \eqref{L62c} and \eqref{L62d} from Lemma \ref{lem:63} we have
\[
\abs{\partial_tb_{n,j,\pm}^{\,\uw,\ut'}(t)}\leq\norm{\partial_{t_{\nu}}b_n^{\,\uw,\ut}(j,\,\cdot\,)}_{\class^0(\cercle)}\norm{\partial_{t_{\nu}}\eul^{-\ii\alpha(\uw;\ut)}}_{\class^0(\cercle)}\leq\hat Cn^{\varepsilon}\frac{6}{\abs{z(\uw;\ut)}}\leq C'n^{\varepsilon}\left(1+\frac{1}{\abs{t}_{2\pi}}\right).
\]
Thus $\C{M}(b_{n,j,\pm}^{\,\uw,\ut'},\Delta_1')\coloneqq\sup_{\Delta_1'}\abs{b_{n,j,\pm}^{\,\uw,\ut'}(t)}+\int_{\Delta_1'}\abs{\partial_tb_{n,j,\pm}^{\,\uw,\ut'}(t)}\dd t$ can be estimated by 
\[ 
\C{M}(b_{n,j,\pm}^{\,\uw,\ut'},\Delta_1')\leq 1+C'n^{\varepsilon}\left(2\pi+\int_{\Delta_1'}\frac{\dd t}{\abs{t}_{2\pi}}\right)=1+2C'n^{\varepsilon}\left(\pi+\int_{n^{-\gamma}}^{\pi}\frac{\dd t}{t}\right)=\ord(n^{\varepsilon}\ln n).  
\] 
\subsubsection{End of the proof of estimates \eqref{pr32''} in case 2)}
  
Due to \eqref{63t}, \eqref{644a}, \eqref{62'} and (H1a), there is a constant $C>0$ such that   
\begin{equation}  \label{65}
\abs{\C{J}_{n,j,\pm}^{\uw,\ut'}(\Delta_1)}\leq Cn^{-\gamma+2\varepsilon}\,\left(1+\abs{z(\uw';\ut')}^{-3/4}\right).
\end{equation}

If $\nu=2$, then $\uw'=\omega_1\in\Omega^*$ and $\abs{z(\uw';\ut')}=2\sin(\omega_1/2)\geq 2\sin(\pi/N)>0$ and it is clear that \eqref{65} implies \eqref{pr32c'''}. We then get \eqref{pr32c''} using \eqref{pr32c'''} and \eqref{64e}. This proves assertion (a) in Proposition \ref{prop:4'}.
 
If $\nu\geq 3$, then $\uw'=(\uw'',\omega_{\nu-1})\in(\Omega^*)^{\nu-2}\times\Omega^*$ and $\ut'=(\ut'',t_{\nu-1})\in\D{R}^{\nu-2}\times\D{R}$. If 
\[
\hat\alpha(\uw',\ut'')=\alpha(\uw'',\ut'')-\omega_{\nu-1}/2-\pi/2
\] 
then $t_{\nu-1}+\hat\alpha(\uw';\ut'')\notin 2\pi\D{Z}$ implies $\hat z(\uw;\ut')\neq 0$ and
\[ 
\C{J}_{n,j,\pm}^{\uw,\ut'}(\Delta_1)\leq C'n^{-\gamma+2\varepsilon}\left(1+\abs{t_{\nu-1}+\hat\alpha(\uw',\ut'')}_{2\pi}^{-3/4}\right)
\]
and therefore
\[
\int_{\Delta_2}\C{J}_{n,j,\pm}^{\uw,\ut'}(\Delta_1)\,\dd t_{\nu-1}\leq C'n^{-\gamma+2\varepsilon}\int_{\Delta_2-\hat\alpha(\uw',\ut'')}(1+\abs{t}_{2\pi}^{-3/4})\,\dd t. 
\] 
Thus, \eqref{pr32d'''} follows from the fact that the function $t\to 1+\abs{t}_{2\pi}^{-3/4}$ is locally Lebesgue integrable on $\D{R}$. We then obtain \eqref{pr32d''} using \eqref{pr32d'''} and \eqref{64e}. Assertion (b) in Proposition \ref{prop:4'} is proven.\qed
\section{Proofs of Lemmas \ref{lem:61} and \ref{lem:22} (\textup{b})}\label{sec:7} 

The proofs of Lemmas \ref{lem:61} and \ref{lem:22} (b) are completed in Sections \ref{sec:72} and \ref{sec:73}, respectively. Both proofs are based on properties of phase functions introduced in \cite{BZ5}*{Section 10}. 
\subsection{Proof of Lemma \ref{lem:61}}\label{sec:72}

We first fix the definition of $b_n^{\uw,\ut}(j,\nk)$ to complete that of $\F{g}_n^{\uw,\ut}(j)$.
\subsubsection{Definition of $b_n^{\,\uw,\ut}(j,\nk)$}\label{sec:71}

Let $Q_n^{\uw,\ut}$ be the operators introduced in \cite{BZ5}*{(10.1)}. Their definition involves a phase $\tilde\psi_n^{\uw,\ut}\equiv\tilde\psi_n^{\uw,\ut}(j,\nk)$ whose construction is given in \cite{BZ5}*{(10.5) and Section 10.3}. By \cite{BZ5}*{Lemma 10.3} these operators satisfy 
\begin{equation} \label{710} 
\abs{H_n^{\uw,\ut}(j,j)-Q_n^{\uw,\ut}(j,j)}\leq\nu\,n^{\gamma-1+3\varepsilon},
\end{equation}  
where $\uw\in(\Omega^*)^{\nu}$, $\ut\in\croch{-t_0,t_0}^{\nu}$, $n^{\varepsilon}\geq\nu$ and $n\geq\hat n$. Recall that $g_n^{\uw,\ut}(j)\coloneqq H_n^{\uw,\ut}(j,j)$. Let similarly
\begin{equation} \label{710''} 
\tilde{\F{g}}_n^{\uw,\ut}(j)\coloneqq Q_n^{\uw,\ut}(j,j).
\end{equation}     
Since $\nu\,n^{\gamma-1+3\varepsilon}\leq n^{-\gamma+4\varepsilon}$ follows from $\nu\leq n^{\varepsilon}$ and $\gamma\leq\frac{1}{2}$, the estimate \eqref{710} implies   
\begin{equation} \label{710'} 
\abs{g_n^{\uw,\ut}(j)-\tilde{\F{g}}_n^{\uw,\ut}(j)}\leq n^{-\gamma+4\varepsilon}.
\end{equation} 
If $\abs{j-n}\leq n^{\gamma}$, then \cite{BZ5}*{Section 10.2.2 and (10.5)} gives the expression   
\begin{equation} \label{71'} 
\tilde{\F{g}}_n^{\uw,\ut}(j)=\eul^{\ii j\abs{\uw}_1}\int_0^{2\pi}\eul^{\ii(\psi_{n,1}^{\uw,\ut}+\psi_{n,2}^{\uw,\ut})(\nk)+\ii(j-n)(\varphi_{n,1}^{\uw,\ut}+\varphi_{n,2}^{\uw,\ut})(\nk)}\frac{\dd\eta}{2\pi}\,.   
\end{equation}    
By \cite{BZ5}*{Lemma 10.1} and $\gamma\leq\frac{1}{2}$ the functions $\psi_{n,1}^{\uw,\ut},\,\varphi_{n,1}^{\uw,\ut}\colon\cercle\to\D{R}$ satisfy  
\begin{subequations}  \label{700}
\begin{align} \label{70a}
&\norm{\psi_{n,1}^{\uw,\ut}}_{\class^2(\cercle)}\leq C\nu\,n^{\gamma},\\   
\label{70b}
&\norm{\varphi_{n,1}^{\uw,\ut}}_{\class^2(\cercle)}\leq C\nu\,n^{\gamma-1}\leq C\nu\,n^{-\gamma}.
\end{align}
By \cite{BZ5}*{Lemma 10.2} and $\gamma\leq\frac{1}{2}$ the functions $\psi_{n,2}^{\uw,\ut},\,\varphi_{n,2}^{\uw,\ut}\colon\cercle\to\D{R}$ satisfy  
\begin{align}\label{70c}
&\norm{\psi_{n,2}^{\uw,\ut}}_{\class^2(\cercle)}\leq C\nu\,n^{\varepsilon},\\ 
\label{70d}
&\norm{\varphi_{n,2}^{\uw,\ut}}_{\class^2(\cercle)}\leq C\nu\,n^{2(\gamma-1)+\varepsilon}\leq Cn^{-2\gamma+2\varepsilon},
\end{align}
\end{subequations}
for $n^{\varepsilon}\geq\nu$ and $n\geq\hat n$. We can rewrite \eqref{71'} as
\begin{subequations} \label{71a}
\begin{align} \label{71b} 
\tilde{\F{g}}_n^{\uw,\ut}(j)&=\eul^{\ii j\abs{\uw}_1}\int_0^{2\pi}\eul^{\ii\psi_{n,1}^{\uw,\ut}(\nk)}\,\tilde b_n^{\,\uw,\ut}(j,\nk)\frac{\dd\eta}{2\pi}\,,
\intertext{with}   
\label{71c}
\tilde b_n^{\,\uw,\ut}(j,\nk)&\coloneqq\eul^{\ii(\psi_{n,2}^{\uw,\ut}(\nk)+(j-n)(\varphi_{n,1}^{\uw,\ut}+\varphi_{n,2}^{\uw,\ut})(\nk))}.
\end{align}    
\end{subequations}
Let us note that the definition \eqref{61b} of $\F{g}_n^{\uw,\ut}(j)$ is as that of $\tilde{\F{g}}_n^{\uw,\ut}(j)$, replacing $\tilde b_n^{\uw,\ut}$ by $b_n^{\uw,\ut}$ in \eqref{71b}.

The definition of $\psi_{n,1}^{\,\uw,\ut}$ is given in Section~\ref{sec:611}, that of $\varphi_{n,1}^{\uw,\ut}$ and $\psi_{n,2}^{\uw,\ut}$ are given in Sections \ref{sec:722} and \ref{sec:725}, respectively. Concerning $\varphi_{n,2}^{\uw,\ut}$ we observe that \eqref{70d} ensures $(j-n)\varphi_{n,2}^{\uw,\ut}=\ord(n^{-\gamma+2\varepsilon})$ for $\abs{j-n}\leq n^{\gamma}$, hence \eqref{710'} still holds for these $j$ if we forget $(j-n)\varphi_{n,2}^{\uw,\ut}$ in the r.h.s.\ of \eqref{71c}, i.e., if we replace $\tilde{\F{g}}_n^{\uw,\ut}(j)$ by $\F{g}_n^{\uw,\ut}(j)$ given by \eqref{61b} with 
\begin{equation} \label{71''}
b_n^{\,\uw,\ut}(j,\nk)\coloneqq\eul^{\ii(\psi_{n,2}^{\uw,\ut}(\nk)+(j-n)\varphi_{n,1}^{\uw,\ut}(\nk))}. 
\end{equation}  
Further on $b_n^{\,\uw,\ut}(j,\nk)$ is given by \eqref{71''}. It is obvious that \eqref{L61f} is satisfied and it is easy to see that the estimates \eqref{70b} and \eqref{70c} imply \eqref{L61c}. Thus, all that remains to be proved is \eqref{L61e} and for this it suffices to prove 
\begin{subequations}  \label{720}
\begin{align} \label{72}
\abs{\partial_t\varphi_{n,1}^{\uw,(\ut',t)}(\nk)}&\leq\hat C_0n^{-\gamma},\\
\label{72'}
\abs{\partial_t\psi_{n,2}^{\uw,(\ut',t)}(\nk)}&\leq\hat C_0n^{\varepsilon}. 
\end{align} 
\end{subequations}
The proofs of \eqref{72} and \eqref{72'} are given in Sections~\eqref{sec:723} and \eqref{sec:726}, respectively.
\subsubsection{Change of variable $\vartheta_n$}\label{sec:721}
 
As in \cite{BZ5}*{(8.3b)} we define $\varphi_n\colon\cercle\to\D{R}$ by
\begin{equation*}  \label{72A}
\varphi_n(\ek)\coloneqq 2\delta a(n)\bigl(1-\delta a(n)\cos\xi \bigr)\sin\xi.
\end{equation*} 
By (H1b) we have $\delta a(n)=\ord(n^{\gamma-1})$. Thus we can fix $n_0=n_0(\accol{\varphi_n})\in\D{N}$ large enough to ensure 
\begin{equation*}\label{72B}
\sup_{n\geq n_0}\norm{\varphi_n}_{\class^2(\cercle)}\leq\frac{1}{2}\,.   
\end{equation*}  
Further on we always assume $n\geq n_0$. Then we define the bijection $\eta_n\colon\D{R}\to\D{R}$ as in \cite{BZ5}*{(7.10)}:
\begin{equation*}\label{72C}
\eta_n(\xi)\coloneqq\xi-\varphi_n(\ek).  
\end{equation*}
As in \cite{BZ5}*{Section 7.5.1} we denote by $\xi_n\colon\D{R}\to\D{R}$ its inverse, i.e.\ $\xi_n(\eta)-\varphi_n(\eul^{\ii\xi_n(\eta)})=\eta$ and, $\eta\to\xi_n(\eta)-\eta$ being $2\pi$-periodic, we define $\tilde\xi_n\colon\cercle\to\D{R}$ by $\tilde\xi_n(\nk)=\xi_n(\eta)-\eta$. Finally, we define $\vartheta_n\colon\cercle\to\cercle$ as in \cite{BZ5}*{(7.13a)}:
\begin{equation*}\label{72E}
\vartheta_n(\nk)\coloneqq\nk\eul^{\ii\tilde\xi_n(\nk)}=\eul^{\ii\xi_n(\eta)}.
\end{equation*} 
Moreover, by \cite{BZ5}*{Lemma 9.2 (a)}, there exists a constant $C>0$ such that for any $q\in\class^m(\cercle)$, $m=1,2,3$ one has 
\begin{equation}\label{72q}
\norm{q-q\circ\vartheta_n}_{\class^{m-1}(\cercle)}\leq Cn^{-\gamma}\norm{q}_{\class^m(\cercle)}.
\end{equation}   
\subsubsection{Definition of $\varphi_{n,1}^{\uw,\ut}$}\label{sec:722}
  
We first assume $\nu=1$. If $(\omega,t)\in\Omega^*\times\D{R}$ we define $\varphi_{n,1}^{\,\omega,t}\colon\cercle\to\D{R}$ for $n\geq n_0$ as in \cite{BZ5}*{Section 10.3.1}:
\begin{subequations}
\begin{equation}\label{82F} 
\varphi_{n,1}^{\,\omega,t}\coloneqq(\varphi_n\circ\tau_{\omega}-\varphi_n)\circ\vartheta_n\circ\tau_t  
\end{equation}
where $\varphi_n$, $\vartheta_n$, and $n_0$ are as in Section \ref{sec:721} and $\tau_t(\ek)\coloneqq\eul^{\ii(\xi-t)}$. We have then $\varphi_{n,1}^{\,\omega,t}=\varphi_{n,1}^{\,\omega,0}\circ\tau_t$. For $\nu\geq 2$ we define $\varphi_{n,1}^{\uw,\ut}\colon\cercle\to\D{R}$ for $n\geq n_0$ by induction on $\nu$ as in \cite{BZ5}*{(10.6)}:
\begin{equation}   \label{82G}
\varphi_{n,1}^{\uw,\ut}\coloneqq(\varphi_{n,1}^{\uw'\!,\ut'}-\varphi_{n,1}^{-\omega,t})\circ\tau_{\omega},
\end{equation}
\end{subequations}  
where $\uw=(\uw',\omega)\in(\Omega^*)^{\nu-1}\times\Omega^*$, $\ut=(\ut',t)\in\D{R}^{\nu-1}\times\D{R}$, and $\varphi_{n,1}^{-\omega,t}\coloneqq\varphi_{n,1}^{2\pi-\omega,t}$.
 
\subsubsection{Proof of \eqref{72}}\label{sec:723}

By \eqref{82G} we get $\partial_t\varphi_{n,1}^{\uw,\ut}=-\partial_t\varphi_{n,1}^{-\omega,t}\circ\tau_{\omega}$ with $\varphi_{n,1}^{-\omega,t}=\varphi_{n,1}^{-\omega,0}\circ\tau_t$. For any $q\in\class^1(\cercle)$ we have the estimate
\begin{equation}\label{73}
\abs{\partial_t(q\circ\tau_t)(\nk)}=\abs{q^{(1)}\circ\tau_t(\nk)}\leq\norm{q}_{\class^1(\cercle)},
\end{equation}
where $q^{(1)}(\nk)\coloneqq\partial_\eta q(\nk)$. Applying \eqref{73} for $q=\varphi_{n,1}^{-\omega,0}$ we get \eqref{72} since \eqref{70b} ensures $\norm{\varphi_{n,1}^{-\omega,0}}_{\class^1(\cercle)}\leq Cn^{-\gamma}$.
 
\subsubsection{Change of variable $\vartheta_n^{\omega,t}$}\label{sec:724}
  
Let $\varphi_{n,1}^{\omega,t}$ be as above and $n_1\in\D{N}$ such that
\begin{equation}\label{74}
\sup_{n\geq n_1}\norm{\varphi_{n,1}^{\omega,t}}_{\class^2(\cercle)}\leq\frac{1}{2}\,.    
\end{equation}  
It suffices to choose $n_1$ such that $\sup_{n\geq n_1}\norm{\varphi_{n,1}}_{\class^2(\cercle)}\leq 1/4$. Indeed, using \eqref{72q} we get
\[
\norm{\varphi_{n,1}^{\omega,t}}_{\class^2(\cercle)}=\norm{\varphi_{n,1}^{\omega,0}}_{\class^2(\cercle)}\leq(1+Cn^{-\gamma})(\norm{\varphi_n\circ\tau_{\omega}}_{\class^3(\cercle)}+\norm{\varphi_n}_{\class^3(\cercle)})\leq C'n^{-\gamma},
\]
hence \eqref{74} holds if $n_1$ is chosen so that $Cn_1^{-\gamma}\leq 1/2$. From now on we assume $n\geq n_1$ and introduce the bijection $\eta_n^{\omega,t}\colon\D{R}\to\D{R}$ by 
\begin{equation} \label{741}
\eta_n^{\omega,t}(\xi)\coloneqq\xi-\varphi_{n,1}^{\omega,t}(\ek).  
\end{equation}
Let $\xi_n^{\omega,t}\colon\D{R}\to\D{R}$ be its inverse. Then $\xi_n^{\omega,t}(\eta)-\varphi_{n,1}^{\omega,t}(\eul^{\ii\xi_n^{\omega,t}(\eta)})=\eta$ and $\eta\to\xi_n^{\omega,t}(\eta)-\eta$ is $2\pi$-periodic. As before we can define $\tilde\xi_n^{\omega,t}\colon\cercle\to\D{R}$ by the formula $\tilde\xi_n^{\omega,t}(\nk)=\xi_n^{\omega,t}(\eta)-\eta$ and $\vartheta_n^{\omega,t}\colon\cercle\to\cercle$  by 
\begin{equation}\label{vt}
\vartheta_n^{\omega,t}(\nk)\coloneqq\nk\eul^{\ii\tilde\xi_n^{\omega,t}(\nk)}=\eul^{\ii\xi_n^{\omega,t}(\eta)}.
\end{equation} 
Since we can use $\varphi_{n,1}^{\omega,0}$ in place of $\varphi_n$ in \cite{BZ5}*{proof of Lemma 9.2}, there is a constant $C>0$ such that for any $q\in\class^2(\cercle)$ one has 
\begin{equation}\label{L81c}
\norm{q-q\circ\vartheta_n^{\omega,0}}_{\class^1(\cercle)}\leq Cn^{-\gamma}\norm{q}_{\class^2(\cercle)}.
\end{equation}   

\begin{lemma}\label{lem:81} 
\emph{(a)} One has 
\begin{equation}\label{L81a}
\vartheta_n^{\omega,t}=\tau_{-t}\circ\vartheta_n^{\omega,0}\circ\tau_t.
\end{equation} 
\emph{(b)} There is a constant $C>0$ such that
\begin{equation}\label{L81b}
\abs{\partial_t\vartheta_n^{\omega,t}(\ek)}\leq Cn^{-\gamma}.
\end{equation}  
\end{lemma}    

\begin{proof} 
(a) 
We first note that \eqref{82F} implies $\varphi_{n,1}^{\omega,t}=\varphi_{n,1}^{\omega,0}\circ\tau_t$. Using this relation in \eqref{741} we find $\eta_n^{\omega,t}(\xi)=t+\eta_n^{\omega,0}(\xi-t)$, hence also $\xi_n^{\omega,t}(\eta)=t+\xi_n^{\omega,0}(\eta-t)$. We have then $\tilde\xi_n^{\omega,t}(\nk)\coloneqq\xi^{\omega,t}_n(\eta)-\eta=\xi_n^{\omega,0}(\eta-t)-(\eta-t)=\tilde\xi_n^{\omega,0}(\eul^{\ii(\eta-t)})$, i.e.,
\begin{equation}\label{84c}
\tilde\xi_n^{\omega,t}=\tilde\xi_n^{\omega,0}\circ\tau_t, 
\end{equation}
and \eqref{L81a} follows from \eqref{84c} and \eqref{vt}.

(b) Since $\tilde\xi_n^{\omega,t}=\tilde\xi_n^{\omega,0}\circ\tau_t$ the estimate \eqref{73} ensures 
\begin{equation}\label{84d}
\abs{\partial_t\vartheta_n^{\omega,t}(\nk)}=\abs{\partial_t\tilde\xi_n^{\omega,t}(\nk)}\leq\norm{\tilde\xi_n^{\omega,0}}_{\class^1(\cercle)} 
\end{equation} 
and the right hand side of \eqref{84d} is $\ord(n^{-\gamma})$ due to  \cite{BZ5}*{Lemma 7.3} with $\varphi_{n,1}^{\omega,0}$ in place of $\varphi_n$.  
\end{proof} 
 
\subsubsection{Definition of $\psi_{n,2}^{\uw,\ut}$}\label{sec:725}

For $\omega\in\Omega^*$ we denote 
\begin{align}\label{725''} 
\psi_{n,\R{I}}^{\omega}&\coloneqq\psi_{n,2}\circ\tau_{\omega}-\psi_{n,2},\\
\label{725'} 
\psi_{n,\R{I\!I}}^{\omega}&\coloneqq\psi_{n,1}^{\omega,0}\circ\vartheta_n-\psi_{n,1}^{\omega,0},  
\end{align}  
where $\psi_{n,1}^{\omega,0}(\ek)=-4a(n)\sin\frac{\omega}{2}\,\cos(\xi-\frac{\omega}{2})$ as in \eqref{61cc} and
\begin{equation}\label{725} 
\psi_{n,2}(\ek)\coloneqq-a(n)\,\delta a(n)\,\sin 2\xi,   
\end{equation} 
as in \cite{BZ5}*{(9.7c)}. A direct computation gives the expression
\begin{equation}\label{7250} 
\psi_{n,\R{I}}^{\omega}(\ek)=2a(n)\delta a(n)\sin\omega\,\cos(2\xi-\omega).   
\end{equation} 
For $\nu=1$, $(\omega,t)\in\Omega^*\times\D{R}$ we define $\psi_{n,2}^{\,\omega,t}$ as in \cite{BZ5}*{Section 9.3.2}, i.e.,
\begin{subequations}  \label{724}
\begin{equation}\label{724''} 
\psi_{n,2}^{\,\omega,t}\coloneqq\bigl(\psi_{n,\R{I\!I}}^{\omega}+\psi_{n,\R{I}}^{\omega}\circ\vartheta_n\bigr)\circ\tau_t. 
\end{equation}
For $\nu\geq 2$, $\uw=(\uw',\omega)\in(\Omega^*)^{\nu-1}\times\Omega^*$, $\ut=(\ut',t)\in\D{R}^{\nu-1}\times\D{R}$, we use the same induction formula as in \cite{BZ5}*{(10.12)}, i.e.,  
\begin{equation}\label{724b}
\psi_{n,2}^{\uw,\ut}=(\psi_{n,1}^{-\omega,t}-\psi_{n,1}^{-\omega,t}\circ\vartheta_n^{\omega,t}+\psi_{n,2}^{\uw',\ut'}-\psi_{n,2}^{-\omega,t}+\psi_{n,1}^{\uw',\ut'}\circ\vartheta_n^{\omega,t}-\psi_{n,1}^{\uw',\ut'})\circ\tau_{\omega}, 
\end{equation} 
\end{subequations}
where $\psi_{n,i}^{-\omega,t}\coloneqq\psi_{n,i}^{2\pi-\omega,t}$, $i=1,2$.

\subsubsection{End of the proof of Lemma \ref{lem:61}}\label{sec:726}

We recall that \eqref{72} was proved in Section~\ref{sec:723} and following the remark from Section~\ref{sec:71}, it only remains to prove \eqref{72'}, i.e.\ $\abs{\partial_t\psi_{n,2}^{\uw,\ut}(\nk)}\leq\hat C_0n^{\varepsilon}$ where $t=t_{\nu}$.

If $\nu=1$ then $\psi_{n,2}^{\omega,t}=\psi_{n,2}^{\omega,0}\circ\tau_t$ and \eqref{72'} follows from \eqref{70c} using \eqref{73}. 
  
If $\nu\geq 2$ then $\psi_{n,2}^{\uw,\ut}$ is given by \eqref{724b}. To estimate $\partial_t\psi_{n,2}^{\uw,\ut}$ we first observe that 
\[
\psi_{n,1}^{-\omega,t}-\psi_{n,1}^{-\omega,t}\circ\vartheta_n^{\omega,t}=(\psi_{n,1}^{-\omega,0}-\psi_{n,1}^{-\omega,0}\circ\vartheta_n^{\omega,0})\circ\tau_t  
\] 
and \eqref{73} allows us to estimate
\begin{equation}\label{85c}
\abs{\partial_t(\psi_{n,1}^{-\omega,t}-\psi_{n,1}^{-\omega,t}\circ\vartheta_n^{\omega,t})(\nk)}\leq\norm{\psi_{n,1}^{-\omega,0}\circ\vartheta_n^{\omega,0}-\psi_{n,1}^{-\omega,0}}_{\class^1(\cercle)}. 
\end{equation}  
The right hand side of \eqref{85c} can be estimated by a constant $C$ independent of $n$ by using \eqref{L81c} for $q=\psi_{n,1}^{-\omega,0}$ and \eqref{70a}, i.e., $\norm{\psi_{n,1}^{-\omega,0}}_{\class^1(\cercle)}=\ord(n^{\gamma})$.  
It remains to estimate 
\begin{equation}\label{85d}
\abs{\partial_t(\psi_{n,1}^{\uw',\ut'}\circ\vartheta_n^{\omega,t})(\nk)}\leq\norm{\psi_{n,1}^{\uw',\ut'}}_{\class^1(\cercle)}\,\abs{\partial_t\vartheta_n^{\omega,t}(\nk)}. 
\end{equation}
However \eqref{L81b} and \eqref{70a} allow us to estimate the right hand side of \eqref{85d} by $C\nu\leq Cn^{\varepsilon}$ which completes the proof of \eqref{72'}.\qed

\subsection{Proof of Lemma \ref{lem:22} \textup{(b)}}\label{sec:73}  

This section is devoted to the proof of Lemma \ref{lem:22} under assumptions (H1) and (H2). Let $\psi_{n,\R{I}}^{\omega}$, $\psi_{n,\R{I\!I}}^{\omega}$ be given by \eqref{725''} and \eqref{725'}, respectively. Then 
\begin{equation}\label{L73'} 
\psi_{n,2}^{\omega,0}=\psi_{n,\R{I}}^{\omega}+\psi_{n,\R{I\!I}}^{\omega}+r_n^{\omega} 
\end{equation}
holds with
\begin{equation} \label{L730} 
r_n^{\omega}\coloneqq\psi_{n,\R{I}}^{\omega}\circ\vartheta_n-\psi_{n,\R{I}}^{\omega}.
\end{equation} 

\begin{lemma}\label{lem:82} 
We have the estimates 
\begin{align}\label{L73''} 
&\norm{r_n^{\omega}}_{\class^0(\cercle)}=\ord(n^{-\gamma}),\\\label{L73} 
&\abs{\psi_{n,\R{I\!I}}^{\omega}(\pm\eul^{\ii\omega/2})}=\ord(n^{-\gamma}).
\end{align} 
\end{lemma} 

\begin{proof}  
We observe that \eqref{72q} ensures 
\[
\norm{\psi_{n,\R{I}}^{\omega}\circ\vartheta_n-\psi_{n,\R{I}}^{\omega}}_{\class^0(\cercle)}\leq Cn^{-\gamma}\norm{\psi_{n,\R{I}}^{\omega}}_{\class^1(\cercle)}.
\] 
To complete the proof of \eqref{L73''} it suffices to observe that $\norm{\psi_{n,\R{I}}^{\omega}}_{\class^1(\cercle)}=\ord(1)$. 

To prove \eqref{L73} we write the Taylor expansion at order $2$:
\[
\psi_{n,\R{I\!I}}^{\omega}(\eul^{\ii\omega/2})=\partial_\eta\psi_{n,1}^{\omega,0}(\eul^{\ii\omega/2})\left(\vartheta_n(\eul^{\ii\omega/2})-\eul^{\ii\omega/2}\right)+\tilde r_n^{\omega} 
\]
with
\[
\abs{\tilde r_n^{\omega}}\leq\norm{\psi_{n,1}^{\omega,0}}_{\class^2(\cercle)}\,\left\lvert\vartheta_n(\eul^{\ii\omega/2})-\eul^{\ii\omega/2}\right\rvert^2\leq Cn^{\gamma}\abs{\tilde\xi_n(\eul^{\ii\omega/2})}^2\leq C'n^{-\gamma}. 
\]
A similar reasoning can be applied to $-\eul^{\ii\omega/2}$ instead of $\eul^{\ii\omega/2}$ and to complete the proof of \eqref{L73} we observe that $\partial_\eta\psi_{n,1}(\pm\eul^{\ii\omega/2})=0$. 
\end{proof}

To begin the proof of \eqref{gnn} in case (b) we observe that  
\begin{equation*}\label{82'a} 
g_n(n)=\sum_{m=1}^{\lfloor N/2\rfloor}\alpha_m\Re\scal{\vece_n, H_n^{2\pi m/N,0}\vece_n}+\sum_{m=1}^{\lfloor(N-1)/2\rfloor}\tilde \alpha_m\Im\scal{\vece_n,H_n^{2\pi m/N,0}\vece_n} 
\end{equation*}
and Lemma \ref{lem:61} in the case $\nu=1$ ensures 
\begin{equation*}\label{82'b} 
g_n(n)=\sum_{m=1}^{\lfloor N/2\rfloor}\alpha_m\Re\F{g}_n^{2\pi m/N,0}(n)+\sum_{m=1}^{\lfloor(N-1)/2\rfloor}\Im\F{g}_n^{2\pi m/N,0}(n) +\ord(n^{-\gamma+5\varepsilon})  
\end{equation*}
with   
\begin{equation*}\label{82'c} 
\F{g}_n^{\omega,0}(n)=\eul^{\ii\omega n}\int_0^{2\pi}\eul^{\ii\psi_{n,1}^{\omega,0}(\nk)}b_n^{\omega}(\nk)\,\frac{\dd\eta}{2\pi} 
\end{equation*} 
for $\omega\in\Omega^*$ and  
\begin{equation*}\label{82'd} 
b_n^{\omega}(\nk)=\eul^{\ii\psi_{n,2}^{\omega,0}(\nk)}
\end{equation*}
according to \eqref{71''}. Then Lemma \ref{lem:41} with $\mu=4a(n)\sin\frac{\omega}{2}$ and $\eta_0=\pi+\frac{\omega}{2}$ gives 
\begin{equation}\label{82'e} 
\F{g}_n^{\omega,0}(n)=\eul^{\ii\omega n}\sum_{\kappa=\pm 1}\frac{\eul^{\ii\kappa\left(4a(n)\sin\frac{\omega}{2}-\frac{\pi}{4}\right)}}{2\sqrt{2\pi a(n)\sin\frac{\omega}{2}}}b_n^{\omega}(-\kappa\eul^{\ii\omega/2})+\ord(a(n)^{-1})
\end{equation} 
and applying Lemma \ref{lem:82} we obtain 
\begin{equation}\label{82'f} 
b_n^{\omega}(\pm\eul^{\ii\omega/2})=\eul^{\ii\psi_{n,\R{I}}^{\omega}(\pm\eul^{\ii\omega/2})}+\ord(n^{-\gamma}).  
\end{equation}  
However using \eqref{7250} we get $\psi_{n,\R{I}}^{\omega}(\pm\eul^{\ii\omega/2})=2a(n)\delta a(n)\sin\omega$. Thus combining \eqref{82'f} with \eqref{82'e} and $a(n)^{-1}=\ord(n^{-\gamma})$, we obtain   
\begin{equation}\label{82'g} 
\F{g}_n^{\omega,0}(n)=\frac{\cos\left(4a(n)\sin\frac{\omega}{2}-\frac{\pi}{4}\right)}{\sqrt{2\pi a(n)\sin\frac{\omega}{2}}}\eul^{\ii(\omega n+2a(n)\delta a(n)\sin\omega)}+\ord(n^{-\gamma}).
\end{equation} 
To complete the proof we observe that \eqref{82'g} ensures 
\begin{align*}
\Re\F{g}_n^{\omega,0}(n)&=\frac{\cos\left(4a(n)\sin\frac{\omega}{2}-\frac{\pi}{4}\right)}{\sqrt{2\pi a(n)\sin\frac{\omega}{2}}}\cos\bigl(\omega n+2a(n)\delta a(n)\sin\omega\bigr)+\ord(n^{-\gamma}),\\
\Im\F{g}_n^{\omega,0}(n)&=\frac{\cos\left(4a(n)\sin\frac{\omega}{2}-\frac{\pi}{4}\right)}{\sqrt{2\pi a(n)\sin\frac{\omega}{2}}}\sin\bigl(\omega n+2a(n)\delta a(n)\sin\omega\bigr)+\ord(n^{-\gamma}).
\end{align*}

\appendix
\section{The quantum Rabi model}   \label{appendix}

The quantum Rabi model couples a quantized single-mode radiation and a two-level quantum system. 

Let $\C{H}_{\field}$ be a complex Hilbert space equipped with an orthonormal basis $\accol{\vece_n}_0^\infty$ and let $\hat a$, $\hat a^\dagger$ be the photon annihilation and creation operators defined in $\C{H}_{\field}$ by
\begin{alignat*}{2}
\hat a\,\vece_n&=\sqrt{n}\,\vece_{n-1},&\quad& n=0,\,1,\,2,\,\dots,\\
\hat a^\dagger\vece_n&=\sqrt{n+1}\,\vece_{n+1},&&n=0,\,1,\,2,\, \dots
\end{alignat*}
(with $\vece_{-1}=0$). To define the quantum Rabi model we fix four positive parameters:  
\begin{compactenum}[(i)]
\item
$\omega$, the frequency of the quantized one-mode electromagnetic field, 
\item
$E$, the level separation energy, 
\item
$g$, the coupling constant, 
\item
$\hbar$, the Planck constant.    
\end{compactenum}
The quantum Rabi Hamiltonian is then the self-adjoint operator in $\C{H}_{\field}\otimes\D{C}^2$ given by    
\begin{equation} \label{R}
H_{\rabi}=\hbar\omega\,\hat a^\dagger\hat a\otimes I_{\D{C}^2}+I_{\C{H}_{\field}}\otimes\tfrac{1}{2}E\sigma_z+\hbar g(\hat a^\dagger+\hat a)\otimes\sigma_x, 
\end{equation} 
where $\sigma_x=\left(\begin{smallmatrix}0&1\\1&0\end{smallmatrix}\right)$ and $\sigma_z=\left(\begin{smallmatrix}1&0\\0&-1\end{smallmatrix}\right)$. Then we have the decomposition (see \cite{Tur1} or \cite{BZ6}*{Section 3.1})   
\begin{equation} 
\C{H}_{\field}\otimes\D{C}^2=\C{H}_+\oplus\C{H}_-,  
\end{equation}   
where $\C{H}_+$ and $\C{H}_-$ are invariant under $H_{\rabi}$ and the restrictions $H_{\pm}\coloneqq H_{\rabi}\vert_{\C{H}_{\pm}}$ have the form
\begin{equation} \label{Jpm} 
H_{\pm}=-\frac{1}{2}\hbar\omega+\hbar\omega\,J_{\pm},  
\end{equation} 
where the operator $J_+$ (resp.\ $J_-$) is defined in an appropriate basis by the Jacobi matrix \eqref{J} whose entries are given by \eqref{J'} with $a_1=\frac{g}{\omega}$ and $\rho=\frac{E}{2\hbar\omega}$ (resp.\ $\rho=-\frac{E}{2\hbar\omega}$). Therefore   
\[
\spec(H_{\rabi})=\accol{\lambda_n(H_+)}_{n=1}^\infty\cup\accol{\lambda_n(H_-)}_{n=1}^\infty,
\]
where 
\[
\lambda_n(H_{\pm})=-\frac{1}{2}\hbar\omega+\hbar\omega\lambda_n(J_{\pm}).\tag*{\qed}
\]


\begin{bibdiv}
\begin{biblist}
\bib{BNS}{article}{
   author={Boutet de Monvel, Anne},
   author={Naboko, Serguei},
   author={Silva, Luis O.},
   title={The asymptotic behavior of eigenvalues of a modified Jaynes--Cummings model},
   journal={Asymptot. Anal.},
   volume={47},
   date={2006},
   number={3-4},
   pages={291--315},
} 
\bib{BZ1}{article}{
   author={Boutet de Monvel, Anne},
   author={Zielinski, Lech},
   title={Eigenvalue asymptotics for Jaynes--Cummings type models without modulations}, 
   journal={BiBoS preprint},
   number={08-03-278},
   eprint={http://www.physik.uni-bielefeld.de/bibos/},
   year={2008},
}
\bib{BZ2}{article}{
   author={Boutet de Monvel, Anne},
   author={Zielinski, Lech},
   title={Explicit error estimates for eigenvalues of some unbounded Jacobi matrices},
   conference={
   title={Spectral Theory, Mathematical System Theory, Evolution Equations, Differential and Difference Equations: IWOTA10},},
   book={
      series={Oper. Theory Adv. Appl.},
      volume={221},
      publisher={Birkh\"auser Verlag},
      place={Basel},},
   date={2012},
   pages={187--215},
}
\bib{BZ4}{article}{
   author={Boutet de Monvel, Anne},
   author={Zielinski, Lech},
   title={Asymptotic behavior of large eigenvalues of a modified Jaynes--Cummings model}, 
   conference={
      title={Spectral Theory and Differential Equations},
   },
   book={
      series={Amer. Math. Soc. Transl. Ser. 2},
      volume={233},
      publisher={Amer. Math. Soc., Providence, RI},
   },
   date={2014},
   pages={77--93},
}  
\bib{BZ5}{article}{
   author={Boutet de Monvel, Anne},
   author={Zielinski, Lech},
   title={Asymptotic behavior of large eigenvalues for Jaynes--Cummings type models},
   journal={J. Spectr. Theory},
   volume={7},
   date={2017},
   number={2},
   pages={559--631},
}   
\bib{BZ6}{article}{
   author={Boutet de Monvel, Anne},
   author={Zielinski, Lech},
   title={On the spectrum of the quantum Rabi model},
   status={submitted},
}   
\bib{Bra11}{article}{
   author={Braak, Daniel},
   title={Integrability of the Rabi Model},
   journal={Phys. Rev. Lett.},
   volume={107},
   date={2011},
   number={10},
   pages={100401, 4 pp.},
} 
\bib{Bra16}{article}{
   author={Braak, Daniel},
   author={Chen, Qing-Hu},
   author={Batchelor, Murray T.},
   author={Solano, Enrique},
   title={Semi-classical and quantum Rabi models: 
   in celebration of 80 years},
   journal={J. Phys. A: Math. Theor.},
   volume={49},
   date={2016},
   number={30},
   pages={300301, 4 pp.},
} 
\bib{FKU}{article}{
   author={Feranchuk, I. D.}, 
   author={Komarov, L. I.}, 
   author={Ulyanenkov, A. P.},
   title={Two-level system in a one-mode quantum field: numerical solution on the basis of the operator method},
   journal={J. Phys. A: Math. Gen.},
   volume={29},
   date={1996},
   number={14},
   pages={4035--4047},
}
\bib{FILU}{book}{
   author={Feranchuk, Ilya D.}, 
   author={Ivanov, Alexey}, 
   author={Le, Van-Hoang},
   author={Ulyanenkov, Alexander P.},
   title={Non-perturbative Description of Quantum Systems},
   series={Lecture Notes in Physics},
   volume={894},
   publisher={Springer, Cham},
   date={2015},
   pages={xv+362},
}
\bib{He13}{article}{
   author={He, Shu}, 
   author={Zhang, Yu-Yu}, 
   author={Chen, Qing-Hu}, 
   author={Ren, Xue-Zao}, 
   author={Liu, Tao}, 
   author={Wang, Ke-Lin},
   title={Unified analytical treatments of qubit-oscillator systems},
   journal={Chinese Phys. B},
   volume={22},
   date={2013},
   number={6},
   pages={064205},
} 
\bib{Ir}{article}{
   author={Irish, E. K.},
   title={Generalized Rotating-Wave Approximation for Arbitrarily Large Coupling},
   journal={Phys. Rev. Lett.},
   volume={99},
   date={2007},
   number={17},
   pages={173601},
}
\bib{JN}{article}{
   author={Janas, Jan},
   author={Naboko, Serguei},
   title={Infinite Jacobi matrices with unbounded entries: asymptotics of eigenvalues and the transformation operator approach},
   journal={SIAM J. Math. Anal.},
   volume={36},
   date={2004},
   number={2},
   pages={643--658},
}
\bib{JC}{article}{
   author={Jaynes, E. T.},
   author={Cummings, F. W.},
   title={Comparison of quantum and semiclassical radiation theories with application to the beam maser},
   journal={Proc. IEEE},
   volume={51},
   date={1963},
   number={1},
   pages={89--109},
} 
\bib{Ma}{article}{
   author={Malejki, Maria},
   title={Asymptotics of large eigenvalues for some discrete unbounded Jacobi matrices},
   journal={Linear Algebra Appl.},
   volume={431},
   date={2009},
   number={10},
   pages={1952--1970}, 
}
\bib{Schm}{article}{
   author={Schmutz, M.},
   title={Two-level system coupled to a boson mode: the large $n$ limit},
   journal={J. Phys. A: Math. Gen.},
   volume={19},
   date={1986},
   number={17},
   pages={3565--3577},
}
\bib{Scu}{book}{
    author={Scully, Marlan O.},
    author={Zubairy, M. Suhail},
    title={Quantum Optics},
    publisher={Cambridge University Press, Cambridge},
    date={1997},
} 
\bib{St}{book}{
   author={Stein, Elias M.},
   title={Harmonic analysis: real-variable methods, orthogonality, and
   oscillatory integrals},
   series={Princeton Mathematical Series},
   volume={43},
   note={With the assistance of Timothy S. Murphy;
   Monographs in Harmonic Analysis, III},
   publisher={Princeton University Press, Princeton, NJ},
   date={1993},
   pages={xiv+695},
}
\bib{Tur1}{article}{
   author={Tur, {\`E}. A.},
   title={Jaynes--Cummings model: solution without rotating wave approximation},
   journal={Optics and Spectroscopy},
   volume={89},
   date={2000},
   number={4},
   pages={574--588},
}
\bib{Tur2}{article}{
   author={Tur, E. A.},
   title={Jaynes--Cummings model without rotating wave approximation. Asymptotics of eigenvalues},
   date={2002},
   eprint={arXiv.org/abs/math-ph/0211055},
   pages={12},
}
\bib{X}{article}{
   author={Xie, Qiongtao}, 
   author={Zhong, Honghua}, 
   author={Batchelor, Murray T.}, 
   author={Lee, Chaohong}, 
   title={The quantum Rabi model: solution and dynamics},
   journal={J. Phys. A: Math. Theor.},
   volume={50},
   date={2017},
   number={11},
   pages={113001, 40 pp.},
}
\bib{Y}{article}{
   author={Yanovich, Eduard A.},  
   title={Asymptotics of eigenvalues of an energy operator in a problem of quantum physics},
   conference={ 
      title={Operator Methods in Mathematical Physics},
   },
   book={
      series={Oper. Theory Adv. Appl.},
      volume={227},
      publisher={Birkh\"auser/Springer Basel AG, Basel},
   },
   date={2013},
   pages={165--177},
}
\end{biblist}
\end{bibdiv}
\end{document}